\documentclass[a4paper,allowtoday,onecolumn,unpublished]{quantumarticle} 
\pdfoutput=1 

\usepackage[utf8]{inputenc}  
\usepackage[T1]{fontenc} 	
\usepackage{amsmath, amssymb, amsfonts, amsthm}  
\usepackage{mathtools}  
\usepackage{dsfont}  
\usepackage{physics}  
\usepackage{graphicx}  
\usepackage{tikz}  
\usetikzlibrary{quantikz2}  
\usepackage{pifont} 
\usepackage{caption} 
\usepackage{booktabs}
\renewcommand{\arraystretch}{1.4} 
\usepackage{microtype}  
\usepackage{parskip}  
\usepackage{xcolor}  
\usepackage[dvipsnames,x11names]{xcolor}  
\usepackage[colorlinks=true, citecolor=OliveGreen, linkcolor=Purple, urlcolor=Magenta]{hyperref}  
\usepackage{comment}  

\usepackage{thmtools}
\usepackage{thm-restate}
\usepackage{subcaption}

\usepackage[numbers,sort&compress]{natbib} 	

\newtheorem{theorem}{Theorem}
\newtheorem{lemma}[theorem]{Lemma}

\makeatletter 
\newcommand{\doublewidetilde}[1]{{%
  \mathpalette\double@widetilde{#1}%
}}
\newcommand{\double@widetilde}[2]{%
  \sbox\z@{$\m@th#1\map{#2}$}%
  \ht\z@=.9\ht\z@
  \map{\box\z@}%
}

\newcommand{\dket}[1]{\vert#1\rangle\hspace{-.8mm}\rangle}
\newcommand{\dbra}[1]{\langle\hspace{-.8mm}\langle #1\vert}
\newcommand{\dketbra}[2]{\vert #1 \rangle\!\rangle \langle\!\langle #2 \vert}

\newcommand{\map}[1]{\widetilde{#1}}
\newcommand{\smap}[1]{\doublewidetilde{{#1}_{}}}

\renewcommand{\mathbb}{\mathds}  
\newcommand{\id}{\mathbb{1}}  
\newcommand{\ChoiU}{\dketbra{U}{U}} 
\definecolor{LavenderBlue}{RGB}{95, 0, 255} 

\renewcommand{\H}{\mathcal{H}}
\renewcommand{\L}{\mathcal{L}}

\renewcommand{\P}{\textup{P}}
\newcommand{\F}{\textup{F}}
\newcommand{\M}{\textup{M}}
\newcommand{\I}{\textup{I}}
\renewcommand{\O}{\textup{O}}
\newcommand{\A}{\textup{A}}

\title{Higher-order quantum computing with known input states}
\author{Vanessa Brzi\'c}
\affiliation{Sorbonne Universit\'{e}, CNRS, LIP6, F-75005 Paris, France}
\author{Satoshi Yoshida}
\orcid{0000-0002-0521-5209}
\affiliation{Department of Physics, Graduate School of Science, The University of Tokyo, 7-3-1 Hongo, Bunkyo-ku, Tokyo 113-0033, Japan}
\author{Mio Murao}
\orcid{0000-0001-7861-1774}
\affiliation{Department of Physics, Graduate School of Science, The University of Tokyo, 7-3-1 Hongo, Bunkyo-ku, Tokyo 113-0033, Japan}
\author{Marco Túlio Quintino}
\orcid{0000-0003-1332-3477}
\affiliation{Sorbonne Universit\'{e}, CNRS, LIP6, F-75005 Paris, France}

\date{}

\begin{document}

\maketitle

\abstract{
In higher-order quantum computing (HOQC), one typically considers the universal transformation of unknown quantum operations, treated as blackboxes. It is also implicitly assumed that the resulting operation must act on arbitrary, and thus unknown, input states. In this work, we explore a variant of this framework in which the operation remains unknown, but the input state is fixed and known. We argue that this assumption is well-motivated in certain practical contexts, such as unitary programming, and show that classical knowledge of the input state can significantly enhance performance. We demonstrate that in the SAR protocol, this knowledge leads to an exponential advantage through a repeat-until-success strategy, highlighting the operational power of known-state higher-order transformations. Moreover, this assumption allows us to distinguish between protocols designed for pure, bipartite, and mixed states, which enables us to identify the class of mixed states for which deterministic and exact implementation becomes possible.}

\tableofcontents

\section{Introduction}

In the standard approach to quantum information, quantum states are treated as objects that evolve under quantum operations, such as channels or instruments—often referred to as quantum gates. These operations act on quantum states to produce other quantum states. In higher-order quantum computing (HOQC), by contrast, the operations themselves are treated as objects that can be transformed. Such transformations are carried out by what are known as higher-order quantum operations, which take quantum operations as inputs and return quantum operations as outputs. \cite{taranto2025review,Chiribella_2008,Chiribella_2009}. 
For example, in the case of unitary operations, a standard HOQC task is to transform an arbitrary unitary operation  \( U \) into another unitary operation \( f(U) \), where \( f:\mathrm{SU}(d)\to \mathrm{SU}(d)\) is a function mapping unitary operators to unitary operators. The HOQC framework can be used, for instance, to design quantum circuits that invert an arbitrary unitary operation $f(U)=U^{-1}$, transpose it, $f(U)=U^T$, or take its complex conjugate, $f(U)=U^*$ ~\cite{Quintino2019prob,Quintino_2022det,Ebler_2023}. Another relevant task that can be naturally analysed within the HOQO framework is unitary storage-and-retrieval (SAR) \cite{Sedl_k_2019}, also known as unitary learning \cite{Bisio_2010}, and closely related to unitary programming~\cite{Nielsen_1997}, which
corresponds to the delayed input state variant of $f(U) = U$.

In this work, we consider higher-order quantum operations in a setting where the output operation is applied to a quantum state known to the user. This classical knowledge allows the HOQO to be specifically tailored to the given input state (see Fig.~\ref{fig:comparison}), leading to significant performance gains in certain tasks. Interestingly, we also identify tasks, such as probabilistic exact conjugation, for which knowledge of the input state offers no advantage.

\begin{figure}[t!]
  \centering
  \begin{subfigure}[t!]{0.45\textwidth}
    \centering
    \includegraphics[width=\linewidth]{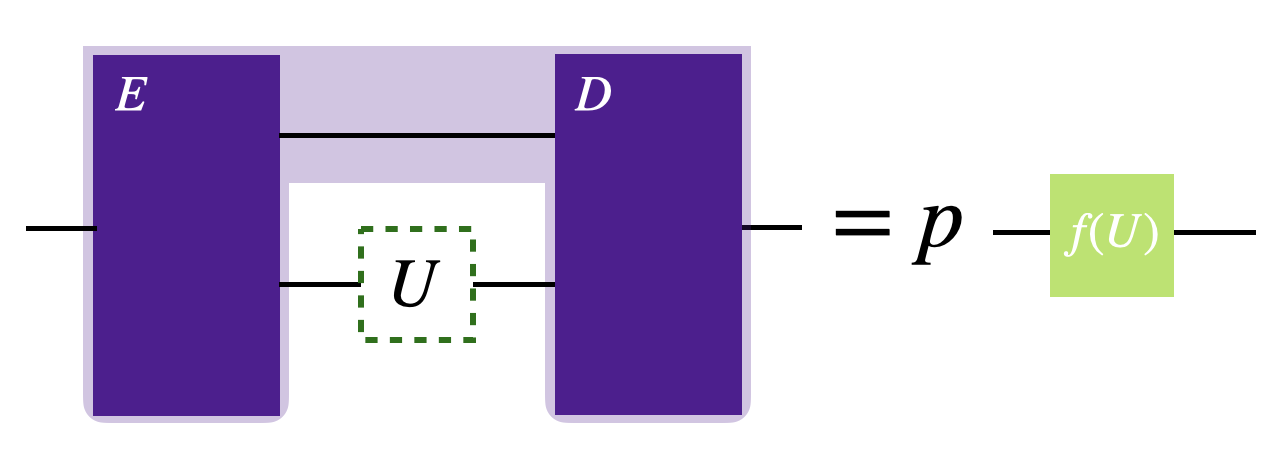}
    \captionsetup{width=0.9\textwidth}
    \caption{Standard HOQC case where the input operation and the input state are unknown. Since the input state is unknown, this scenario does not differentiate between mixed, pure, or bipartite inputs.}
  \end{subfigure}
  \quad \quad 
  \begin{subfigure}[t!]{0.45\textwidth}
    \centering
    \includegraphics[width=\linewidth]{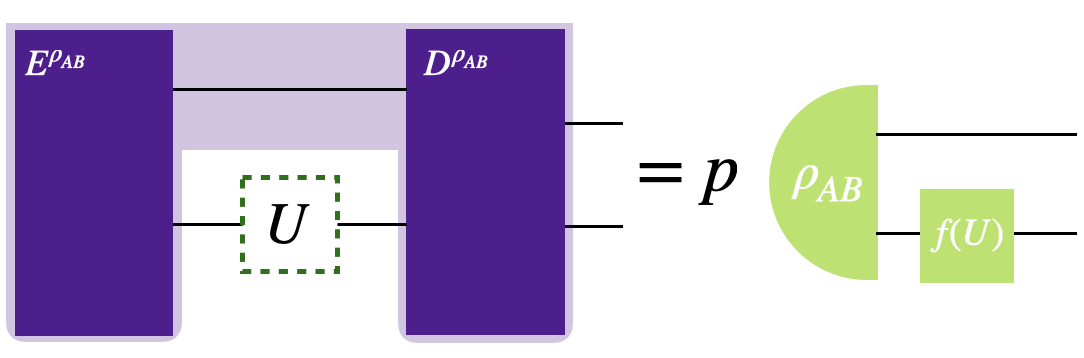}
    \captionsetup{width=0.9\textwidth}
    \caption{The scenario considered in our work assumes classical knowledge of the input state, enabling the design of the encoder and decoder to implement a transformation of interest on the known input state.}
  \end{subfigure}
\captionsetup{width=0.9\textwidth}
  \caption{A protocol implementing a unitary transformation via the action of a supermap. Elements shown in purple are optimized; dashed elements represent unknown and universal components; green indicates fixed behaviour determined by the other elements.}
  \label{fig:comparison}
\end{figure}

The assumption of a known input state is particularly well motivated in the operational setting of unitary SAR.

Namely, in the SAR task, one has access to $k$ uses of a unitary operation $U$, with the goal of storing its action into a quantum state so that it can later be retrieved (see Fig.~\ref{Fig_delay}). One may imagine, for instance, that today we have access to a powerful quantum computer implementing $U$, but not yet the input state on which we wish to perform the computation. In such a case, one can apply the $k$ calls of $U$ on one half of an entangled state $\ket{\phi_{AB}}$, which can be stored until the desired input state becomes available. At a later time, when the state $\ket{\psi}$ is ready, we perform a global operation, called the retrieving instrument, to reconstruct $U\ket{\psi}$. Previous works~\cite{Bisio_2010,Sedl_k_2019} studied the case where the retrieval operation is independent of the input state $\ket{\psi}$. However, in some scenarios it is natural to assume classical knowledge of the state $\ket{\psi}$. In this case, the retrieval operation can depend on $\ket{\psi}$ and be optimised accordingly. As we show in this work, such classical knowledge fundamentally changes the problem, yielding an exponential improvement in the performance of unitary SAR with respect to the number of calls $k$. The same also applies to unitary programming \cite{Nielsen_1997}, which becomes equivalent to SAR when the program state is itself prepared by a quantum encoder.\footnote{In unitary programming the program state is assumed to be given a priori, and in principle it may originate from any more general mechanism. In contrast, in SAR the program state is not arbitrary nor externally supplied: it is explicitly constructed by the protocol itself, via the encoder followed by the uses of the unknown unitary.}

Since the input state is now assumed to be known, the structure of the HOQO protocol depends on the nature of the state, whether it is pure known state implementation $U \mapsto f(U)\ketbra{\psi}{\psi} f(U)^{\dagger}$, mixed $U \mapsto f(U)\rho f(U)^{\dagger}$, or acting on part of a bipartite system $U \mapsto (\mathbb{1}_{A}  \otimes f(U)_{B})\ketbra{\psi}{\psi}_{AB} (\mathbb{1}_{A} \otimes f(U)^{\dagger}_{B})$, with each case requiring a different characterization and leading to different optimal strategies. For bipartite protocols, we focus on the probabilistic case, which includes, as a special case, the single-subsystem pure-state $f(U)\ketbra{\psi}{\psi} f(U)^{\dagger}$ probabilistic implementation. Our results demonstrate that, the success probability increases with the amount of entanglement. For mixed states, we prove that a deterministic and exact implementation is possible using only a single call to the input unitary, for a certain class of states of the form $\rho = \eta \ketbra{\psi}{\psi} + \frac{1-\eta}{d}\mathbb{1}$, with the particular values of visibility $\eta$, as discussed in the Thm.~\ref{thm:mix_eta}.

\begin{figure}[t!]
  \centering
  \begin{subfigure}[t]{0.45\textwidth}
    \includegraphics[width=\linewidth]{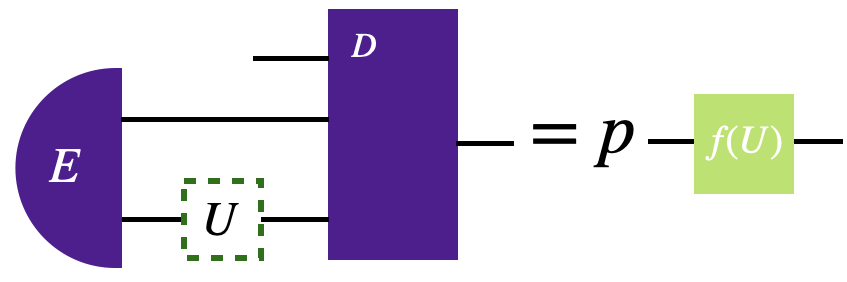}
    \captionsetup{width=0.9\textwidth}
    \caption{The standard HOQC case where the input operation and the input state are unknown and where the input state is delayed. }
  \end{subfigure}
    \quad \quad 
  \begin{subfigure}[t]{0.5\textwidth}
    \includegraphics[width=\linewidth]{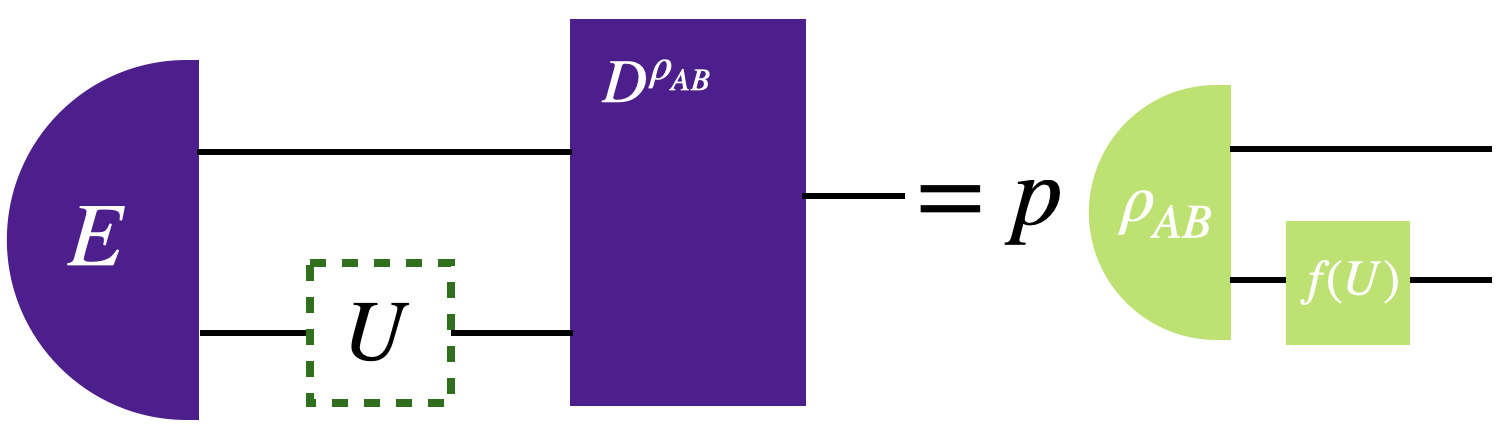}
    \captionsetup{width=0.9\textwidth}
    \caption{The scenario considered in our work assumes that the known input state is delayed, arriving after the storage phase. This enables the decoder to be designed using classical knowledge of the state, in general enhancing the probability of successful retrieval.}
    \label{fig:fig2}
  \end{subfigure}
  \captionsetup{width=0.9\textwidth}
  \caption{A protocol implementing unitary transformation via action of a supermap with delayed input state. }
  \label{Fig_delay}
\end{figure}

\section{Summary of main results}

In this section, we state the main theorems and present a summary of the main results. 

\subsection{HOQC with a single call of the input operation and known bipartite input state}

Let us consider a class of protocols that transforms an arbitrary $d$-dimensional unitary operations described by a unitary operator $U\in\mathrm{SU}(d)$ into the state $(\mathbb{1}_{\A}  \otimes f(U)_{B}  )\ket{\psi}_{AB}$,
where $f:\mathrm{SU}(d) \to \mathrm{SU}(d)$ is a function that transforms unitary operators into unitary operators and $\ket{\psi}_{AB}\in\mathcal{H}_A \otimes \mathcal{H}_B$ is a fixed and known bipartite state with $\mathcal{H}_{A}\cong \mathbb{C}^d$, and $\mathcal{H}_B$ is an arbitrary finite dimensional complex linear space. For a general case of mixed bipartite states, this protocol is depicted in Fig.~\ref{fig:comparison}. We present the optimal probabilistic protocols for a single use of the input unitary, given a pure bipartite input state.

\begin{restatable}{theorem}{BipProb}\label{thm:probab_gen_bip}
Given a single use of a $d$-dimensional unitary $U \in \mathrm{SU}(d)$, 
the optimal probabilistic exact protocol that transforms $U$ into 
$(\mathbb{1}_{\A} \otimes f(U)_{B})\ket{\psi}_{AB}$, for a known bipartite 
pure state $\ket{\psi}_{AB} \in \mathcal{H}_A \otimes \mathcal{H}_B$ and for all $U$, has the following maximal 
success probabilities:
\begin{itemize}
    \item for the unitary transposition, $f(U) = U^{T}$,
    \begin{align}
    p^{\psi_{AB}}_\textup{max,trans}=\frac{1}{d\norm{\Tr_{A}(\ketbra{\psi}{\psi}_{AB} )}_{\textup{op}}}
    \end{align} 
    \item for the unitary conjugation, $f(U) = \overline{U}$,
\begin{equation}
 p^{\psi_{AB}}_{\textup{max,conj}}  =
\begin{cases}
       1,& \textup{if } d = 2\\
    0,             &  \textup{if } d > 2
\end{cases}
\end{equation}
\item for the unitary inversion, $f(U) = U^{-1}$,
\begin{equation}
 p^{\psi_{AB}}_{\textup{max,inv}}  =
\begin{cases}
        \frac{1}{2\norm{\Tr_{A}(\ketbra{\psi}_{AB} )}_{\textup{op}}},& \textup{if } d = 2\\
    0,             &  \textup{if } d > 2
\end{cases} 
\end{equation}
\end{itemize}
where $\norm{\cdot}_\textup{op}$ denotes the operator norm~\cite{Nielsen_Chuang_2010}, corresponding to the largest eigenvalue of the operator. $\mathbb{1}$ denotes the identity operator on the corresponding Hilbert space.
\end{restatable}

Note that, for unitary transposition and qubit inversion (and also for the unitary storage and retrieval problem, discussed later), we see that greater entanglement with an external reference 
leads to a higher attainable success probability. In the limiting cases of no entanglement of the input state, we obtain the result for a supermap acting on a single pure unentangled state, we have $\psi_{AB}=\ketbra{\psi}{\psi}_A \otimes \ketbra{\phi}{\phi}_B$, and hence
$
p^{\psi_{AB}} = \frac{1}{d}
$, where $\ket{\phi}$ is an arbitrary pure state. On the other hand, when $\psi_{AB}=\ketbra{\phi^{+}}{\phi^{+}}_{AB}$ is the maximally entangled state, we have $p^{\psi_{AB}} = 1$. 

\subsection{HOQC with a single call of the input operation and known mixed input state}

The results of the previous section, in the limiting case of no entanglement of the input state, reduce to the scenario in which a unitary operator \( U \in \mathrm{SU}(d) \) is mapped to the pure state \( f(U)\ket{\psi} \). In this section, we first present an analogous result for the deterministic approximate realisation. We then use this result to construct a protocol that deterministically and exactly implements the transformation \( f(U)\rho f(U)^{\dagger} \), given a single input unitary \( U \in \mathrm{SU}(d) \).

\begin{restatable}{theorem}{ThmPureDet}\label{thm:single_det}
Given a single use of a $d$-dimensional unitary $U \in \mathrm{SU}(d)$, 
the optimal probabilistic exact protocol that transforms $U$ into 
$f(U)\ket{\psi}$, for a known pure state $\ket{\psi}\in\mathbb{C}^d$ and for all $U$ $\in \mathrm{SU}(d)$, 
has the following maximal success probabilities:
\begin{itemize}
    \item the transposition operation $U^{T}$ can be implemented, attaining optimal average fidelity 
    \begin{equation}
    \expval{F^{\psi}}_{\textup{trans}} = \frac{2d + 1}{d^2 + d}
 \end{equation}
\item the conjugation operation $\overline{U}$  can be implemented, attaining optimal average fidelity 
    \begin{equation}
 \expval{F^{\psi}}_{\textup{conj}}  =
\begin{cases}
       1,& \textup{if } d = 2\\
    \frac{2}{d + 1},              &  \textup{if } d >  2 
\end{cases} 
\end{equation}
\item the inversion operation $U^{-1}$ can be implemented, attaining optimal average fidelity 
\begin{equation}
 \expval{F^{\psi}}_{\textup{inv}}  =
\frac{d+3}{d(d+1)} 
\end{equation}
\end{itemize}
\end{restatable}

The results presented here exhibit a quadratic improvement compared to the case of an unknown input state~\cite{ChiribellaEbler_2016,Quintino_2022det}. More specifically, when the input state is unknown, we have that: $\expval{F}_{\textup{trans}} = \frac{1}{d^2}$, $\expval{F}_{\textup{conj}} = \frac{2}{d(d-1)}$ and $\expval{F}_{\textup{inv}} = \frac{1}{d^2}$. Also, we remark that, as detailed in Sec.~\ref{sec:tasks_applications}, without loss in performance, we can always restrict our analysis to covariant protocols, hence, for the tasks presented in Thm.~\ref{thm:single_det}, the optimal average fidelity coincides with the worst-case fidelity. 

As we will discuss in detail in Sec.~\ref{sec:mixed}, we make use of the optimal average fidelity for the case of the pure input state presented above to prove the following theorem, 

\begin{theorem}\label{thm:mix_eta}
Let $\rho$ be a known state of the form $\rho = \eta \ketbra{\psi}{\psi} + \frac{1-\eta}{d}\mathbb{1}$, where $\mathbb{1}$ denotes the identity operator on the corresponding Hilbert space. 

There exists a deterministic and exact transformation from single copy of any $U\in\mathrm{SU}(d)$ to $f(U)\rho f(U)^{\dagger}$ if and only if $\eta \leq \frac{d\expval{F^{\psi}} - 1}{d-1}$, where $\expval{F^{\psi}}$ is optimal average fidelity of the known pure state protocol. Then the largest value of visibilities for which the perfect implementation is still possible are:
 \begin{itemize}
        \item transposition $\eta^{*} = \frac{d}{d^2 - 1}$
        \item conjugation $\eta^{*} = \begin{cases}
       1,& \textup{if } d = 2\\
    \frac{1}{d + 1}              &  \textup{if } d >  2 
\end{cases} $
        \item inversion $\eta^{*} = \frac{2}{d^2 - 1}$. 
    \end{itemize}
\end{theorem}

\begin{figure}[t!]
    \centering
\includegraphics[width=0.68\textwidth]{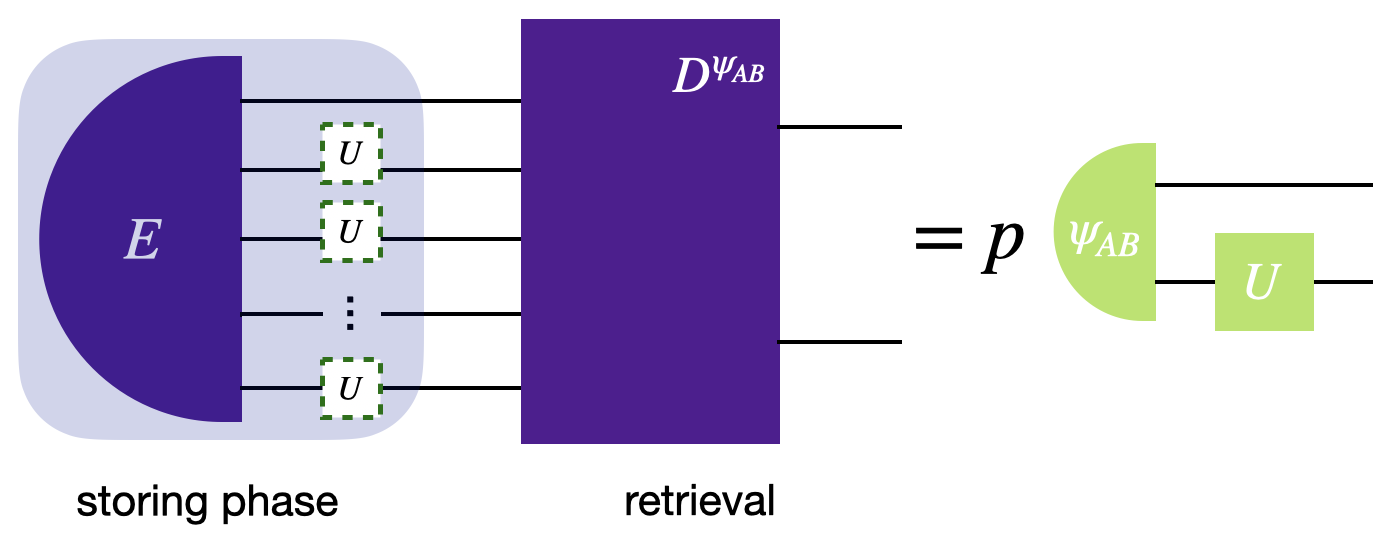}
    \caption{ Storage-and-retrieval protocol using $k$ parallel queries to the unknown unitary $U$, and known input state $\rho$. Note that during the storage phase, the input state $\rho$ is not known, and hence, the encoder operator $E$ cannot depend on $\rho$ (otherwise, the problem would be trivial, as we could simply take $E^\rho = \rho$). However, differently from the standard SAR protocol~\cite{Sedl_k_2019,Bisio_2010}, in the retrieval step, the information of the input state $\rho$ is available, hence the decoder operation $D_\rho$ has an explicit dependence on $\rho$. }
    \label{fig:SAR_k}
\end{figure}

\subsection{Exponential advantage in known state SAR protocol with multiple calls }

The storage-and-retrieval protocol (SAR)~\cite{Sedl_k_2019}, also known as unitary learning~\cite{Bisio_2010}, protocol consists of two stages, as depicted in Fig.~\ref{fig:SAR_k}, and falls within the class of delayed input state protocols (see Fig.~\ref{Fig_delay}). In the \emph{storage phase}, while \(U\) is accessible, its action is encoded into a quantum memory state, typically entangled with an ancillary system. In the \emph{retrieval phase}, once the target state \(\rho\) becomes available and access to unitary \(U\) is no longer possible, a quantum channel \(\map{R}\) is applied to the stored memory and \(\rho\), with the goal of achieving \(U\rho U^\dagger\). The SAR protocol is closely related to the problem of programming quantum gates~\cite{Nielsen_1997}, and can be viewed as a special case in which the program state is prepared by sending part of an entangled state to the operation one desired to store. Consequently, by the no-programming theorem, it is not possible to construct deterministic and exact SAR transformation that works for $\forall U \in \mathrm{SU}(d)$. Hence, the goal is to approximate, or to obtain the exact state $U\rho U^\dagger$ with some probability $p$.

The assumption of a known input state is particularly natural in the SAR setting. In practice, the user typically knows the state on which they wish to retrieve the unitary at the time of retrieval. The storage phase can therefore be viewed as fully universal and independent of the input, whereas the retrieving operation may exploit classical information that is available in the retrieval phase. Operationally, one may think of the storage part of the protocol as prepared on the sender’s side and inaccessible to the user, while the retrieval part is tailored to the specific known input state of interest to the user. This scenario is reminiscent of delegated quantum computation, where a user relies on an external provider to prepare resources while retaining control over the final computation applied to their chosen input. In both cases, part of the process is carried out by an external party, while the user only interacts at the ``end'' on a specific known input.

The distinction between the standard SAR protocol and its known state variant closely parallels the difference between quantum teleportation, which operates on an arbitrary, unknown state, and remote state preparation, where the state to be transmitted is known to the sender. As discussed in the manuscript, we show that SAR protocol in the known state setting, as in the unknown state case, is equivalent to implementing the transposition operation \cite{Quintino2019prob}. 

The results for SAR apply not only to unitaries but also to general completely positive trace preserving (CPTP) maps, i.e., quantum channels, as we will explain further in Sec.~\ref{sec:SAR}. Accordingly, we state the following theorem.

\begin{restatable}{theorem}{Single_U_SAR}\label{thm:SAR_all}
Given a single call of an arbitrary quantum channel 
$\map{C}: \mathcal{L}(\mathbb{C}^{d_{\I}}) \to \mathcal{L}(\mathbb{C}^{d_{\O}})$, the probabilistic exact storage-and-retrieval (SAR) protocol with known input state can be implemented with success probability
    \begin{itemize}
        \item probabilistic exact SAR protocol retrieving on a known pure input state $\ket{\psi}\in \mathbb{C}^{d_{\I}}$ can be implemented with a maximal success probability,
\begin{equation}
    p^{\psi}_{\textup{max,SAR}} = \frac{1}{d} 
\end{equation}
while deterministic approximate SAR protocol retrieving on a known pure input state $\ket{\psi}\in \mathbb{C}^{d_{\I}}$ can be implemented with optimal average fidelity  \begin{equation}
    \expval{F^{\psi}}_{\textup{SAR}} = \frac{2d + 1}{d^2 + d}.
\end{equation}

\item Probabilistic SAR on a part of a known bipartite pure state $\ket{\psi}_{AB}\in \mathbb{C}^{d_{\I}}\otimes \mathbb{C}^{d_{\A}}$, can be implemented with the maximal success probability,
    \begin{align}
    p^{\psi_{AB}}_{\textup{max,SAR}}=\frac{1}{d\norm{\Tr_{A}(\ketbra{\psi}{\psi}_{AB} )}_{\textup{op}}}
    \end{align}. 
\item For a known mixed state of the form $\rho = \eta \ketbra{\psi}{\psi} + \frac{1-\eta}{d}\mathbb{1}\in \mathcal{L}(\mathbb{C}^{d_{\I}})$, the SAR task can be implemented deterministically and exactly, if and only if
\begin{equation}
    \eta\leq \frac{d}{d^2 - 1}.
\end{equation}
    \end{itemize}
\end{restatable}

From this, we observe that pure input state SAR offers a quadratic improvement, both in the deterministic and probabilistic settings, compared to the single-query SAR with an unknown input state. More specifically, when the input state is unknown, we have that $p_{\textup{max,SAR}} = \frac{1}{d^2} $ and $\expval{F^{\psi}}_{\textup{SAR}}=\frac{2}{d^2}$.

In contrast, the performance in the bipartite and mixed state cases depends on particular
properties of the states: for bipartite inputs, it scales with the amount of entanglement, while for mixed inputs, it is dependent on the level of noise. These results follow from the corresponding results for the supermap realising transposition transformation, as we will discuss in details in Sec.~\ref{sec:SAR}.

In most general setting, we can also consider the case where in the storing phase, we have access to $k$ calls of $d-$dimensional unitary, as depicted in Fig.~\ref{fig:SAR_k}. In this case, we construct a repeat-until-success strategy~\cite{Dong_2021} that not only achieves an exponential improvement in success probability, but also reduces the overall resource requirements, compared to the unknown-input state case, which can be summarized in the following theorem:
\begin{restatable}{theorem}{SARk}\label{thm:RUS}
   Given $k$ uses of an arbitrary quantum channel $\map{C}$, that is, a completely positive trace-preserving (CPTP) linear map $\map{C}: \mathcal{L}(\mathbb{C}^{d_{\I}}) \rightarrow \mathcal{L}(\mathbb{C}^{d_{\O}})$, a probabilistic exact storage-and-retrieval (SAR) protocol with a known pure input state $\ket{\psi}\in\mathbb{C}^d$ can be implemented with success probability 
\begin{equation}\label{p_SAR_k*}
    p^{\psi}_{\textup{SAR}} \geq 1 - \left(1-\frac{1}{d}\right)^{k} .
\end{equation}
Furthermore, an SAR protocol that retrieves the action of an arbitrary channel $\map{C}: \mathcal{L}(\mathbb{C}^{d_{\I}}) \rightarrow \mathcal{L}(\mathbb{C}^{d_{\O}})$ on part of a known bipartite state $(\map{C} \otimes \map{\mathbb{1}})\ketbra{\psi}{\psi}_{AB}$, using $k$ calls of $\map{C}$ in the storing phase, achieves a success probability of
\begin{equation}\label{p_SAR_kintr}
    p^{\psi_{AB}}_{\textup{SAR}} \geq 1 - \left(1-\frac{1}{d\norm{\Tr_{A}(\ketbra{\psi}{\psi}_{AB} )}_{\textup{op}}}\right)^{k}.
\end{equation}
\end{restatable}

From Thm.~\ref{thm:RUS}, we see that the probability of success approaches one exponentially on the number of calls $k$. We remark that the optimal performance of unitary SAR for the unknown state case is $p_{\textup{SAR}}=1-\frac{d^2 -1}{k+d^2-1}$ \cite{Sedl_k_2019}, hence we have an exponential improvement. One implementation that attains the probability in Eq.~\eqref{p_SAR_kintr} is repeat-until-success strategy which we will discuss further in Sec.~\ref{sec:RUS}. 

As discussed in Sec.~\ref{sec:numerics_intro} and Sec.~\ref{sec:numerics_main}, numerical results indicate that the repeat-until-success strategy described in Thm.~\ref{thm:RUS} is 
optimal for probabilistic SAR when the known input state is pure. Also, in Sec.~\ref{sec:Garazi} we compare the SAR task with a related, but not equivalent, recent results on port-based remote state preparation \cite{Muguruza_2024}.

\subsection{Implementation of homomorphisms and antihomomorphisms of the unitary group \texorpdfstring{$\mathrm{SU}(d)$}{SU(d)} with multiple calls}

In this section, we present results for protocols that employ \( k \) calls to an unknown unitary \( U \), with the goal of implementing a transformation \( f(U) \) on a known pure input state \( \ket{\psi} \). Specifically, for a given function \( f : \mathrm{SU}(d) \to \mathrm{SU}(d) \) and a fixed state \( \rho \in \mathcal{L}(\mathbb{C}^d) \), we aim to transform an arbitrary unitary operation $\map{U}(\cdot) = U (\cdot) U^\dagger$ into the state \( f(U)\rho f(U)^\dagger \) using \( k \) calls to \( U \).

When multiple queries to the unknown operation are allowed, we distinguish three classes of protocols based on the ordering of the queries: \emph{parallel} protocols, where all calls to \( U \) occur simultaneously, \emph{sequential} protocols, where the calls occur one after another,\footnote{With additional processing operations inserted between queries.} and \emph{general} protocols, which are consistent with quantum theory but go beyond standard causally ordered circuit structures.

As detailed and proven in Sec.~\ref{sec:SAR}, for any number of calls \( k \), the task of parallel unitary transposition is equivalent to unitary SAR with \( k \) calls. A protocol for parallel unitary transposition can be converted into a protocol for unitary SAR with the same success probability, and \textit{vice versa}. It then follows from Thm.~\ref{thm:RUS} that parallel unitary transposition can be achieved with success probability \( p_{\textup{trans}} \geq 1 - \left(1 - \frac{1}{d} \right)^k \). As in the case of unitary SAR, we have numerical evidence that this upper bound is tight when the known input state is pure.

For the $k>1$ unitary inversion and conjugation protocols, we have the following theorem,
\begin{restatable}{theorem}{kProbInvConjg}\label{thm:k_prob_inv}
Given $k<d-1$ calls of the unknown unitary $U \in \mathrm{SU}(d)$, any probabilistic exact quantum protocol realising $\overline{U}\ketbra{\psi}{\psi}U^{T}$ unitary conjugation or $U^{-1}\ketbra{\psi}{\psi}U$ unitary inversion on arbitrary pure state $\ket{\psi}\in\mathbb{C}^d$ necessarily has a success probability of $p=0$.

This no-go result holds for any type of protocols, parallel, sequential, and general protocols. That is, when $k<d-1$, even processes without a definite causal order cannot convert a $d$-dimensional unitary into $\overline{U}\ketbra{\psi}{\psi}\overline{U}^\dagger$  or $U^{-1}\ketbra{\psi}{\psi}U$.  
\end{restatable}

We remark that an analogous result, proving null success probability for $k < d-1$, holds for the conjugation protocol with an unknown input state~\cite{Quintino2019prob}. However, the result presented here, discussed in detail in Sec.~\ref{sec:application_k>1} and proven in App.~\ref{proofsk}, is a strictly stronger result: even when the input state is known, if \(k < d-1\), any probabilistic protocol still has zero success probability. Additionally, in the case of a parallel protocol using \(k = d-1\) queries, it is known that, even for an unknown input state, one can deterministically implement the complex conjugation of an arbitrary unitary \(U\) with success probability \(p = 1\)~\cite{Miyazaki_2017}. Thus, interestingly, this implies that, for pure input state conjugation, having knowledge of the input state offers no advantage. Our theorem thus generalises the result established for unknown input states and proves a conjugation no-go theorem of conjugation under weaker assumptions.

\subsection{Numerical results for the \texorpdfstring{$k>1$}{k>1} calls of the unitary operation scenario}\label{sec:numerics_intro}

As shown in the main text, the optimal performance of unitary transposition  inversion, and complex conjugation can be formulated as a semidefinite program (SDP) and solved numerically. As proven in Thm.~\ref{thm:SAR_transp_Eq}, the parallel transposition protocol is equivalent to SAR. Therefore, all results established for transposition apply equally to SAR. The results of these optimizations are reported in Table~\ref{tab:transposition_intro} and Table~\ref{tab:inversion_intro}, with further details on the numerical methods provided in Sec.~\ref{sec:numerics_main}.
\begin{table}[!ht]
\centering
\renewcommand{\arraystretch}{1.4}
\setlength{\arrayrulewidth}{0.3mm}
\begin{subtable}[t]{0.46\textwidth}
\centering
\begin{tabular}{c|c|c|c}
\multicolumn{4}{c}{\( \boldsymbol{p^{\textup{max}}_{\textup{trans}}(k = 2, d)}\)} \\
\hline
\(d\) & \textbf{Parallel} & \textbf{Adaptive} & \textbf{General} \\
\hline
2 & 0.7500 & 0.8334  & 0.8334 \\
3 & 0.5556 & 0.6112 & 0.6112\\
\specialrule{0.8pt}{0.8pt}{0.8pt}
\multicolumn{4}{c}{\(\boldsymbol{p^{\textup{max}}_{\textup{trans}}(k = 3, d)}\)} \\
\hline
\(d\) & \textbf{Parallel} & \textbf{Adaptive} & \textbf{General} \\
\hline
2 &0.8750 & 1& 1 \\
\hline
\end{tabular}
\subcaption{Optimal success probabilities for the transposition protocol}
\label{tab:numerics_prob_intro}
\end{subtable}
\hspace{2em}
\begin{subtable}[t]{0.46\textwidth}
\centering
\begin{tabular}{c|c|c|c}
\multicolumn{4}{c}{\(\boldsymbol{\expval{F^{\textup{max}}_{\textup{trans}}}(k = 2, d)}\)} \\
\hline
\(d\) & \textbf{Parallel} & \textbf{Adaptive} & \textbf{General} \\
\hline
2 & 0.9396 & 0.9671 & 0.9671 \\
3 & 0.7530 & 0.8190 & 0.8190 \\
\specialrule{0.8pt}{0.8pt}{0.8pt}
\multicolumn{4}{c}{\(\boldsymbol{\expval{ F^{\textup{max}}_{\textup{trans}}}(k = 3, d)}\)} \\
\hline
\(d\) & \textbf{Parallel} & \textbf{Adaptive} & \textbf{General} \\
\hline
2 & 0.9769 & 1& 1 \\
\hline
\end{tabular}
\subcaption{Optimal average fidelities for the transposition protocol}
\label{tab:numerics_det_intro}
\end{subtable}
\caption{Numerical results for the transposition protocol with $k \in \{2,3\}$ calls and $d \in \{2,3\}$}
\label{tab:transposition_intro}
\end{table}


\begin{table}[!ht]
\centering
\renewcommand{\arraystretch}{1.4}
\setlength{\arrayrulewidth}{0.3mm}

\begin{subtable}[t]{0.46\textwidth}
\centering
\begin{tabular}{c|c|c|c}
\multicolumn{4}{c}{\( \boldsymbol{p^{\textup{max}}_{\textup{inv}}(k = 2, d)}\)} \\
\hline
\(d\) & \textbf{Parallel} & \textbf{Adaptive} & \textbf{General} \\
\hline
2 & 0.7500 & 0.8334  & 0.8334 \\
3 & 0.3333 & 0.3333 & 0.3333 \\
\specialrule{0.8pt}{0.8pt}{0.8pt}
\multicolumn{4}{c}{\(\boldsymbol{p^{\textup{max}}_{\textup{inv}}(k = 3, d)}\)} \\
\hline
\(d\) & \textbf{Parallel} & \textbf{Adaptive} & \textbf{General} \\
\hline
2 & 0.8750 & 1 & 1\\
\hline
\end{tabular}
\subcaption{Optimal success probabilities for the inversion protocol}
\label{tab:numerics_prob2_intro}
\end{subtable}
\hspace{2em}
\begin{subtable}[t]{0.46\textwidth}
\centering
\begin{tabular}{c|c|c|c}
\multicolumn{4}{c}{\(\boldsymbol{\expval{ F^{\textup{max}}_{\textup{inv}}}(k = 2, d)}\)} \\
\hline
\(d\) & \textbf{Parallel} & \textbf{Adaptive} & \textbf{General} \\
\hline
2 & 0.9396 & 0.9671 & 0.9671\\
3 & 0.6710 & 0.7443 & 0.7443 \\
\specialrule{0.8pt}{0.8pt}{0.8pt}
\multicolumn{4}{c}{\(\boldsymbol{\expval{F^{\textup{max}}_{\textup{inv}}}(k = 3, d)}\)} \\
\hline
\(d\) & \textbf{Parallel} & \textbf{Adaptive} & \textbf{General} \\
\hline
2 & 0.9769 & 1 & 1 \\
\hline
\end{tabular}
\subcaption{Optimal average fidelities for the inversion protocol}
\label{tab:numerics_det2_intro}
\end{subtable}

\caption{Numerical results for the inversion protocol with $k \in \{2,3\}$ calls and $d \in \{2,3\}$}
\label{tab:inversion_intro}
\end{table}

From the numerical tables, we can draw the following conclusions:
\begin{itemize}
    \item For $k=3$ and $d=2$, both unitary transposition and inversion on an arbitrary known state can be implemented deterministically and exactly, even in the pure input state setting. This stands in parallel to the universal input state scenario, where $k=4$ queries were required to achieve deterministic and exact implementation at $d=2$ \cite{yoshida_exact}.  
    \item Our numerical results indicate that for $k \in \{2,3\}$ and $d \in \{2,3\}$, the repeat-until-success strategy, with success probability $p_{\textup{RUS}} = 1 - \left(1 - \tfrac{1}{d}\right)^{k}$, is optimal for both the transposition and SAR tasks.
        \item General strategies appear to offer no advantage in the tested regime, $k \in \{2,3\}$ and $d \in \{2,3\}$, for none of the considered functions. By contrast, for protocols with an unknown input state that implement unitary inversion or transposition, it has been shown that general strategies with indefinite causal order outperform sequential ones.~\cite{Quintino_2022det,Quintino2019prob}.
        \item In the unknown input state case, the performance of deterministic parallel unitary transposition was shown to be equivalent to that of deterministic unitary inversion, with both attainable through an estimation protocol~\cite{Quintino_2022det,Bisio_2010}. In contrast, this equivalence no longer holds in the known input state setting. Furthermore, in the unknown input state case, when $k<d$, deterministic sequential protocols cannot outperform parallel ones~\cite{Yoshida2024}. By contrast, in the known input state case we find that, even when $k<d$, deterministic sequential protocols can indeed surpass parallel ones.

\end{itemize}

\section{Review of basic concepts on higher-order quantum computing} \label{sec:Review_HOQC} 

Let us now consider first higher-order formalism for the case of a 
single input unitary call. To set the stage for the analysis in the context 
of a known input state, in this section we summarize the formalism developed in the literature for the case of higher-order quantum computation (HOQC) with an unknown input state~\cite{taranto2025review,
Chiribella_2008,Chiribella2007architecture,Chiribella_2009}, reviewing the 
definitions and results relevant to our work.

\subsection{The Choi representation of linear operators and linear maps}

For convenience, instead of working directly with functions, we will use the fact that quantum operations can be encoded into a single vector or matrix form. This is known as the Choi representation of linear operators and maps and it will be a central tool in our analysis.

The Choi vector representation of a linear operator $U : \mathcal{H}_{\textup{in}} \rightarrow \mathcal{H}_{\textup{out}}$ is the vector  
$\dket{U} \in \mathcal{H}_{\textup{in}}\otimes \mathcal{H}_{\textup{out}}$, and is defined as, 
\begin{align}
\dket{U}\coloneqq\sum_i \ket{i}_\textup{in} \otimes \left(U\ket{i}\right)_\textup{out}.
\end{align}
where $\{\ket{i}\}$ is the computational basis of $\mathcal{H}_{\textup{in}}$.
Let $\mathcal{L}(\mathcal{H})$ be the set of linear operators acting on $\mathcal{H}$. In this work, all linear spaces are taken as finite-dimensional complex linear spaces, that is, $\H\cong \mathbb{C}^d$, for some integer $d\in\mathbb{N}$, hence $\mathcal{L}(\mathcal{H})$ is the set of $d$ by $d$ complex matrices. The Choi operator representation of a linear map $\map{C} : \L(\H_{\textup{in}}) \to \L( \H_{\textup{out}})$ is the operator $C \in \mathcal{L}(\mathcal{H}_{\textup{in}} \otimes \mathcal{H}_{\textup{out}})$ defined as, 
\begin{align}
    C\coloneqq\sum_{i j} \ketbra{i}{j} \otimes \map{C}(\ketbra{i}{j}).
\end{align}
The Choi representation provides convenient methods to represent and to compose quantum objects, particularly useful with circuit structures. Let $\map{A} : \mathcal{L}(\mathcal{H}_1) \to \mathcal{L}(\mathcal{H}_2)$ and $\map{B} : \mathcal{L}(\mathcal{H}_2) \to \mathcal{L}(\mathcal{H}_3)$ be linear maps, then the Choi operator of their composition $\map{C} \coloneqq \map{B} \circ \map{A}$ is given by the \textit{link product}~\cite{Chiribella_2009}, $C_{13} = B_{23} * A_{12}$. The link product is defined as
\begin{align} \label{eq:link_product}
    B_{23} * A_{12} = \Tr_2\left([A_{12}^{T_2} \otimes \mathbb{1}_3][\mathbb{1}_1 \otimes B_{23}]\right),
\end{align}
where $A \in \mathcal{L}(\mathcal{H}_{1} \otimes \mathcal{H}_{2})$ and $B \in \mathcal{L}(\mathcal{H}_{2} \otimes \mathcal{H}_{3})$ are the Choi operators of the corresponding maps, and $A^{T_2}$ denotes the partial transpose of $A$ over\footnote{When we keep track of the spaces where the linear operator act, the link product is a commutative and associative operation, we refer to section ``2.2.4 Link Product: Composing Linear Maps in the Choi Representation'' of Ref.~\cite{taranto2025review} for a detailed discussion and various illustrative examples.} $\mathcal{H}_2$.

\subsection{Supermaps, superchannels, and superinstruments (with a single call of the input operation)} \label{superuniversch} 

A linear supermap, or shortly, a supermap, is a linear transformation between linear maps. For instance,\footnote{We denote supermaps by a double tilde, quantum maps by a single tilde, and matrices by simple capital letters.}
\begin{equation}
\smap{S}: ( \mathcal{L}(\mathcal{H}_{\textup{in}})  \to \mathcal{L}(\mathcal{H}_{\textup{out}})) \to ( \mathcal{L}(\mathcal{H}_{\textup{in}^{\prime}}) \to \mathcal{L}(\mathcal{H}_{\textup{out}^{\prime}})).
\end{equation}
is a supermap $\smap{S}$ that transforms maps  $\map{C}:\mathcal{L}(\mathcal{H}_{\textup{in}}) \to \mathcal{L}(\mathcal{H}_{\textup{out}})$ into linear maps $\smap{S}(\map{C}): \mathcal{L}(\mathcal{H}_{\textup{in}^{\prime}}) \to \mathcal{L}(\mathcal{H}_{\textup{out}^{\prime}})$.
In this work, we label the spaces associated with the input operation by $\textup{in}\coloneqq \textup{I}$ and $\textup{out} \coloneqq \textup{O}$, and the spaces associated with the output operation by $\textup{in}^{\prime}\coloneqq \textup{P}$ and $\textup{out}^{\prime} \coloneqq \textup{F}$, standing for ``input'', ``output'', ``past'', and ``future'', respectively.
\begin{figure}[t!]
    \centering
\includegraphics[width=0.8\textwidth]{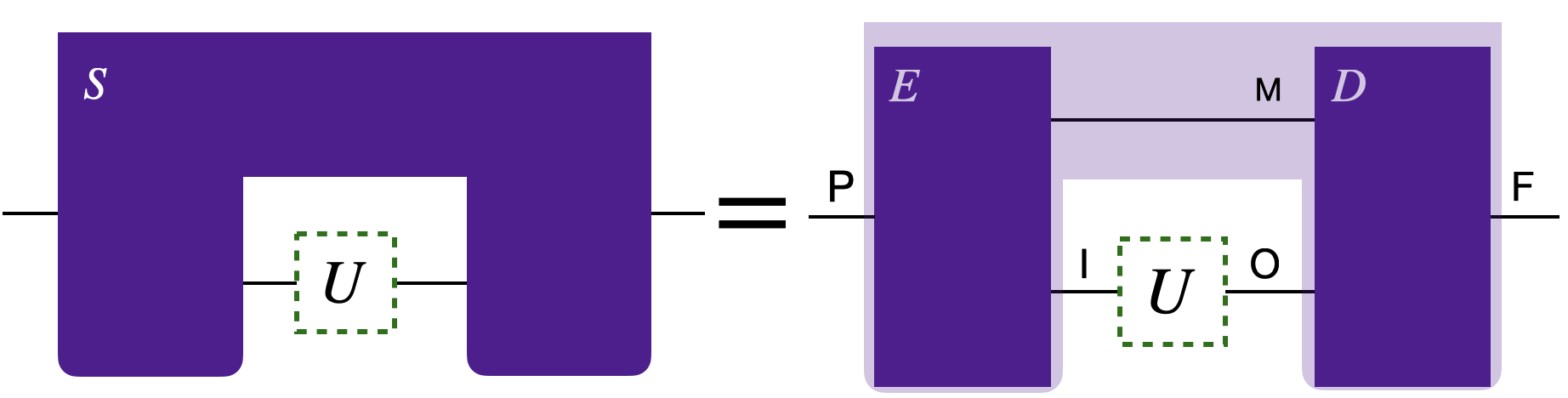}
    \caption{A superchannel acting on a universal input state and a universal single-call operation. Any superchannel can be represented by an encoder operation, an ancillary system that serves as a memory, and a decoder operation. The lines corresponding to the subsystems on which the supermap acts are labeled as follows: '$\textup{P}$' for the past, '$\textup{I}$' for the operation input, ''$\textup{O}$' for the operation output, '$\textup{M}$' for the ancillary memory system and '$\textup{F}$' for future. }
    \label{fig:unknownS} 
\end{figure}
Deterministic transformations between quantum states are described by \emph{quantum channels}—linear maps \(\map{C}\) that are \emph{completely positive} (CP) and \emph{trace preserving} (TP). Analogously, deterministic transformations between quantum channels are described by \emph{quantum superchannels}. A linear map \(\smap{S}\) is a valid transformation on CP maps if both \(\smap{S}\) and its linear extension \(\smap{S} \otimes \smap{I}\) map CP maps to CP maps, a property known as \emph{completely CP preserving} (CCPP)~\cite{Chiribella_2008,Quintino2019prob,taranto2025review}. A supermap \(\smap{S}\) is called a \emph{superchannel} if it is CCPP and also preserves the trace-preserving property—that is, it maps quantum channels to quantum channels, even when acting on part of a bipartite channel.
Every single-slot superchannel $\smap{S}$ can be realised as an ordered quantum circuit via \cite{Chiribella_2008}, 
\begin{align}\label{realis_unknown}
\smap{S}(\map{C})(\rho) = \left( \map{D}\circ [ \map{C} \otimes \map{\mathbb{1}} ] \circ \map{E}\right)(\rho),
\end{align}
where $\map{E} : \mathcal{L}(\mathcal{H}_{\P}) \rightarrow \mathcal{L}(\mathcal{H}_{\I} \otimes \mathcal{H}_{\M})$ is a quantum channel, known as \emph{encoder}, which transforms an arbitrary (unknown) input state $\rho \in \mathcal{L}(\mathcal{H}_{\P})$, and $\map{D} : \mathcal{L}(\mathcal{H}_{\O} \otimes \mathcal{H}_{\M}) \rightarrow \mathcal{L}(\mathcal{H}_{\F})$ is a quantum channel, known as \emph{decoder}, see Fig.~\ref{fig:unknownS}. The map $\map{C} : \mathcal{L}(\mathcal{H}_{\I}) \rightarrow \mathcal{L}(\mathcal{H}_{\O})$ represents an arbitrary quantum operation that is plugged into the supermap, while $\map{\mathbb{1}} : \mathcal{L}(\mathcal{H}_{\M}) \rightarrow \mathcal{L}(\mathcal{H}_{\M})$ denotes the identity map on an ancillary subspace $\mathcal{H}_{\M}$, sometimes referred to as the \emph{memory system}. As depicted in Fig.~\ref{fig:unknownS}, the encoder is a channel that may introduce this ancillary subspace, which the decoder subsequently acts upon. In the figure, we assume that the inserted operation is unitary, as this is the case we focus on in this work.

Using Choi representation, a supermap $\smap{S}:\left(\mathcal{L}\left(\mathcal{H}_{\I}\right) \rightarrow L\left(\mathcal{H}_{\O}\right)\right) \rightarrow\left(\mathcal{L}\left(\mathcal{H}_{\P}\right) \rightarrow \mathcal{L}\left(\mathcal{H}_{\F}\right)\right)$ can be represented as a map $\widetilde{S}: \mathcal{L}\left(\mathcal{H}_{\O} \otimes \mathcal{H}_{\I}\right) \rightarrow \mathcal{L}\left(\mathcal{H}_{\F}\otimes \mathcal{H}_{\P}\right)$ acting on Choi operators. By exploiting the Choi representation again, we can represent a supermap by the operator $S \in \mathcal{L}( \mathcal{H}_{\F} \otimes \mathcal{H}_{\P} \otimes \mathcal{H}_{\O} \otimes \mathcal{H}_{\I})$. In the case of a single use of the input operation, a superchannel $S$ is characterized  by the following set of constraints \cite{Chiribella2007architecture, Chiribella_2008}, 
\begin{align} \label{Eq:superch_unkn_constr}
    S_{\P \I \O \F} &\geq 0\\
    \Tr_{\F}(S_{\P \I \O \F}) &= \Tr_{\O\F}(S_{\P \I \O \F})\otimes \frac{\mathbb{1}_{\O}}{d_{\O}}\\
    \Tr_{\I\O\F}(S_{\P \I \O \F}) &= \Tr_{\P\I\O\F}(S_{\P \I \O \F}) \otimes \frac{\mathbb{1}_{\P}}{d_{\P}}\\
    \Tr(S_{\P \I \O \F}) &= d_{\P}d_{\O}
\end{align}
where $d_{i}$ is the dimension of the space $\mathcal{H}_{i}$. 

Nondeterministic transformations on quantum channels are given by \textit{superinstruments}~\cite{Quintino2019prob,taranto2025review,marco_maps}. A superinstrument is a set of supermaps $\{ \smap{S}_{i} \}$ that sum to a superchannel: $\smap{S}_{\textup{ch}} = \sum_i \smap{S}_{i}$, where $\smap{S}_i$ is the $i$-th superinstrument element. When a superinstrument $\{\smap{S_i}\}_i$ is applied to a channel $\map{C}$ and an input state $\rho$, the outcome $i$ is obtained with probability $p_i = \Tr(\smap{S}_{i}(\map{C})(\rho))$. The corresponding output state is then given by $\frac{\smap{S}_{i}(\map{C})(\rho)}{\Tr(\smap{S}_{i}(\map{C})(\rho))}$.

\begin{figure}[t!]
    \centering
\includegraphics[width=0.7\textwidth]{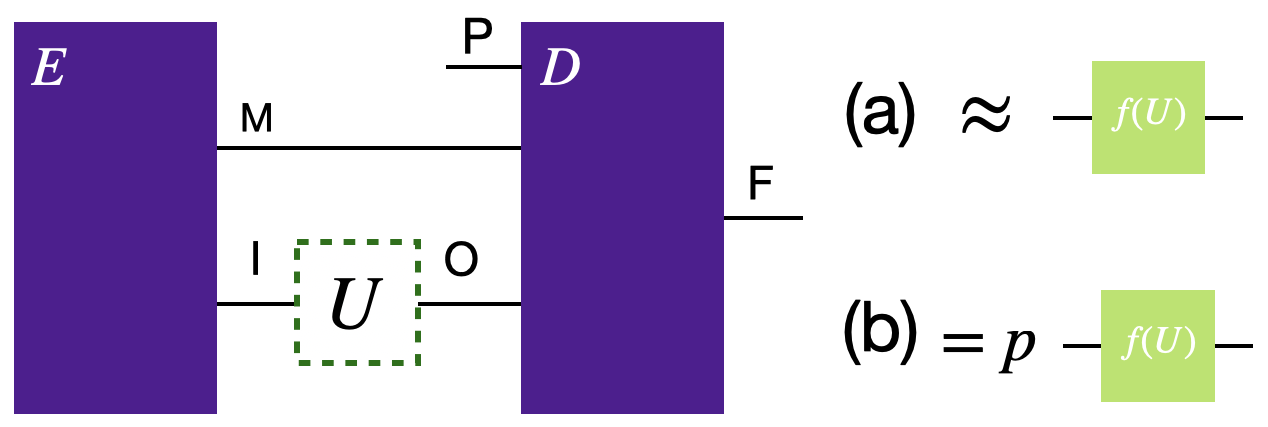}
    \caption{ A universal supermap with a delayed input state, realising transformation on an arbitrary unknown input state (a) approximately or (b) with a certain probability of success \( p \). }
    \label{fig:delay_supermap_unkn}
\end{figure} 

We also consider a class of tasks in which the unknown input state becomes available only after access to the input unitary has been lost, as illustrated in Fig.~\ref{fig:delay_supermap_unkn} and exemplified by well-known protocols such as unitary learning~\cite{Bisio_2010} and unitary storage-and-retrieval~\cite{Sedl_k_2019}. Such protocols, referred to as delayed input state protocols, can be represented by a supermap of the following form \cite{Quintino_2022det}:  
\begin{align}
\smap{S}_{\textup{delay}}(\map{C})(\rho) =  \map{D}_{\textup{delay}} \left(( \map{C} \otimes \map{\mathbb{1}} ) (E))\otimes \rho_{\textup{in}} \right)
\end{align}
where the encoder is now simply a quantum state $E \in \mathcal{L}(\mathcal{H}_{\I} \otimes \mathcal{H}_{\M})$ on which the unitary operation is applied, while the decoder map $\map{D}_{\textup{delay}} : \mathcal{L}(\mathcal{H}_{\P} \otimes \mathcal{H}_{\I} \otimes \mathcal{H}_{\M}) \rightarrow \mathcal{L}(\mathcal{H}_{\F})$ transforms the unknown, delayed input state $\rho \in \mathcal{L}(\mathcal{H}_{\P})$ into the final output. In the Choi picture, this can be represented as, 
\begin{equation}
    S_{\P\I\O\F} = E_{\M\I} * D_{\P\M\O\F}
\end{equation}
where $E_{\M\I} \in \mathcal{L}(\mathcal{H}_{\M} \otimes \mathcal{H}_{\I})$ and $D_{\\PM\O\F} \in \mathcal{L}(\mathcal{H}_{\P}\otimes\mathcal{H}_{\M} \otimes \mathcal{H}_{\O}\otimes \mathcal{H}_{\F})$.

The case depicted in Fig.~\ref{fig:unknownS}, where both the states and the operations are treated as universal inputs, even though most general from the perspective of the inputs, it is most restrictive from the perspective of the set of supermaps that can implement the task. This motivates us to explore scenarios where some knowledge of the inputs is available, as such knowledge generally enlarges the set of supermaps that can implement the function.

\section{Higher-order quantum computing with known input states}\label{sec:Knstate_intro}

As formalized in the previous section, earlier works treated a supermap as a universal map over both input states and input operations, meaning that its design and constituent gates were independent of the inputs. In contrast, we consider the case where the supermap is no longer universal with respect to the input state; rather, its structure depends on the specific state on which the desired transformation of the input operation is to be realised. On the other hand, we will keep the assumption of a supermap acting as universal machine on the input unitaries, which will be crucial in deriving symmetries of the problem. 

\subsection{General framework and definitions}

We consider higher-order quantum computation (HOQC) tasks that transform 
an arbitrary input operation an arbitrary input operation $\map{C} : \mathcal{L}(\mathcal{H}_{\I})\rightarrow \mathcal{L}(\mathcal{H}_{\O})$, into some output operation $f(\map{C}): \mathcal{L}(\mathcal{H}_{\P})\rightarrow \mathcal{L}(\mathcal{H}_\F)$ on a known state $\rho\in\L(\H_\P)$, as depicted in Fig.~\ref{fig:known_state}. As we detail later, in this known input state scenario,  it is sometimes convenient to analyse separately the \emph{pure-state case} $\map{C} \mapsto f(\map{C})(\ketbra{\psi}{\psi})\in \mathcal{L}(\mathcal{H}_\F)$ and \emph{mixed-state case} $\map{C} \mapsto f(\map{C})(\rho) \in \mathcal{L}(\mathcal{H}_\F)$ and \emph{bipartite-state case}, where we aim to apply the operation only to a part of an entangled state, $\map{C} \mapsto (f(\map{C})_{\P\to\F} \otimes \map{\mathbb{1}}_{A})(\ketbra{\psi}{\psi}_{\P\A}) \in \mathcal{L}(\mathcal{H}_\F)\otimes \mathcal{L}(\mathcal{H}_\A)$.

In this section, we focus on protocols that implement \( f(\map{C}) \) using a single call (that is, \(k=1\) calls of the input operation) of the unknown input operation. Later, we extend the analysis to protocols with multiple calls (where \(k > 1\) are available) of the input operation, implementing the transformation for any unknown unitary channel $\map{C}$. 

\begin{figure}[t!]
    \centering
\includegraphics[width=0.7\textwidth]{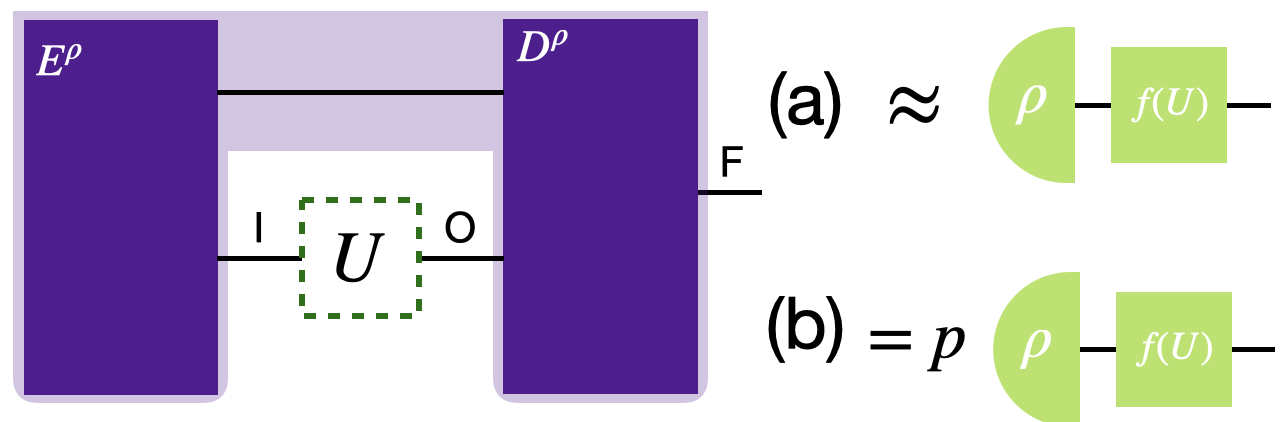}
    \caption{Supermap implementing a unitary transformation on a known state. The design of the encoder and decoder are tailored to implement a transformation of interest on the known program state. (a) Deterministic approximate realisation. (b) Probabilistic exact implementation. Elements shown in purple are optimized, dashed elements represent unknown and universal components, and green indicates that their actions are fixed as a consequence of the other elements. }
    \label{fig:known_state}
\end{figure}

As depicted in Fig.~\ref{fig:known_state}, instead of having a free input line coming from the past, the supermap now depends on classical information, denoted with a superscript\footnote{For convenience, we will often drop the parameter label of the supermap when it is clear from the context that it is a known state supermap.} \(\smap{S}^{\rho}\). Similarly to Eq.~\eqref{realis_unknown}, with a single call to the input operation we define a known state supermap via the following realisation:
\begin{align}\label{known_realis}
    \smap{S}^{\rho}(\map{C}) \coloneqq \left( \map{D}^{\rho} \circ [ \map{C} \otimes \map{\mathbb{1}} ] \right)(E^{\rho}).
\end{align}
where the encoder is a channel that takes a quantum state from system $\mathcal{H}_{\P}$, in this case, the encoder is a quantum state that depends on the known state $\rho$, i.e., $E^{\rho} \in \mathcal{L}\left(\mathcal{H}_{\I} \otimes \mathcal{H}_{\M}\right)$. The decoder is again a quantum channel, defined similarly as before, $\map{D}^{\rho}: \mathcal{L}\left(\mathcal{H}_{\M} \otimes \mathcal{H}_{\O}\right) \rightarrow \mathcal{L}(\mathcal{H}_{\F})$, although now carrying a $\rho$ parametrization. Here, $\map{C} : \mathcal{L}(\mathcal{H}_{\I}) \rightarrow \mathcal{L}(\mathcal{H}_{\O})$ is the input operation, and $\map{\mathbb{1}}$ represents the identity operation on the auxiliary space $\mathcal{L}(\mathcal{H}_{\M})$. In the Choi representation, Eq.~\eqref{known_realis} can be represented as
\begin{equation}\label{choi_kn}
    S^{\rho}_{\I\O\F} = E^{\rho}_{\I\M} * D^{\rho}_{\M\O\F},
\end{equation}
where $E^{\rho}_{\I\M} \in \mathcal{L}(\mathcal{H}_{\I} \otimes \mathcal{H}_{\M})$ and $D^{\rho}_{\M\O\F} \in \mathcal{L}(\mathcal{H}_{\M} \otimes \mathcal{H}_{\O}\otimes \mathcal{H}_{\F})$.  
Note that, $\smap{S}^\rho$ is then a supermap that transforms quantum channels into quantum states and may be viewed as a particular instance of superchannel where there is no global past system $\mathcal{H}_{\P}$, or that the global past system is one-dimensional. This set of quantum transformations also found applications in other contexts, e.g., to analyse causal correlations \cite{Costa_2017QmCaus1} \cite{Nery_2021}, quantum memory \cite{Giarmatzi_2021,Taranto_2024}, time-ordered quantum processes~\cite{taranto2025review}.

As in the previous section, we can also characterise deterministic quantum transformations from channels to states in terms of their Choi operators. A linear operator $S^{\rho} \in \mathcal{L}(\mathcal{H}_\I \otimes \mathcal{H}_{\O} \otimes \mathcal{H}_{\F})$, is the Choi operator of a superchannel $\smap{S}^\rho: (\mathcal{L}(\mathcal{H}_{\I})\to \mathcal{L}(\mathcal{H}_{\O})) \to \mathcal{L}(\mathcal{H}_{\F})$ that transforms channels into states if and only if it respects:
\begin{align}
     S^{\rho}_{\I\O\F} &\geq 0\label{Eq:superch_known_constr}\\
   \Tr_{\F}(S^{\rho}_{\I\O\F})  &= \Tr_{\O \F}(S^{\rho}_{\I\O\F})\otimes \frac{\mathds{1}_{\O}}{d_{\O}} \\
    \Tr(S^{\rho}_{\I\O\F}) &= d\label{Eq:superch_known_constr3}
\end{align}

Probabilistic transformations are again described by a superinstrument $\{ \smap{S}^{\mathcal{\rho}}_{i} \}$, whose elements sum to a superchannel: $\smap{S}_{\textup{ch}} = \sum_i \smap{S}^{\mathcal{\rho}}_{i}$, where $\smap{S}_i$ is the $i$-th superinstrument element. Contrary to the previous section, the superinstrument $\{\smap{S}^{\mathcal{\rho}}_i\}_i$ is now applied only to input channels $\map{C}$ to produce its action on a fixed input state $\rho$. The outcome $i$ is obtained with probability $p_i = \Tr(\smap{S}^{\mathcal{\rho}}_{i}(\map{C}))$ and the corresponding output state is given by $\frac{\smap{S}^{\mathcal{\rho}}_{i}\left(\map{C} \right)}{\Tr(\smap{S}^{\mathcal{\rho}}_{i}(\map{C}))}$. 

\begin{figure}[b!]
    \centering
\includegraphics[width=0.7\textwidth]{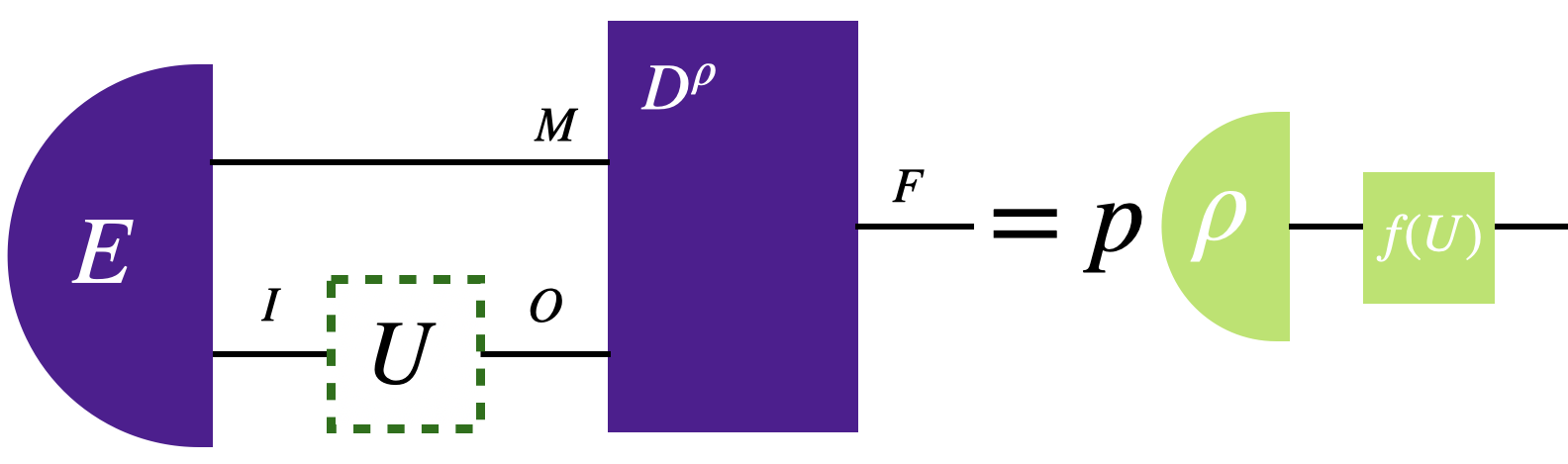}
    \caption{Known input state supermap acting on a unitary with delayed access to the input state.
 Crucially, the classical description of the state becomes available only after the storing phase, making the encoder independent of the input state.}
    \label{fig:delayed_FigKN}
\end{figure} 

\subsection{Delayed known input state case}\label{sec:delay_KN}

Analogous to the case of supermaps with unknown input state, we define delayed input state tasks as those in which the storing phase is independent of the input state. In the known input scenario, however, there is a subtlety: the structural form of a delayed input supermap is indistinguishable from that of an undelayed one, i.e. there is no additional wire explicitly representing the delay. The crucial distinction lies in the fact that the information about the state becomes available only after the operation has already been stored by the encoder. As a result, both the encoder and the storing phase must be entirely independent of the input state.

Accordingly, we define delayed input protocols as those where access to the classical description of the input state is delayed, as depicted in Fig.~\ref{fig:delayed_FigKN}, and formally given by
\begin{equation}\label{choi_kn_delay}
S^{\rho}_{\I\O\F} = E_{\I\M} * D^{\rho}_{\M\O\F},
\end{equation}
where $E_{\I\M} \in \mathcal{L}(\mathcal{H}_{\I} \otimes \mathcal{H}_{\M})$ is the state independent encoder and $D^{\rho}_{\M\O\F} \in \mathcal{L}(\mathcal{H}_{\M} \otimes \mathcal{H}_{\O}\otimes \mathcal{H}_{\F})$ is the decoder, parametrised with the classical knowledge.

\section{HOQC with known input states under group symmetries}

In this work, we consider the action of a supermap on an arbitrary unknown unitary input operation\footnote{Here, unitary operation refers to channels given by $\map{C}(\cdot) = U(\cdot)U^{\dagger}$, where $U\in\mathrm{SU}(d)$ is a unitary operator.}.  To specify the action of the supermap, we define the function it realises, in particular we analyse functions $f:\mathrm{SU}(d)\to\mathrm{SU}(d)$ transforming $d$-dimensional unitary operators into $d$-dimensional unitary operators. In this work, we focus on two classes of functions: \emph{homomorphic} ones, e.g., $f$ satisfies \( f(UV) = f(U)f(V) \), and \emph{antihomomorphic} ones, for which  $f$ satisfeis \( f(UV) = f(V)f(U) \). 

Up to unitary equivalence, there are only three homomorphisms~\cite{fulton1991representation} 
from $\mathrm{SU}(d)$ to itself: the trivial map $f_{\textup{triv}}(U) = \mathbb{1}$, 
the identity $f_{\textup{id}}(U) = U$, and complex conjugation $f_{\textup{conj}}(U) = \overline{U}$.
Similarly, the only antihomomorphisms from $\mathrm{SU}(d)$ to $\mathrm{SU}(d)$ are: the trivial map \( f_{\textup{triv}}(U) = \mathbb{1} \), the transposition \( f_{\textup{trans}}(U) = U^{T} \) and the inversion \( f_{\textup{inv}}(U) = U^{-1} \).

For our purposes, even though identity map \( f_{\textup{id}}(U) = U \) might seem trivial, it is non-trivial in scenarios where the input state is delayed, as shown in Fig.~\ref{fig:delayed_FigKN}. In the context of known state scenario, we define delayed-input protocols as protocols where the access to the classical description of the input state is delayed. As depicted in Fig.~\ref{fig:delayed_FigKN}, the storage phase can be regarded as fully universal and inaccessible to the user, while the retrieved operation is specifically designed for a state known to the user.
Specifically, the delayed-input case of \( f(U) = U \), for the unknown input state case, corresponds to the task known as storage-and-retrieval~\cite{Sedl_k_2019}, also referred to as unitary learning~\cite{Bisio_2010}, and is closely related to unitary programming~\cite{Nielsen_1997}.
In this task, one aims to recover the action of a unitary operation on a quantum state that becomes available only after the unitary itself is no longer accessible. That is, the unitary must be stored in quantum memory during the access phase, and later retrieved and applied once the target state arrives.

To summarise, the functions that we will closely analyse in Sec.~\ref{sec:tasks_applications}, and that correspond to all the nontrivial homomorphisms and antihomomorphisms, are: transposition 
$f(U)=U^T$, conjugation $f(U)=U^*$, inversion $f(U)=U^{-1}$  and $f(U) = U$ with a \emph{delayed input state} (storage-and-retrieval scenario). 
When the number of calls to the input unitary is restricted, we often encounter no-go results that prohibit deterministic and exact implementations \cite{ChiribellaEbler_2016,Bisio_2013Parallel}. 
To circumvent these limitations, one instead considers either \emph{probabilistic exact} realisations~\cite{Quintino2019prob}, or \emph{deterministic approximate} ones~\cite{Quintino_2022det,Ebler_2023}, where the transformation is implemented with a certain success probability or average fidelity, as depicted in Fig.~\ref{fig:known_state} for the unitary case. In this work we determine the optimal circuit structure that achieves the desired transformation with the highest success probability or fidelity. 

In this section, we restrict our attention to the case of a single use of the input unitary transformation. Later, in Sec.~\ref{sec:k>1} we analyse the case where the user has in access to $k$ copies of the input operation, where the deterministic and exact implementations of some of the functions above are possible~\cite{Miyazaki_2017, yoshida_exact,chen2024quantum}.

\subsection{Probabilistic exact unitary transformations}\label{sec:probabilistic}

Given a function $f: \mathrm{SU}(d) \rightarrow \mathrm{SU}(d)$, we say that there exists a probabilistic exact protocol that implements $f$ on a particular quantum state $\rho$ with probability $p$ if there exists a superinstrument $\{ \smap{S}_{\checkmark}, \smap{S}_{\times}\}$ where, for unitary channel $\map{U}(\cdot)=U(\cdot)U^\dagger $ we have
\begin{align}\label{probab1}
\smap{S}_{\checkmark}(\map{U}) = pf(U)\rho f(U)^{\dagger} \quad \forall U \in \mathrm{SU}(d).
\end{align}
The supermap $\smap{S}_{\checkmark}$ is part of the superinstrument $\{ \smap{S}_{\checkmark}, \smap{S}_{\times}\}$, where the success branch $\smap{S}_{\checkmark}$ implements the desired transformation exactly, while the failure branch $\smap{S}_{\times}$ produces an outcome that is usually discarded. This means that there exists a total deterministic supermap $\smap{S}_{\textup{ch}}$ such that\footnote{For notational simplicity, in some parts of the article such as in the results sections, we denote the success branch supermap $\smap{S_{\checkmark}}$ simply by \( S \).} $\smap{S}_{\checkmark} + \smap{S}_{\times} = \smap{S}_{\textup{ch}}$. 

In the Choi representation,  the problem can be formulated such that its solution corresponds to the supermap maximizing the probability in ~\eqref{probab1}. Given a function $f: \mathrm{SU}(d) \rightarrow \mathrm{SU}(d)$, we define probabilistic exact realisation of a task that applies the transformation $f(U)$ on an arbitrary but known state $\rho$ as:
\begin{align}\label{Eq:prob_SDP1}
    &\textup{max} \; \; p \\
    \textup{s.t.} \; \; &S^{\rho}_{\checkmark, \I\O\F}*\dketbra{U}{U}_{\I\O} = pf(U) \rho_{\F} f(U)^{\dagger} \;\;\; \forall U\in\mathrm{SU}(d) \label{Eq:def_prob_impl} \\
    &0 \leq S^{\rho}_{\checkmark,\I\O\F} \leq S^{\rho}_{\textup{ch}, \I\O\F} \\
    &\Tr_{\F}(S^{\rho}_{\textup{ch}}) = \sigma_{\I} \otimes \mathbb{1}_{\O} \quad \Tr(\sigma_{\I}) = 1.
\end{align}
Here we used equivalent but alternative way of writing ~\eqref{Eq:superch_known_constr}); namely, we introduced additional variable $\sigma_{\I} = \frac{\Tr_{\O\F}(S^{\rho}_{\O\F\I})}{d_{\O}}$, where $\sigma_{\I}$ is a quantum state, i.e., $\Tr(\sigma_{\I}) =1$. 

The problem stated in Eq.~\eqref{Eq:prob_SDP1}, maximizes a linear objective function subjected to positive semidefinite and affine constraints. Problems of such nature can be tackled with the formalism of semidefinite programming (SDP) Ref.~\cite{watrous2018theory, Skrzypczyk2023_SDP}. The SDP formulation enables us to employ various analytical and numerical methods to tackle problems consider in this work. Analytical methods were used to prove theorems presented in the single slot case in Sec.~\ref{sec:tasks_applications}, while numerical were used to obtain results and provide insights for $k>1$ calls presented in Sec.~\ref{sec:application_k>1}.

We note however that, in the problem above, there are infinitely many constraints, since we require the equation $S^{\rho}_{\checkmark, \I\O\F}*\dketbra{U}{U}_{\I\O} = pf(U) \rho_{\F} f(U)^{\dagger}$ to hold for all $ U\in\mathrm{SU}(d)$. In the next subsection, we show that the problem defined in Eq.~\eqref{Eq:def_prob_impl} exhibits several symmetries that significantly simplify its structure. In particular, these symmetries allow us to reduce the infinitely many constraints to a small, finite set\footnote{Also, due to linearity, one can convert the infinite set of constraints in Eq.~\eqref{Eq:def_prob_impl} into a finite set of constraints. For that, it is enough that the equation $S^{\rho}_{\checkmark, \I\O\F}*\dketbra{U}{U}_{\I\O} = pf(U) \rho_{\F} f(U)^{\dagger}$ holds for a finite set of unitaries $\{U_i\}_i$ such that the linear space spanned by $\{\dketbra{U_i}{U_i}\}_i$  is equals to the linear space spanned by $\{\dketbra{U}{U}\}$ where $U\in\mathrm{SU}(d)$~\cite{Quintino2019prob}. Note however that this approach does not make explicit use of the symmetries, hence not considerably simplifying the problem.}.

\subsubsection{Symmetries of the problem}\label{sec:sym_prob}

An important consequence of requiring a supermap to act as a universal machine over all input unitaries, i.e., Eq.~\eqref{Eq:def_prob_impl} holds $\forall U\in\mathrm{SU}(d) $, is that the defining equation of the protocol exhibits a specific symmetry properties. These symmetries are determined by the function being implemented and will further be used to simplify the SDP problem, allowing to solve some of this problems analytically. We now identify three types of symmetries that can be obtained from the problem mathematically formalised as an SDP at Eq.~\eqref{Eq:prob_SDP1}, where two are particular to HOQO with known input state task. 
\begin{figure}[t!]
    \centering
\includegraphics[width=0.9\textwidth]{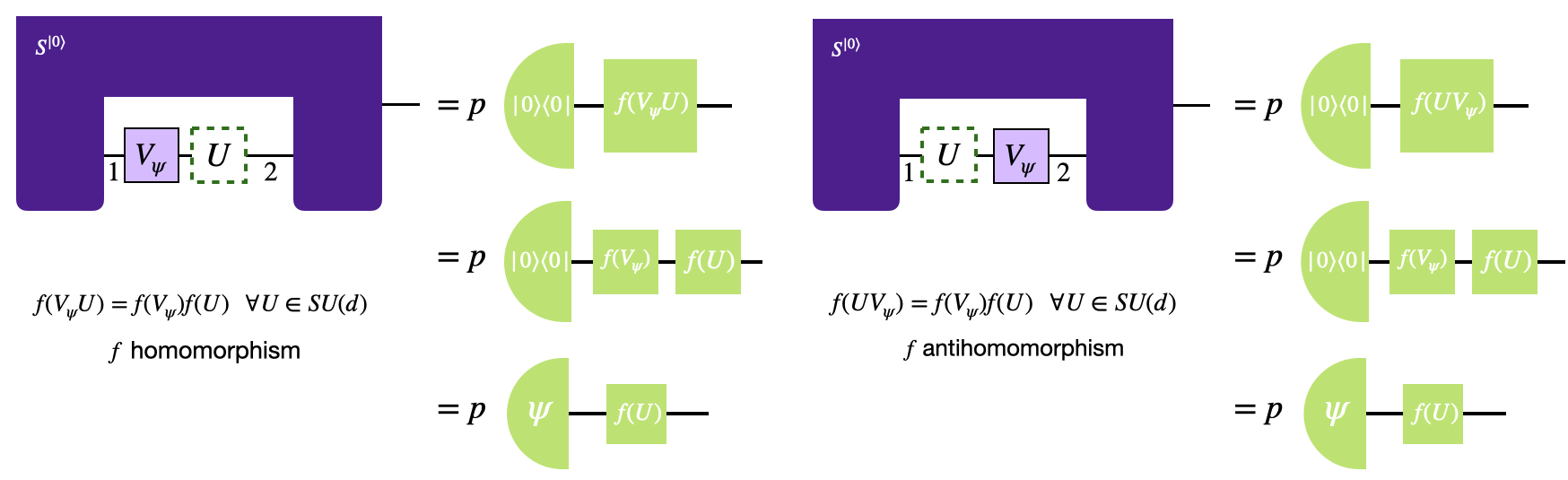}
    \caption{ Diagrammatic proof that the performance of the supermap implementing homomorphic and antihomomorphic transformations is independent of the choice of the known input state. The figure shows that given a supermap $S$ that transforms $U$ into $f(U)\ket{0}$ with probability $p$, there exits a supermap $S^{c}$  that transforms $U$ into $f(U)\ket{\psi}$ with probability $p$, where $\ket{\psi}$ is an arbitrary qudit state, and $V_\psi$ is a unitary operation that such that $f(V_\psi)\ket{0}=\ket{\psi}$ . Hence, without loss in performance, we can consider a supermap acting on a $\ketbra{0}{0}$.
    }
    \label{fig:independence_of_state}
\end{figure}
\begin{itemize}
    \item \textbf{Independence from the choice of known state (pure state case):} Without loss of performance, we can choose to solve the problem for any arbitrary pure state \( \ket{\psi} \). Specifically, for the functions considered in this work, any protocol that converts an arbitrary $U$ into $f(U)\ket{0}$ with probability $p$, can be converted into another protocol arbitrary $U$ into $f(U)\ket{\psi}$ with the same probability $p$, where $\ket{\psi}$ is a known arbitrary state. We now state this result formally, and its proof is illustrated in Fig.~\ref{fig:independence_of_state}.
    
    Even if this section is focused on transformations where the input unitary operation is called only once, the theorem below is stated in a more general version, where $k$ calls of the input unitary operations are available, as we will detail in Sec.~\ref{sec:k>1}.
    
    \begin{restatable}[Independence of the choice of known state]{theorem}{IndependenceThmProb} \label{thm:independ}
    Let $f:\mathrm{SU}(d) \rightarrow \mathrm{SU}(d)$ be a homomorphism or antihomomorphism, and let $S^{\ket{0}}_{\I\O\F}$ be a parallel/sequential/general\footnote{When a single call of the input operation is available, i.e., $k=1$, there is no distinction between these sets of strategies. In the case of $k>1$ calls, we distinguish three classes of protocols: parallel, sequential, and general strategies, which we will specify in detail in Sec.~\ref{sec:k>1}.} $k$-slot superinstrument element such that
    \begin{equation}
         S^{\ket{0}}_{\I\O\F} * \ChoiU^{\otimes k}_{\I\O} = p f(U)\ketbra{0}{0 } f(U)^{\dagger} \quad \forall U\in\mathrm{SU}(d).
    \end{equation}
That is, it transforms $k$ calls of $U$ into the state $f(U)\ket{0}$ with probability $p$.

For any pure quantum state $\ket{\psi}\in\mathbb{C}^d$, there exists a parallel/sequential/general superinstrument with an instrument element such that
 \begin{equation}
         S^{\ket{\psi}}_{\I\O\F} * \ChoiU^{\otimes k}_{\I\O} = p f(U)\ketbra{\psi}{\psi } f(U)^{\dagger} \quad \forall U\in\mathrm{SU}(d)
    \end{equation}
with the same probability $p$. 
\end{restatable}

We remark that the situation is more subtle when dealing with mixed states and bipartite states that may be entangled with an external reference. The mixed case is discussed in Sec.~\ref{sec:mixed}, and the bipartite case is discussed in Sec.~\ref{sec:bip}.

\begin{figure}[t!]
  \centering
  \begin{subfigure}[b!]{0.7\textwidth}
    \centering
    \includegraphics[width=\linewidth]{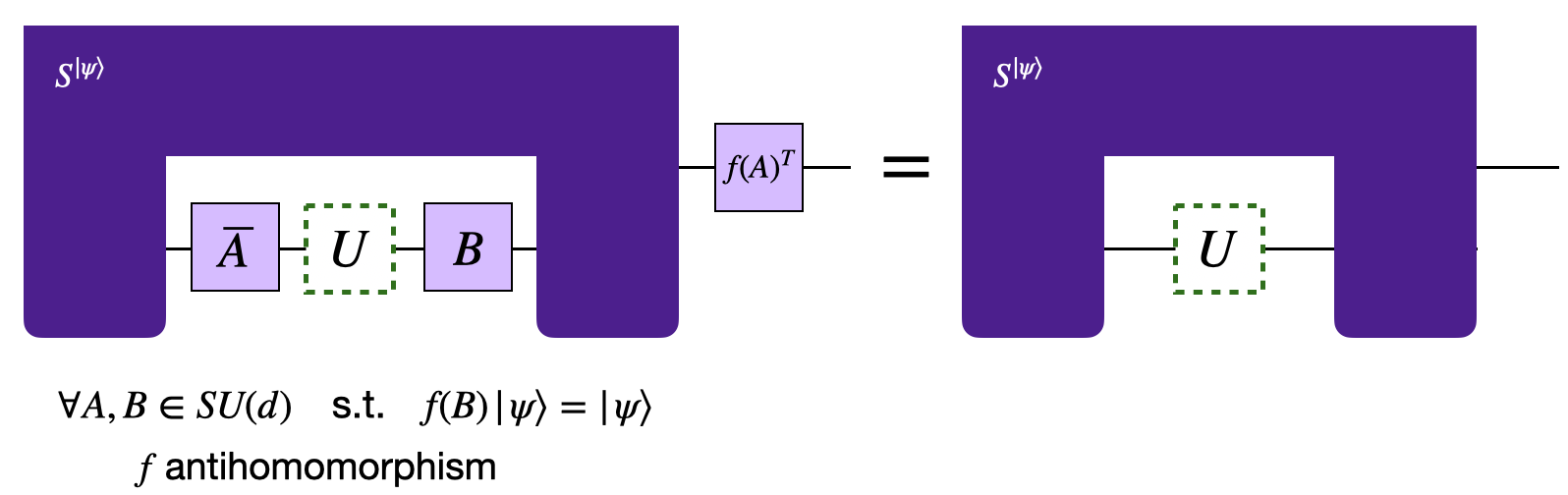}
    \captionsetup{width=1.0\textwidth}
    \caption{Full twirling symmetry of the known state supermap implementing an antihomomorphic function. The symmetry holds for all \( A, B \in \mathrm{SU}(d) \), where \( f(B)\ket{\psi} = \ket{\psi} \).}
    \label{fig:sym1}
  \end{subfigure}
  \quad \quad 
  \begin{subfigure}[b!]{0.7\textwidth}
    \centering
    \includegraphics[width=\linewidth]{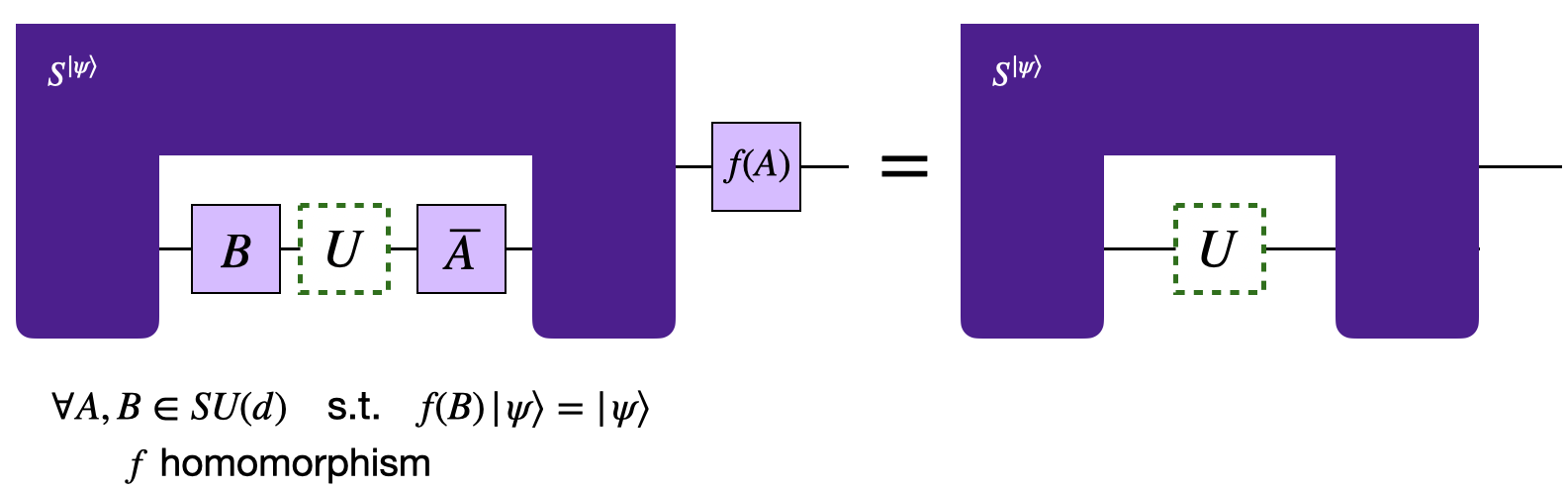}
    \captionsetup{width=1.0\textwidth}
    \caption{Full twirling symmetry of the known state supermap implementing a homomorphic function. The symmetry holds for all \( A, B \in \mathrm{SU}(d) \), where \( f(B)\ket{\psi} = \ket{\psi} \).
}
    \label{fig:sym2}
  \end{subfigure}
\captionsetup{width=0.9\textwidth}
  \caption{For homomorphic and antihomomorphic functions $f$, a supermap $S^{\ket{\psi}}$ is said to be covariant if it respects the symmetries described in this figures. As proved in Thm.~\ref{thm:twirl_prob_constr} and Thm.~\ref{thm:KN_sym}, when our goal is to transform an arbitrary unitary channel represented by $U\in\mathrm{SU}(d)$ into the state $f(U)\ket{\psi}$, we may restrict ourselves to covariant supermaps without loss in performance. Also, this covariance property holds even when multiple calls of the input unitary are available. 
  }
  \label{fig:symmetries}
\end{figure}

\item \textbf{Covariance symmetries:} Covariance symmetries arise directly from the universality over input unitaries and will allow us to  restrict our optimisation set to supermaps that are invariant under a particular group twirling transformation. Here, we state two such symmetries: one that closely relates the case of an unknown state supermap~\cite{Quintino2019prob,Quintino_2022det}, as depicted in Fig.~\ref{fig:symmetries}, and another that is intrinsic to the structure of the the known state problem, as depicted in Fig.~\ref{fig:sym_KN}. The former allows us to identify the basis of the supermap that maximizes problem in Eq.~\eqref{Eq:prob_SDP1}, and we present it in the following theorem.

\begin{restatable}[Input state independent covariance]{theorem}{TwirlProbConstr}\label{thm:twirl_prob_constr}
Let $f:\mathrm{SU}(d) \rightarrow \mathrm{SU}(d)$, a homomorphism or antihomomorphism, $\ket{\psi}\in\mathbb{C}^d$ be a fixed state, and $S_{\I\O\F}$ be a parallel/sequential/general $k$-slot superinstrument element such that
    \begin{equation}
         S_{\I\O\F} * \ChoiU^{\otimes k}_{\I\O} = p f(U)\ketbra{\psi} f(U)^{\dagger} \quad \forall U\in\mathrm{SU}(d).
    \end{equation}
    That is, it transforms $k$ calls of $U$ into the state $f(U)\ket{\psi}$ with probability $p$.

There exists a parallel/sequential/general superinstrument with an instrument element $S^{c}_{\I\O\F}$ such that
 \begin{equation}
         S^{c}_{\I\O\F} * \ChoiU^{\otimes k}_{\I\O} = p f(U)\ketbra{\psi}{\psi } f(U)^{\dagger} \quad \forall U\in\mathrm{SU}(d)
    \end{equation}
with the same probability $p$, where  $S^{c}_{\I\O\F}$ respects $[S^{c}, U_f]=0$, for every $U_f$ s.t.: $U_f \coloneqq U^{\otimes k}_{\I} \otimes \mathbb{1}^{\otimes k}_{\O} \otimes \overline{U}_{\F}$ for transposition, $U_f = \mathbb{1}^{\otimes k}_{\I} \otimes U^{\otimes k}_{\O} \otimes U_{\F} $ for conjugation and $U_f = U^{\otimes k}_{\I} \otimes \mathbb{1}^{\otimes k}_{\O} \otimes U_{\F}$ for inversion (as discussed in Sec.~\ref{symmetries}). That is, for every $S_{\I\O\F}$, there exists $S^{c}_{\I\O\F}$ covariant with respect to $U_f$ group s.t. it attains the same success probability.

Additionally, if $S^c_{\I\O\F}$ is covariant superinstrument element $[S^c, U_f]=0$ for its associated $f$, then  
\begin{equation}   S^c_{\I\O\F}*\dketbra{\mathbb{1}}{\mathbb{1}}^{\otimes k}_{\I\O} = p f(\id) \ketbra{\psi}{\psi}_{\F}  f(\id)^\dagger 
\end{equation}
implies
\begin{equation}
     S^c_{\I\O\F}*\dketbra{U}{U}^{\otimes k}_{\I\O} = p f(U)\ketbra{\psi}{\psi}_{\F}f(U)^{\dagger} \quad \forall U \in \mathrm{SU}(d).
\end{equation}
That is, when a covariant protocol $S^c$  transforms $k$ calls of the identity channel described by the unitary $U=\id$ into the state $f(\id)\ket{\psi}$ with probability $p$, it necessarily transforms $k$ calls of an arbitrary unitary channel $U$
into $f(U)\ket{\psi}$ with the same probability $p$.
\end{restatable}

\begin{figure}[t!]
    \centering
\includegraphics[width=0.9\textwidth]{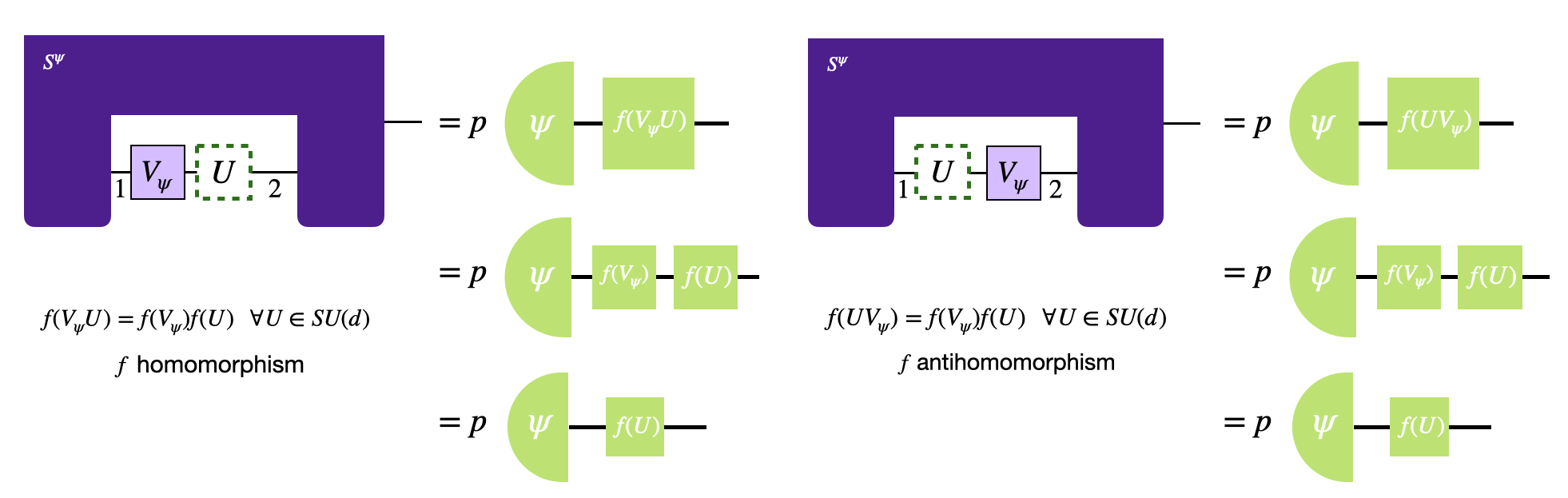}
    \caption{ Input state dependent covariance of known input state supermaps implementing homomorphic and antihomomorphic functions. This constitutes the second covariance property, specific to known state 
supermaps, with respect to which the supermap can be twirled without any 
loss of performance, thereby simplifying the problem.}
    \label{fig:sym_KN}
\end{figure}

Note that for all functions $f$ considered here, $U\mapsto U_f$ is homomorphism from $\mathrm{SU}(d)$  to $\mathrm{SU}(d)^{\otimes 2k+1}$, and thus a representation of $\mathrm{SU}(d)$. Hence, we can employ methods from representation theory. 

When $\mathcal{H}$ is a finite-dimensional Hilbert space carrying a unitary representation of a compact group, such as $\mathrm{SU}(d)$, the space $\mathcal{H}$ can be decomposed as $\mathcal{H} \cong \bigoplus_\lambda \mathcal{R}_\lambda \otimes \mathcal{M}_\lambda$, where $\lambda$ runs over all irreducible representations (irreps) of the group $\mathrm{SU}(d)$, $\mathcal{R}_\lambda$ is the space carrying the irrep $\lambda$, and $\mathcal{M}_\lambda$ is its corresponding multiplicity space~\cite{fulton1991representation}. In this decomposition, we can write
$U_f= \bigoplus_\lambda U_\lambda \otimes \id_{\mathcal{M}_\lambda}$, where $U_\lambda \in \mathcal{L}(\mathcal{R}_\lambda)$. 

In the following sections, up to Sec.~\ref{sec:k>1}, we will first focus on the case of $k=1$ query of the input unitary. For this $k=1$ scenario, we have the following explicit statements.
\begin{itemize}
\label{symmetries}
    \item \textbf{Transposition} exhibits the symmetry $[S_{\I\O\F}, U_{\I} \otimes \id_\O \otimes \overline{U}_{\F}] = 0 \; \; \forall U$. The linear space of operators that commute with $U_\I \otimes \overline{U}_{\F}$ for all $U \in \mathrm{SU}(d)$ is the two-dimensional space with basis
$\{ \ketbra{\phi^{+}}{\phi^{+}}_{\I\F}, \mathbb{1}_{\I\F} - \ketbra{\phi^{+}}{\phi^{+}}_{\I\F} \}$. 
    \item 
\textbf{ Storage and Retrieval (SAR)} As we will make precise in Sec.~\ref{sec:SAR} and Thm.~\ref{thm:SAR_transp_Eq}, the SAR problem is equivalent to the problem of transposition\footnote{This equivalence is also generalised to the case where multiple calls of the unknown input unitary are available, where unitary SAR is equivalent to parallel unitary transposition.}.  
\item 
\textbf{Conjugation} exhibits the symmetry 
$[S_{\I\O\F}, \id_\I \otimes U_{\O} \otimes U_{\F}] = 0 \;\; \forall U$, 
which implies that the linear space of operators commuting with 
$U_{\O} \otimes U_{\F}$ for all $U \in \mathrm{SU}(d)$ 
is the two-dimensional space with basis 
$\{ P^{S}_{\O\F}, P^{A}_{\O\F} \}$, 
where $P^{S}$ and $P^A$ are the projectors onto the symmetric and anti-symmetric subspaces, respectively.
\item 
\textbf{Inversion} exhibits the symmetry $[S_{\I\O\F}, U_\I \otimes \id_\O \otimes U_{\F}] = 0 \; \; \forall U$ which implies a basis for the linear space of operators that commute with all unitaries $U_\I\otimes {U}_\F$  is $\{ P^{S}_{\I\F}, P^{A}_{\I\F} \}$. 
\end{itemize}

As previously mentioned, the known input state problem possesses a symmetry not present in the unknown input state case. 

\begin{restatable}[Input state dependent covariance]{theorem}{KNsym}\label{thm:KN_sym}
Let $f:\mathrm{SU}(d) \rightarrow \mathrm{SU}(d)$ be a homomorphism or antihomomorphism, let $\ket{\psi}\in\mathbb{C}^d$ be a fixed state, and let $S_{\I\O\F}$ be a parallel/sequential/general $k$-slot superinstrument element such that
\begin{equation}
         S_{\I\O\F} * \ChoiU^{\otimes k} = p f(U)\ketbra{\psi} f(U)^{\dagger} \quad \forall U\in\mathrm{SU}(d).
\end{equation}
That is, it transforms $k$ calls of $U$ into the state $f(U)\ket{\psi}$ with probability $p$.

There exists a parallel/sequential/general superinstrument with an instrument element $S^{c}_{\I\O\F}$ such that
\begin{equation}
         S^{c}_{\I\O\F} * \ChoiU^{\otimes k} = p f(U)\ketbra{\psi}{\psi } f(U)^{\dagger} \quad \forall U\in\mathrm{SU}(d)
\end{equation}
with the same probability $p$, where $S^{c}_{\I\O\F}$ is covariant; that is, it respects $[S^{c}, V_f]=0$ for every $V$, where \( V_f \) is defined as follows:
\begin{itemize}
  \item For transposition: \( \mathbb{1}_{\I}^{\otimes k}\otimes \overline{V}_{\O}^{\otimes k}\otimes \mathbb{1}_{\F}  \), with $V\in \mathrm{SU}(d)$ and \( V_{\O}^T \ket{\psi} = \ket{\psi} \),
  \item For conjugation: \( V_f = V_{\I}^{\otimes k} \otimes \mathbb{1}^{\otimes k}_{\O} \otimes \mathbb{1}_{\F} \), with $V\in \mathrm{SU}(d)$ and  \( \overline{V}_{\I} \ket{\psi} = \ket{\psi} \),
  \item For inversion: \( \mathbb{1}_{\I}^{\otimes k}\otimes \overline{V}_{\O}^{\otimes k}\otimes \mathbb{1}_{\F} \), with $V\in \mathrm{SU}(d)$ and  \( V_{\O}^{-1} \ket{\psi} = \ket{\psi} \). 
\end{itemize}
\end{restatable}

\end{itemize}

By exploiting these symmetries, we derive analytical solutions to the SDP problems addressed in this work, as detailed in Sec.~\ref{sec:tasks_applications}.

\subsection{Deterministic non-exact unitary transformations on known pure states}\label{sec:deterministic}

Another way to circumvent the impossibility of deterministic and exact realisation is to require the transformation to be deterministic, but allow the outcome to be non-exact. In this case, our goal is to obtain the best possible approximation. More precisely, given a function $f: \mathrm{SU}(d) \rightarrow \mathrm{SU}(d)$, and a fixed quantum state $\rho$, we have a deterministic approximate realisation of a task if there exists a supermap implementation that approximately achieves the desired transformation, but implements it in every run, i.e. for unitary channels $\map{U}(\rho)=U \rho U^\dagger$, we have that
\begin{align}\label{approx_general_rel}
\smap{S}^{\rho}(\map{U}) \approx f(U)\rho f(U)^{\dagger} \quad \forall U \in \mathrm{SU}(d)
\end{align}
or in the Choi picture,
\begin{equation}
    S^{\rho}_{\I\O\F}*\ChoiU_{\I\O} \approx f(U)\rho f(U)^{\dagger} \quad \forall U \in \mathrm{SU}(d).
\end{equation}
One might notice that this trivially holds for any supermap construction, as one can always achieve poor approximations of the given function. Also, one may always use a using the measure-and-prepare strategy, which we estimate the input unitary operation and then simply prepares the desired state $f(U)\rho f(U)^\dagger$.
However, our focus is on quantifying how close the realised transformation is to the desired function and identifying the best possible approximation based on the chosen figure of merit.

There are several ways to define a figure of merit, which, in principle, give different meaning to Eq.~\eqref{approx_general_rel}. To define how close the outcome of a protocol is to the desired realisation of the function, we will use the notion of fidelity. Fidelity is generally used to measure the closeness between quantum states or channels, in the following way: 
\begin{itemize}
  \item The fidelity between two arbitrary quantum states  $\sigma, \rho \in \mathcal{L}(\mathcal{H})$ is given by
  \begin{align}
     F(\sigma, \rho)\coloneqq \left(\tr(\sqrt{ \sqrt{\sigma} \rho \sqrt{\sigma}}) \right)^2.
  \end{align}
  When one of the states is pure, for instance, if $\sigma=\ketbra{\psi}$, the formula simplifies to:
  \begin{align}
  F(\ketbra{\psi}, \rho) = \bra{\psi}  \rho \ket{\psi}.
  \end{align}
  
  \item The channel fidelity, between an arbitrary channel $\map{C}:\mathcal{L}(\mathcal{H}_\textup{P}) \to \mathcal{L}(\mathcal{H}_\textup{F})$ and a unitary channel $\map{U}$ is given by:
  \begin{equation}
      F(C, \dketbra{U}{U}) = \frac{1}{d_\P^2} \dbra{U} C\dket{U},
  \end{equation}
  where \(d_\P\) is the dimension of $\mathcal{H}_\P$ and \(C,\dketbra{U}{U} \in \mathcal{L}\left(\mathcal{H}_\P \otimes \mathcal{H}_\F\right)\) are the respective Choi operators of $\map{C}$ and $\map{U}$.
\end{itemize}

In contrast to the unknown state protocols considered in~\cite{Quintino_2022det}, we work with state fidelity rather than channel fidelity, since the output quantity \( S_{\I\O\F} * \ChoiU_{\I\O} \) is a quantum state. When analysing deterministic, non-exact transformations via state fidelity, we restrict our attention to the case where the known quantum state is pure, i.e., \( \rho = \ketbra{\psi} \), as the fidelity function takes a considerably simpler form in this setting. Additionally, in Sec.~\ref{sec:tasks_applications}, we restrict to the case of a single input call to the unitary (\( k = 1 \)), although we present the definitions in full generality to allow for considerations with \( k > 1 \). In this case, the fidelity associated with our tasks is given by:
\begin{align}\label{fidelity}
    F(S_{\I\O\F}*\ChoiU^{\otimes k}_{\I\O}, f(U) \ketbra{ \psi}{ \psi}_{\F}f(U)^{\dagger}]) =& \Tr\left( [S_{\I\O\F} * \ChoiU^{\otimes k}_{\I\O} ] \; [f(U) \ketbra{ \psi}{ \psi}_{\F}f(U)^{\dagger}]\right)   \\
    =& \Tr\left(S_{\I\O\F} \; (\ChoiU_{\I\O}^{T})^{\otimes k} \otimes  [f(U) \ketbra{ \psi}{ \psi}_{\F}f(U)^{\dagger}]\right)  \\
    =& \Tr(S_{\I\O\F} \; \dketbra{\overline{U}}{\overline{U}}^{\otimes k}_{\I\O} \otimes  [f(U) \ketbra{ \psi}{ \psi}_{\F}f(U)^{\dagger}]).
\end{align}
This equation would be an appropriate figure of merit if the task were defined for a single unitary $U$. Since our goal is to define the task for all unitaries, we have multiple ways to formulate our figure of merit based on fidelity:
\begin{itemize}
    \item Average fidelity: 
    \begin{equation}
    \expval{ F^{\psi}} \coloneqq\int_{\text {Haar }} F\left(S^{\psi}_{\I\O\F}*\ChoiU^{\otimes k}_{\I\O},f(U)\ketbra{ \psi}{ \psi}_{\F}f(U)^{\dagger}_{\F} \right) \mathrm{d} U.
\end{equation}
where the integral represent the average of fidelity over the Haar measure $\mathrm{d}U$ of $\mathrm{SU}(d)$.

\item Worst-case fidelity:
\begin{align}\label{eq:def_F_wc}
    F^{\psi}_{\mathrm{wc}} \coloneqq \underset{U \in \mathrm{SU}(d)}{\mathrm{inf}}F(S^{\psi}_{\I\O\F}*\ChoiU^{\otimes k}_{\I\O}, f(U)\ketbra{ \psi}{ \psi}_{\F}f(U)^{\dagger}_{\F} )
\end{align} 
that evaluates the minimum fidelity across all possible inputs. The purpose of this figure of merit is to determine the worst-case scenario that one can expect when choosing a random unitary from the set of all possible unitaries.

\item White noise visibility: A superchannel $S$ is said to have white noise visibility if the transformation of $k$ copies of $U$ into $f(U)$ we can write, \begin{equation}\label{Eq:visib_fig_mer}
    S^{\rho}_{\I \O \F} *\dketbra{U}{U}_{\I \O} ^{\otimes k}=\eta \rho_{\P} * \dketbra{f(U)}{f(U)}_{\P \F}+(1-\eta) \frac{\mathbb{1}_{\F}}{d} \quad \forall U\in\mathrm{SU}(d)
\end{equation}
where $\rho \in \mathcal{L}(\mathbb{C}^{d})$ is a quantum state on which $S$ is implementing transformation and $\eta$ is white noise visibility. The parameter of white noise visibility quantifies how much of the original state remains distinguishable from the completely mixed state.
We will discuss this further in Sec.~\ref{sec:mixed}.
\end{itemize}

Analogously to the unknown input state scenario~\cite{Quintino_2022det}, for supermaps satisfying the covariance properties detailed in Sec.~\ref{sec:sym_det}, the figures of merit introduced here are equivalent.
More precisely, for homomorphic or anti-homomorphic functions \( f \), if \( S \) is a supermap that transforms an arbitrary unitary channel described by \( U \in \mathrm{SU}(d) \) into \( f(U)\ketbra{\psi}f(U)^\dagger \) with average fidelity \( F \), then there exists a supermap \( S^{c} \) that performs the same transformation with worst-case fidelity \( F \), and vice versa.
Similarly, for homomorphic or anti-homomorphic functions \( f \), if \( S \) is a supermap that transforms an arbitrary unitary channel \( U \in \mathrm{SU}(d) \) into \( f(U)\ketbra{\psi}f(U)^\dagger \) with average fidelity \( F \), then there exists a supermap \( S^{c} \) that transforms \( U \) into the state
\begin{equation}
    \eta \rho_{\P} * \dketbra{f(U)}{f(U)}_{\P \F} + (1 - \eta)\frac{\mathbb{1}_{\F}}{d},
\end{equation}
with fidelity \( F = \eta + \frac{1 - \eta}{d} \), and vice versa. A proof of these results is provided in App.~\ref{sec:proofs_fig_mer}.

\subsubsection{SDP formulation of the deterministic exact problem}\label{sec:SDP}

Given a function $f:\mathrm{SU}(d)\to \mathrm{SU}(d)$ and fixed quantum state $\ket{\psi}\in\mathbb{C}^d$, the problem of finding the optimal average fidelity that implements $f(U)$ on $\ket{\psi}$ reads as\footnote{In this work, all the integrals are done as a group average. That is, all integrals are done respective to the Haar measurement of a group. Unless stated otherwise, the Haar measure is assumed to be that of $\mathrm{SU}(d)$. }
\begin{align}\label{Eq:SDP_det}
\textup{max}& \expval{F^{\psi}}\coloneqq \int_{\text {Haar }} F\left(S^{\psi}_{\I\O\F}*\ChoiU_{\I\O}^{\otimes k},f(U)\ketbra{ \psi}{ \psi}_{\F}f(U)^{\dagger} \right) \mathrm{d} U. \\
   \textup{ where}& \ S^{\psi}_{\I\O\F} \textup{ is}  \textup{ a valid $k$-slot superchannel.} 
\end{align}
To recognize this as an SDP problem, it is convenient to introduce the notion of the \emph{performance operator} \( \Omega^\psi_{\I\O\F} \in \mathcal{L}(\mathcal{H}_{\I} \otimes \mathcal{H}_{\O} \otimes \mathcal{H}_{\F}) \), which depends on the target function \( f \). The concept of performance operator was introduced in Ref.~\cite{ChiribellaEbler_2016} and further developed in Ref.~\cite{Quintino_2022det} for analysing deterministic higher-order transformations on unknown input states. Similarly to the case of unknown input states, as shown in App.~\ref{append:proofs_det}, one can show that if we define
\begin{equation}\label{perfk_def}
    \Omega^{\psi}_{\I\O\F} \coloneqq \int dU \, \dketbra{\overline{U}}{\overline{U}}^{\otimes k }_{\I\O} \otimes \left( f(U) \ketbra{\psi}{\psi} f(U)^{\dagger} \right)_{\F},
\end{equation}
then the average fidelity in Eq.~\eqref{Eq:SDP_det} can be expressed as
\begin{align}
    \expval{F^{\psi}} = \Tr(S^{\psi} \Omega^\psi).
\end{align}
Hence, once the performance operator \( \Omega^\psi \) is evaluated, the optimization problem, for the case of a single input unitary $k=1$, can be formulated as the following semidefinite program (SDP):
\begin{align}
    \max_{S} \quad & \Tr(S^{\psi} \Omega^\psi) \\
    \textup{subject to} \quad & S \geq 0, \\
                         \Tr_{\F}(S) = \sigma_{\I} \otimes \mathbb{1}_{\O} \quad & \Tr(\sigma_{\I}) = 1.
\end{align}

The use of a performance operator in the context of analogous problems involving unknown input states was introduced in~\cite{ChiribellaEbler_2016,Quintino_2022det}, where it is defined as
\begin{equation}
\Omega_{\P\I\O\F} \coloneqq \frac{1}{d_{\P}^2} \int_{\textup{Haar}} \dketbra{f(U)}{f(U)}_{P F} \otimes \dketbra{\overline{U}}{\overline{U}}_{\I\O}^{\otimes k} \, \mathrm{d}U.
\end{equation}
As shown in App.~\ref{append:proofs_det}, the performance operators corresponding to the unknown and known input state scenarios are related by
\begin{equation}\label{perform_kn}
    \Omega^{\psi}_{\I\O\F} = \ketbra{\psi}{\psi}_{\P} * \Omega_{\P\I\O\F}.
\end{equation}
with substituting $d_{\P} = 1$ in the formula for the unknown input performance operator.

Let us note that the performance operator $\Omega^{\psi}$ corresponding to a known input state does not respect the same set of symmetries as the operator from the unknown-state scenario, $\Omega$, due to its bias towards the specific state $\ket{\psi}$. In the following section, we analyse in detail the similarities and differences between the symmetry structures arising in these two settings.

\subsubsection{Symmetries of the problem}\label{sec:sym_det}

Analogous to the probabilistic case discussed in Section~\ref{sec:sym_prob}, the analysis of deterministic approximate transformations also exhibits symmetries in the choice of input state and two types of covariances. We summarize these in the following theorems:

\begin{restatable}[Independence of the choice of known state]{theorem}{IndependenceThmDet} 
Given a supermap $S^{\ket{0}}_{\I\O\F}$ such that $f:\mathrm{SU}(d) \to \mathrm{SU}(d)$ is a homomorphism or antihomomorphism, and
\begin{equation}
    \expval{F}^{\ket{0}} \;=\; \int_{\textup{Haar}} \Tr\!\Big( \big(S^{\ket{0}}_{\I\O\F} * \ChoiU_{\I\O}\big)\, f(U)\ket{0}\!\bra{0}_{\F} f(U)^{\dagger} \Big)\, \mathrm{d}U,
\end{equation}
then for every $\ket{\psi} \in \mathbb{C}^d$ there exists a supermap $S^{\ket{\psi}}_{\I\O\F}$ such that
\begin{equation}
   \expval{F}^{\ket{0}} \;=\; \int_{\textup{Haar}} \Tr\!\Big( \big(S^{\ket{\psi}}_{\I\O\F} * \ChoiU_{\I\O}\big)\, f(U)\ketbra{\psi}{\psi}_{\F} f(U)^{\dagger} \Big)\, \mathrm{d}U.
\end{equation}
\end{restatable}

\begin{restatable}[Covariance properties of deterministic superchannel]{theorem}{Omegacovariancef}
\label{thm:S_det_covariance_f}
Let $f:\mathrm{SU}(d)\rightarrow \mathrm{SU}(d)$ be a homomorphism or antihomomorphism, let $\ket{\psi}\in\mathbb{C}^d$ be a fixed state, and let $S_{\I\O\F}$ be a parallel/sequential/general $k$-slot superchannel corresponding to the deterministic approximate protocol. 

Without loss in performance in average fidelity $\expval{F}$, we can restrict our attention to parallel/sequential/general superchannel $S^{c}_{\I\O\F}$ that respects the following covariance properties
\begin{itemize}
\item Input state independent covariance:  for every $U \in \mathrm{SU}(d)$,
$[S^{c}_{\I\O\F}, U_f]=0$,
with
\begin{equation}
    U_f =
\begin{cases}
U_{\I}^{\otimes k} \otimes \mathbb{1}_{\O}^{\otimes k} \otimes \overline{U}_{\F}, & \textup{if } f \textup{ is transposition},\\
\mathbb{1}_{\I}^{\otimes k} \otimes U_{\O}^{\otimes k} \otimes U_{\F}, & \textup{if } f \textup{ is conjugation},\\
U_{\I}^{\otimes k} \otimes \mathbb{1}_{\O}^{\otimes k} \otimes U_{\F}, & \textup{if } f \textup{ is inversion}.
\end{cases}
\end{equation}
    \item Input state dependent covariance:   for every $V \in \mathrm{SU}(d)$ that stabilizes the known state $\ket{\psi}$ in the way specified below,
$[S^{c}_{\I\O\F}, V_f]=0$,
where
\begin{equation}
    V_f =
\begin{cases}
\mathbb{1}_{\I}^{\otimes k}\otimes \overline{V}_{\O}^{\otimes k}\otimes \mathbb{1}_{\F} , \quad V_{\O}^T \ket{\psi}=\ket{\psi}, & \textup{if } f \textup{ is transposition},\\
V_{\I}^{\otimes k} \otimes \mathbb{1}_{\O}^{\otimes k} \otimes \mathbb{1}_{\F}, \quad \overline{V}_{\I}\ket{\psi}=\ket{\psi}, & \textup{if } f \textup{ is conjugation},\\
\mathbb{1}_{\I}^{\otimes k}\otimes \overline{V}_{\O}^{\otimes k}\otimes \mathbb{1}_{\F}, \quad V_{\O}^{-1}\ket{\psi}=\ket{\psi}, & \textup{if } f \textup{ is inversion}.
\end{cases}
\end{equation}
\end{itemize}
\end{restatable}

\subsection{HOQO acting on a part of a known pure bipartite system}\label{sec:bip}
\begin{figure}[b!]
    \centering
\includegraphics[width=0.7\textwidth]{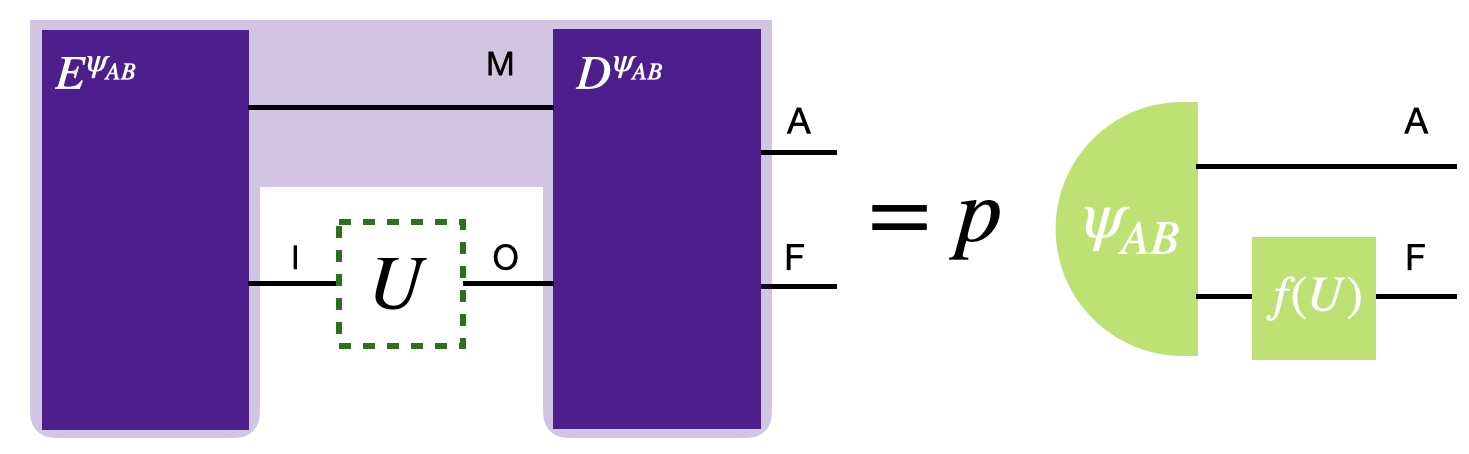}
    \caption{ General supermap protocol implementing a single call operation on a part of a known, arbitrary bipartite state.}
    \label{fig:bipart}
\end{figure}
As mentioned in the introduction, known state protocols allow us to distinguish between scenarios in which a supermap implements the desired transformation on a pure state, a bipartite state, or a mixed state. In this section, we focus on the case where the supermap, using single query $k=1$ of the input unitary $U$, applies \( f(U) \) to part of a bipartite state, $U \mapsto \id_A \otimes f(U)_{B} \ket{\psi}_{AB}\;\; \forall U$. In other words, we are considering a supermap that is processing a subsystem while preserving the entanglement with the external reference, as illustrated in Fig.~\ref{fig:bipart}. As we will see, the success probability increases with the degree of entanglement between the system on which the transformation is implemented and the external reference.

As shown in Thm.~\ref{thm:twirl_prob_constr}, when considering pure single-party states, we may, without loss of generality, restrict our attention to a fixed state such as $\ket{0}$. This is because any protocol designed for $\ket{0}$ can be converted into a protocol for an arbitrary pure state $\ket{\psi}$ without loss in performance. An analogous result holds for bipartite pure states. Namely, given a protocol that transforms an arbitrary unitary $U$ into $\mathbb{1}_{\A} \otimes f(U)_B \ket{\psi}_{AB}$ with some success probability $p$ or fidelity $F$, for a fixed bipartite pure state $\ket{\psi}_{AB}$, there exists a protocol with the same performance, i.e. the same $p$ or $F$, on any other state $\ket{\psi'}_{AB} = (V_A \otimes V_B) \ket{\psi}_{AB}$, where $V_A$ and $V_B$ are local unitaries. Hence, without loss of generality, we may restrict our analysis to bipartite pure states written in Schmidt form, i.e., $\ket{\psi}_{AB} = \sum_i \sqrt{p_i} \ket{i}_A \ket{i}_B$, where $\{p_i\}_i$ is a probability distribution. The proof is depicted in Fig.~\ref{fig:bipart_KNstateSYM}.
\begin{figure}[t!]
  \centering
  \begin{subfigure}[b!]{0.47\textwidth}
    \centering
    \includegraphics[width=\linewidth]{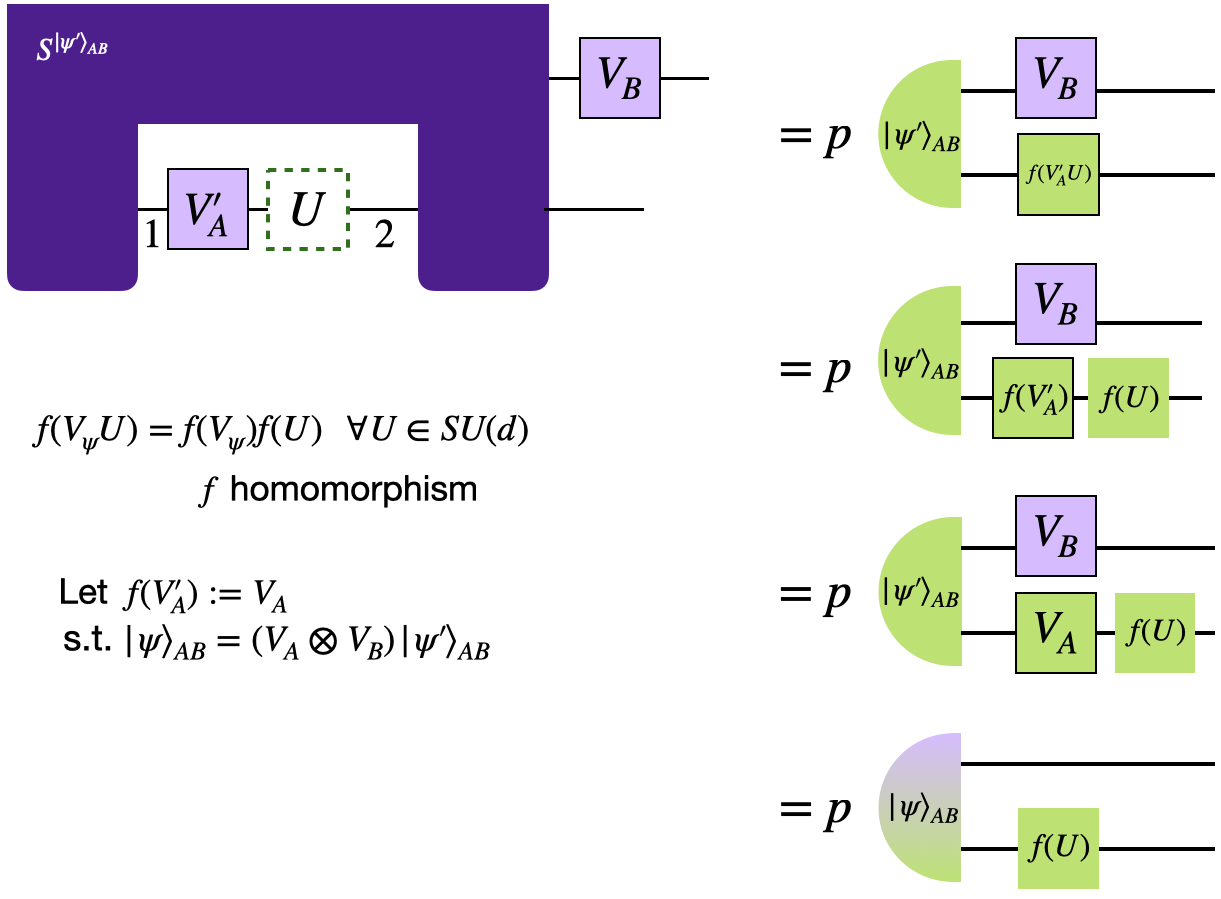}
    \captionsetup{width=1.0\textwidth}
    \caption{The figure illustrates the symmetry of the known pure bipartite protocols for implementing homomorphic functions, which allows the input state to be modified up to local unitaries $V_A$ and $V_B$.}
  \end{subfigure}
  \quad \quad 
  \begin{subfigure}[b!]{0.47\textwidth}
    \centering
    \includegraphics[width=\linewidth]{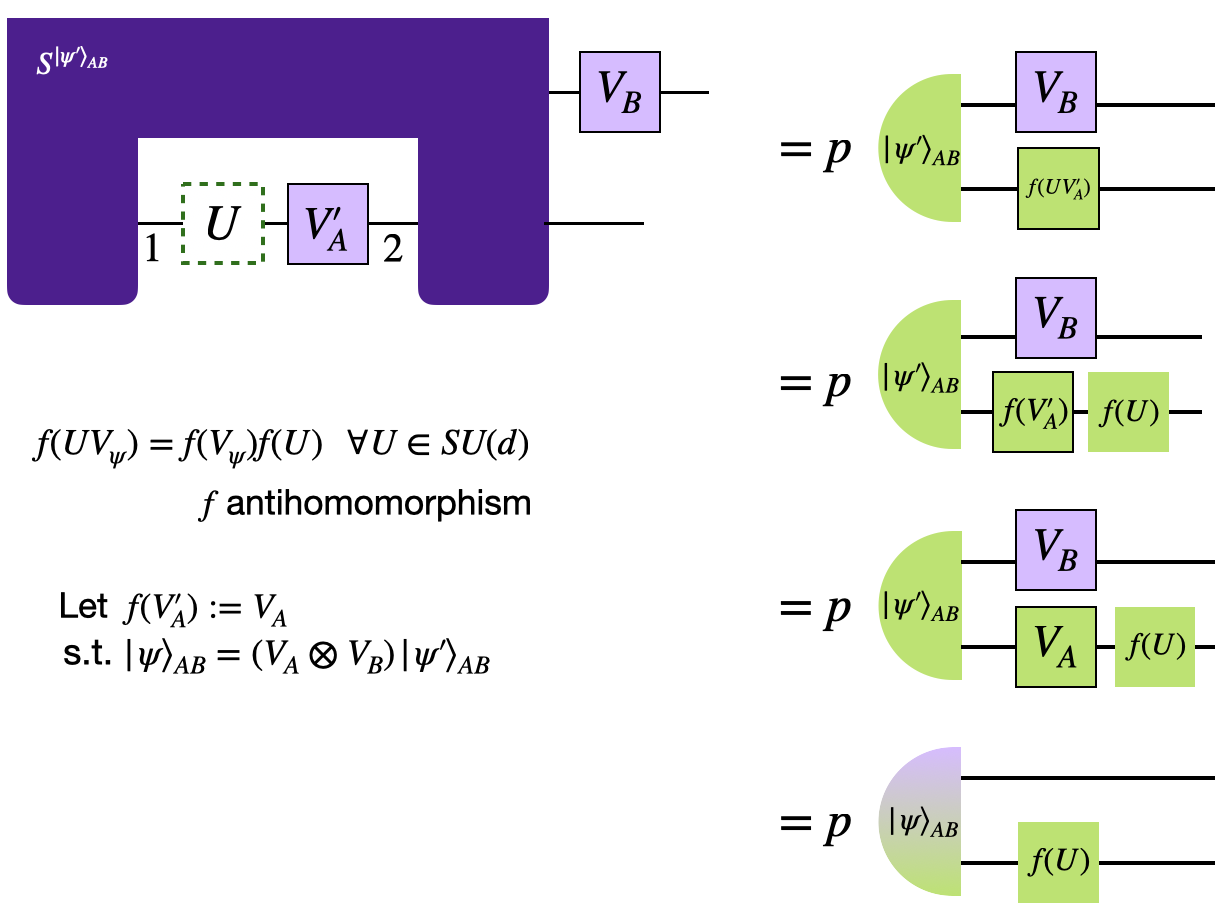}
    \captionsetup{width=0.9\textwidth}
    \caption{The figure illustrates the symmetry of the known pure bipartite protocols for implementing antihomomorphic functions, which allows the input state to be modified up to local unitaries $V_A$ and $V_B$.}
    \label{fig2_bip}
  \end{subfigure}
\captionsetup{width=0.9\textwidth}
  \caption{Determining the performance for an arbitrary known pure bipartite task implementing a homomorphic or antihomomorphic transformation of an unknown unitary $U \in \mathrm{SU}(d)$ immediately yields the equivalent solution for any other bipartite state related by local unitaries.  Elements shown in dark purple are optimized, elements in light purple are symmetries of the optimal supermap; dashed elements represent unknown and universal components; green indicates fixed behavior determined by the other elements. The gradient colour on $\ket{\psi}_{AB}$ reflects that, due to this symmetry, we can reach only those states related by local unitaries, rather than all possible states.
}
  \label{fig:bipart_KNstateSYM}
\end{figure}
\paragraph{Probabilistic exact realisation}

Probabilistic exact protocols achieving the optimal success probability $p$ can be found by solving the following optimisation problem: 
\begin{align}\label{Eq:SDP_prob_bip}
    &\textup{max} \; \; p\\
    \textup{s.t.} \; \; S_{\I\O\F\A}*\dketbra{U}{U}_{\I\O} &= p \Big(\id_A\otimes f(U)_{\F}\Big) \ketbra{\psi}{\psi}_{\A\F}  \Big(\id_A\otimes f(U)_{\F}\Big)^{\dagger} \;\;\; \forall U \in\mathrm{SU}(d)  \label{eq:inf_constraints}\\
    \Tr_{\A\F}(S_{\textup{ch}}) &= \sigma_\I \otimes \mathbb{1}_{\O},  \; \; \Tr(\sigma_\I) = 1\\
      0 \leq & S \leq S_{\textup{ch}}.  \label{Eq:SDP_prob_bip4}
\end{align}
As in the case of pure single subsystem input state, by enforcing the symmetries of the problem, the infinitely many constraints contained in Eq.~\eqref{eq:inf_constraints} may be reduce to a finite set of constraints.
More precisely, analogously to Thm.~\ref{thm:twirl_prob_constr}, in the bipartite case one can also restrict to covariant supermaps. Namely, iff $S_{\I\O\F\A}*\dketbra{U}{U}_{\I\O} = p \Big(\id_A\otimes f(U)_{\F}\Big) \ketbra{\psi}{\psi}_{\A\F}  \Big(\id_A\otimes f(U)_{\F}\Big)^{\dagger} \;\;\; \forall U \in\mathrm{SU}(d) $ then
$\map{\tau}(S)_{\I\O\F\A}*\dketbra{\id}{\id}_{\I\O} = p \Big(\id_A\otimes f(\id)_{\F}\Big) \ketbra{\psi}{\psi}_{\A\F}  \Big(\id_A\otimes f(\id)_{\F}\Big)^{\dagger}$ where \([\map{\tau}(S), U_f] = 0\), with $\map{\tau}(S) = \int U S U^{\dagger} \mathrm{d}U $ defining twirling operation.

Analogously to the covariance properties discussed in Thm.~\ref{thm:twirl_prob_constr}, we now present the corresponding symmetries for the bipartite case.

\begin{restatable}[Covariance properties of probabilistic bipartite supermap]{theorem}{SymBipThm}
\label{thm:symbp}
Let $f:\mathrm{SU}(d)\rightarrow \mathrm{SU}(d)$ be a homomorphism or antihomomorphism, let $\ket{\psi}_{\A \F}\in \mathcal{H}_{\A}\otimes \mathcal{H}_{\F}$ be a fixed state, and let $S_{\I\O\F\A}$ be a parallel/sequential/general $k$-slot superinstrument element corresponding to the probabilistic exact protocol such that
\begin{equation}
         S_{\I\O\F\A} * \ChoiU^{\otimes k} = p (\mathbb{1}_{\A} \otimes f(U)_{\F})\ketbra{\psi}{\psi }_{\A \F}(\mathbb{1}_{\A} \otimes f(U)_{\F}^{\dagger}) \quad \forall U\in\mathrm{SU}(d).
\end{equation}
That is, it transforms $k$ calls of $U$ into the state $\mathbb{1}_{\A} \otimes f(U)_{\F}\ket{\psi}_{\A\F}$ with probability $p$.

There exists a parallel/sequential/general superinstrument with an instrument element $S^{c}_{\I\O\F\A}$ such that
\begin{equation}
         S^{c}_{\I\O\F\A} * \ChoiU^{\otimes k} = p (\mathbb{1}_{\A} \otimes f(U)_{\F})\ketbra{\psi}{\psi }_{\A \F} (\mathbb{1}_{\A} \otimes f(U)_{\F}^{\dagger}) \quad \forall U\in\mathrm{SU}(d)
\end{equation}
with the same probability $p$, where $S^{c}_{\I\O\F\A}$ is covariant

Without loss in performance in $p$, we can restrict our attention to parallel/sequential/general superchannel $S^{c}_{\I\O\F}$ that respects the following covariance properties
\begin{itemize}
\item Input state independent covariance:  For every $U \in \mathrm{SU}(d)$,
$[S^{c}_{\I \O \F \A}, U_f]=0$,
with
\begin{equation}
    U_f = 
\begin{cases}
U_{\I}^{\otimes k} \otimes \mathbb{1}_{\O}^{\otimes k} \otimes \overline{U}_{\F}\otimes \mathbb{1}_{\A}, & \textup{if } f \textup{ is transposition},\\
\mathbb{1}_{\I}^{\otimes k} \otimes U_{\O}^{\otimes k} \otimes U_{\F}\otimes \mathbb{1}_{\A}, & \textup{if } f \textup{ is conjugation},\\
U_{\I}^{\otimes k} \otimes \mathbb{1}_{\O}^{\otimes k} \otimes U_{\F}\otimes \mathbb{1}_{\A}, & \textup{if } f \textup{ is inversion}.
\end{cases}
\end{equation}
    \item Input state dependent covariance:   For every $V \in \mathrm{SU}(d)$ that stabilizes the known state $\ket{\psi}$ in the way specified below,
it respects $[S^{c}_{\I \O \F \A}, B_f]=0$ for every $B$, where \( B_f \) is defined as follows:
\begin{itemize}
  \item For transposition: \( \mathbb{1}_{\I}^{\otimes k}\otimes \overline{B}_{\O}^{\otimes k}\otimes \mathbb{1}_{\F} \otimes \mathbb{1}_{\A} \), with $B\in \mathrm{SU}(d)$ and \( (\mathbb{1}_{1} \otimes B_{2}^T) \ket{\psi}_{12} = \ket{\psi}_{12} \),
  \item For conjugation: \( B_f = B_{\I}^{\otimes k} \otimes \mathbb{1}^{\otimes k}_{\O} \otimes \mathbb{1}_{\F}\otimes \mathbb{1}_{\A} \), with $B\in \mathrm{SU}(d)$ and  \((\mathbb{1}_{1} \otimes  \overline{B}_{2}) \ket{\psi}_{12} = \ket{\psi}_{12} \),
  \item For inversion: \( \mathbb{1}_{\I}^{\otimes k}\otimes \overline{B}_{\O}^{\otimes k}\otimes \mathbb{1}_{\F}\otimes \mathbb{1}_{\A} \), with $B\in \mathrm{SU}(d)$ and  \( (\mathbb{1}_{1} \otimes B_{2}^{-1}) \ket{\psi}_{12} = \ket{\psi}_{12} \). 
\end{itemize}
\end{itemize}
\end{restatable}

\paragraph{Deterministic approximate realisation}

The optimal deterministic approximate protocol can be found by solving the SDP given in Eq.~\eqref{Eq:SDP_det}, where the fidelity is defined in terms of the following performance operator.
\begin{restatable}{theorem}{BipDet}\label{thm:performance}
    Given a single call to the unitary operation $U \in \mathrm{SU}(d)$, 
consider a supermap acting on part of a known pure bipartite input state $\ket{\psi}_{A\F} \in \mathcal{H}_{A} \otimes \mathcal{H}_{\F}$, realising the transformation $U \mapsto (\mathbb{1}_{\A} \otimes f(U)_{\F})\ketbra{\psi}{\psi}_{A\F}(f(U)_{\F} \otimes \mathbb{1}_{\A} )^{\dagger}$. 
The corresponding performance operator in the known input state scenario is
\begin{equation}
\begin{aligned}
\Omega^{\psi}_{A\I\O\F} \coloneqq& \int dU \dketbra{\overline{U}}{\overline{U}}^{\otimes k }_{\I\O} \otimes ( f(U)_{\F} \ketbra{\psi}{\psi}_{\F\A}f(U)^{\dagger}_{\F}) \\
    =&  \ketbra{\psi}{\psi}_{\P\A} * \Omega_{\P\I\O\F},
\end{aligned}
\end{equation}
where $\Omega_{\P\I\O\F}$ is the performance operator for the function $f$ on the unknown input state case~\cite{Quintino_2022det}. 

With the performance operator defined above, the average fidelity is given by,
\begin{equation}
  \expval{F^{\psi}} = \Tr(S^{\psi}_{\I\O\F\A}\Omega^{\psi}_{\I\O\F\A}).
\end{equation}

In particular, when there is no auxiliary system $\mathcal{H}_A$,  or equivalently, when $\mathcal{H}_A$ is one dimensional, we have that $\Omega^{\psi}_{\I\O\F} = \ketbra{\psi}{\psi}_{\P} * \Omega_{\P\I\O\F}$.
\end{restatable}

Analogously to the covariance properties discussed in Thm.~\ref{thm:S_det_covariance_f}, we now present the corresponding symmetries for the bipartite case.

\begin{restatable}[Covariance properties of deterministic bipartite superchannel]{theorem}{BipSymDet}
\label{thm:S_det_covariance_bip}
Let $f:\mathrm{SU}(d)\rightarrow \mathrm{SU}(d)$ be a homomorphism or antihomomorphism, let $\ket{\psi}\in\mathcal{H}_{\A}\otimes \mathcal{H}_{\F}$ be a fixed state, and let $S_{\I\O\F\A}$ be a parallel/sequential/general $k$-slot superchannel corresponding to the deterministic approximate protocol. 

Without loss in performance in $\expval{F}$, we can restrict our attention to parallel/sequential/general superchannel $S^{c}_{\I\O\F\A}$ that respects the following covariance properties
\begin{itemize}
\item Input state independent covariance:  For every $U \in \mathrm{SU}(d)$,
$[S^{c}_{\I\O\F\A}, U_f]=0$,
with
\begin{equation}
    U_f =
\begin{cases}
U_{\I}^{\otimes k} \otimes \mathbb{1}_{\O}^{\otimes k} \otimes \overline{U}_{\F} \otimes \mathbb{1}_{\A}, & \textup{if } f \textup{ is transposition},\\
\mathbb{1}_{\I}^{\otimes k} \otimes U_{\O}^{\otimes k} \otimes U_{\F} \otimes \mathbb{1}_{\A}, & \textup{if } f \textup{ is conjugation},\\
U_{\I}^{\otimes k} \otimes \mathbb{1}_{\O}^{\otimes k} \otimes U_{\F} \otimes \mathbb{1}_{\A}, & \textup{if } f \textup{ is inversion}.
\end{cases}
\end{equation}
    \item Input state dependent covariance:   For every $V \in \mathrm{SU}(d)$ that stabilizes the known state $\ket{\psi}$ in the way specified below,
it respects $[S^{c}_{\I \O \F \A}, B_f]=0$ for every $B$, where \( B_f \) is defined as follows:
\begin{itemize}
  \item For transposition: \( \mathbb{1}_{\I}^{\otimes k}\otimes \overline{B}_{\O}^{\otimes k}\otimes \mathbb{1}_{\F} \otimes \mathbb{1}_{\A} \), with $B\in \mathrm{SU}(d)$ and \( (\mathbb{1}_{1} \otimes B_{2}^T) \ket{\psi}_{12} = \ket{\psi}_{12} \),
  \item For conjugation: \( B_f = B_{\I}^{\otimes k} \otimes \mathbb{1}^{\otimes k}_{\O} \otimes \mathbb{1}_{\F} \otimes \mathbb{1}_{\A}\), with $B\in \mathrm{SU}(d)$ and  \((\mathbb{1}_{1} \otimes  \overline{B}_{2}) \ket{\psi}_{12} = \ket{\psi}_{12} \),
  \item For inversion: \( \mathbb{1}_{\I}^{\otimes k}\otimes \overline{B}_{\O}^{\otimes k}\otimes \mathbb{1}_{\F}\otimes \mathbb{1}_{\A} \), with $B\in \mathrm{SU}(d)$ and  \( (\mathbb{1}_{1} \otimes B_{2}^{-1}) \ket{\psi}_{12} = \ket{\psi}_{12} \). 
\end{itemize}
\end{itemize}
\end{restatable}

The proofs of these theorems are presented in App.~\ref{append:bipart_proofs}.

\subsection{HOQO acting on a mixed input state}\label{sec:mixed}

In this section, we discuss a higher-order quantum operations acting on a known mixed input state, i.e. $U \mapsto f(U)\rho f(U)^{\dagger}$, for $\forall U \in \mathrm{SU}(d)$. 

\paragraph{ Comparison with the bipartite input state case}
There exists a straightforward protocol for implementing a transformation on a known mixed state derived from the known bipartite-state protocol. As illustrated in Fig.~\ref{fig:bipmix}, one can simply trace out one of the output subsystems from the bipartite task, and this will be a feasible protocol with the same sucess probability $p$.
While this approach indeed provides a valid mixed-state supermap construction, it is generally suboptimal. In particular, given $\ketbra{\psi}{\psi}_{\A\F}$ and $S^{\ketbra{\psi}{\psi}_{\A\F}} $, which implement a function $f$ with success probability $p$ for a single-query $U \in \mathrm{SU}(d)$, we can construct $S^{\rho}$ for the mixed state $\rho_{\F} = \Tr_{\A}( \ketbra{\psi}{\psi}_{\A\F})$, achieving the same function $f$ for any $U \in \mathrm{SU}(d)$, with the same probability $p$.  However, for functions $f: \mathrm{SU}(d) \rightarrow \mathrm{SU}(d)$  such as transposition, conjugation, or inversion, this approach often does not yield an optimal protocol, since for various states  $\ketbra{\psi}{\psi}_{AB}$ with $\rho_A=\tr_B(\ketbra{\psi}_{AB})$, we have
\begin{equation}\label{bipmixinqv}
    p^{\textup{max}}(\rho_{A}) >  p^{\textup{max}}(\ketbra{\psi}{\psi}_{AB}). 
\end{equation}

\begin{figure}[t!]
    \centering
\includegraphics[width=0.7\textwidth]{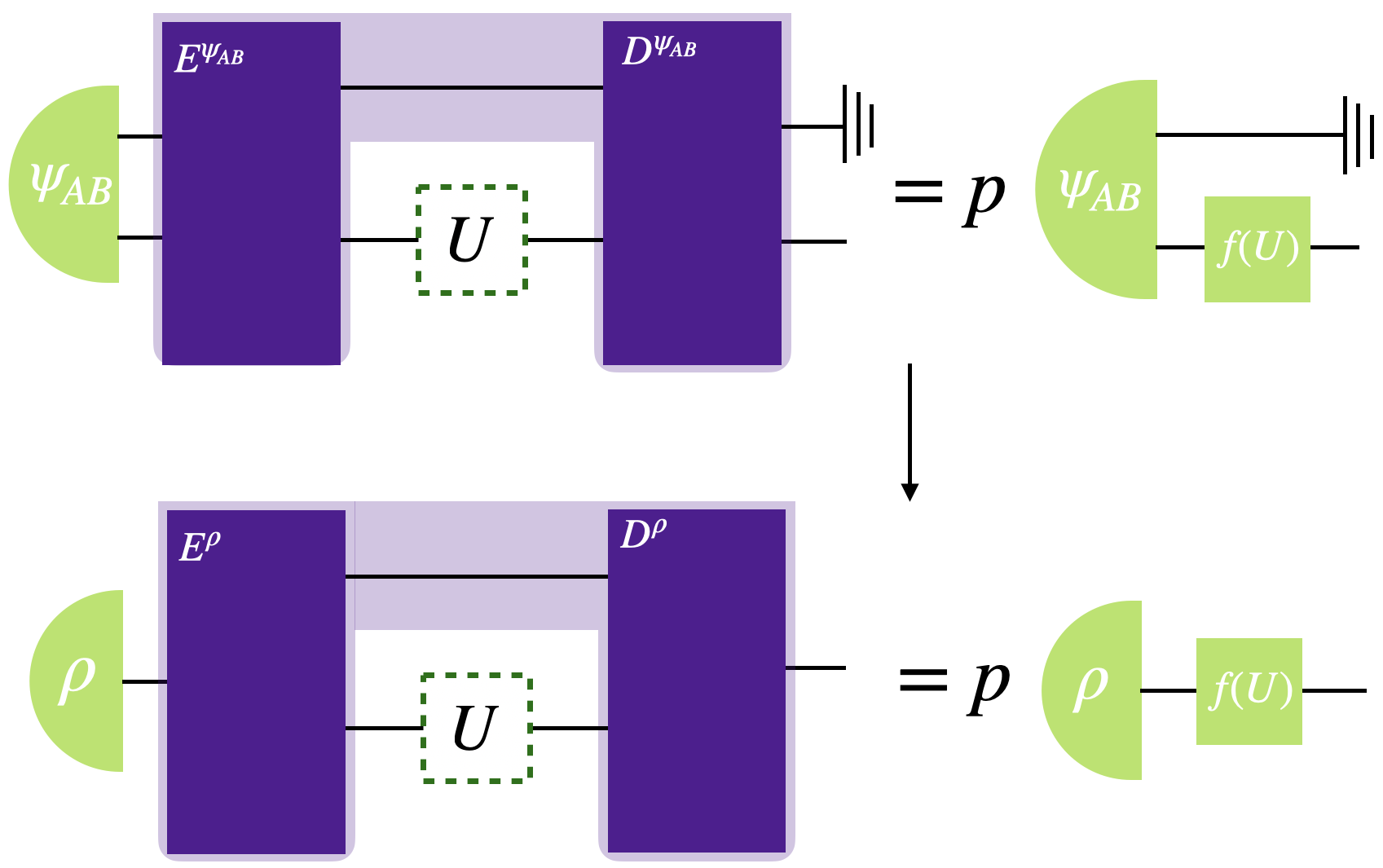}
    \caption{ A supermap protocol acting on a known-mixed-state can be constructed from a corresponding bipartite-state protocol; however, this construction does not yield an optimal implementation for the mixed-state case. In conclusion, the connection between the two problems is nontrivial.}
    \label{fig:bipmix}
\end{figure}

For example, as shown in Thm.~\ref{thm:Mix_visib}, when \( d > 2 \), deterministic and exact unitary conjugation is possible for mixed states of the form \( \rho = \eta \ketbra{\psi} + (1 - \eta)\frac{\id}{d} \) whenever \( \eta \leq \frac{2}{d+1} \). However, for the same dimensions \( d > 2 \), the maximal success probability for probabilistic exact unitary conjugation on pure bipartite states is zero. Consequently, we may encounter situations where \( p^{\textup{max}}(\rho_A) = 1 \), yet for a purification \( \ketbra{\psi}{\psi}_{AB} \), we have \( p^{\textup{max}}(\ketbra{\psi}{\psi}_{AB}) = 0 \). Analogous conclusions hold for unitary transposition.
In particular, for any dimension \( d\geq2 \), and for the state \( \rho = \eta \ketbra{\psi} + (1 - \eta)\frac{\id}{d} \), exact unitary transposition is achievable with \( p^{\textup{max}}(\rho_A) = 1 \) when \( \eta \leq \frac{d}{d^2 + 1} \), as stated in Eq.~\eqref{Eq:trans_det_exact} and proven in Thm.~\ref{thm:Mix_visib}. However, except for the maximally entangled state, all pure bipartite states satisfy the strict inequality \( p^{\textup{max}}(\ketbra{\psi}{\psi}_{AB}) < 1 \) (see Thm.~\ref{bip_trans_thm}). In conclusion, construction via bipartite protocol represent a substantially smaller feasible subset compared to the set attaining probabilities using a general mixed state supermap.

Before concluding this section, we note that, for the same reasons illustrated in Fig.~\ref{fig:bipart_KNstateSYM} and Fig.~\ref{fig:symmetries}, without loss of generality, one may choose any representative $\rho$ from the equivalence class $[\rho] \coloneqq \{ U \rho U^\dagger \mid U \textup{ is unitary} \}$. In other words, we may always assume $\rho$ to be diagonal in the computational basis, and the performance then depends only on the eigenvalues of $\rho$.

\paragraph{Deterministic and exact realisation}
While the maximally mixed state represents a trivial case, the noise present in a mixed state can serve as a resource for achieving a deterministic and exact protocol—even for certain nontrivial states. We now identify the minimal amount of noise needed to still be able to achieve deterministic and exact implementation. 

In order to do so, we will consider the set of mixed states,
\begin{equation}
    \rho = \eta \ketbra{\psi}{\psi} + \frac{1 - \eta}{d} \mathbb{1},
\end{equation}
that is, convex combinations of a pure state $\ket{\psi}$ with the maximally mixed state. Notice that for qubits, this is not a restriction\footnote{Every qubit density matrix can be written in the form \( \rho = \eta \ketbra{\psi} + (1 - \eta)\frac{\id}{2} \), for some pure state \( \ket{\psi} \) and some \( \eta \in [0,1] \). This fact can also be visualized geometrically on the Bloch sphere, as every point inside the sphere can be expressed as a convex combination of an extremal point (on the surface) and the center.}.
Using Thm.~\ref{thm:Mix_visib}, it is possible to construct a protocol that achieves a \emph{deterministic} and \emph{exact} implementation of a supermap. This is done by identifying the values of $\eta \leq \eta^* $, where $ \eta^* $ denotes the critical visibility bellow which deterministic and exact implementation becomes feasible.
To see this, let us begin by introducing the following lemma, similar to the one discussed in Sec.~$3.1.3$ of Ref.~\cite{Quintino_2022det}, but adapted here for the case of known input state scenario.

\begin{restatable}{lemma}{LemmaMixVis}\label{lemma:mix_vis}

Let $\rho=\ketbra{\psi} \in \mathcal{L}(\mathbb{C}^d)$ be a pure quantum state.
If $S$ is a superchannel transforming $k$ uses of $U$ into $f(U)\rho f(U)^{\dagger}$ with average fidelity $\expval{F^{\psi}}$, there exists a superchannel $S^{\prime}$ such that

\begin{equation}
    S_{\I \O \F}^{\prime} *\dketbra{U}{U}_{\I\O} ^{\otimes k}=\eta \rho_{\P} * \dketbra{f(U)}{f(U)}_{P F}+(1-\eta) \frac{\mathbb{1}_{\F}}{d},
\end{equation}
where $\expval{F^{\psi}}=\eta +\frac{1-\eta}{d}$.    
\end{restatable}

The above Lemma then provides a recipe to transform an arbitrary unitary operation into the state $\eta \rho_{\P} * \dketbra{f(U)}{f(U)}_{P F}+(1-\eta) \frac{\mathbb{1}_{\F}}{d}$. Now, what if we consider the task of an arbitrary unitary operation into $f(U)\rho f(U)^\dagger$ with $\rho$ being of the form of $\rho=\eta \rho (1-\eta) \frac{\mathbb{1}}{d}$. Since we know the state $\rho$, we can deterministically purify it and then apply the covariant supermap $S^{\prime}$ from Lemma~\ref{lemma:mix_vis} in $\ketbra{\psi}$ instead of applying it on the mixed state \( \rho = \eta \ketbra{\psi}{\psi} + \frac{1-\eta}{d} \mathbb{1} \). That, in turn perfectly implements the action of  $f(U)$ on a class of states of the form $\rho = \eta \ketbra{\psi}{\psi} + \frac{1-\eta}{d}\mathbb{1}$, and corresponds to implementation of the mixed state supermap $S^{\prime}$. This can be summarized in the following theorem:

\begin{restatable}{theorem}{ThmMixVisib}\label{thm:Mix_visib}
    Let $\rho$ be arbitrary, known state of the form $\rho = \eta \ketbra{\psi}{\psi} + \frac{1-\eta}{d}\mathbb{1}$. There exists a superchannel \( S \) that achieves a \emph{deterministic} and \emph{exact} implementation of \( f(U)\rho f(U)^{\dagger} \) for visibility values \( \eta \leq \eta^* \), where 
    \begin{equation}
        \eta^* = \frac{d\expval{F^{\psi}} - 1}{d-1}
    \end{equation}
    and $\expval{F^{\psi}}$ denotes the optimal average fidelity of the corresponding known pure-state protocol. 
\end{restatable}

\section{Advantages of known state HOQC in specific unitary transformation tasks}\label{sec:tasks_applications}

Now we will explore some of the applications of the formalism introduced in the sections so far. In particular, we will consider protocols that implement transposition, storage-and-retrieval, conjugation and inversion, on a pure state, part of the bipartite state and mixed state, using a single-query to the input unitary.

\subsection{Unitary transposition}
Let us explore the task that implements universal unitary transposition on a known input state, that is, the case where $f(U) = U^{T}$. 

\paragraph{Optimal probabilistic exact transposition protocol on a pure known unentangled state}\label{parag:prob_trans_single}

In Thm.~\ref{thm:trans_single} stated below, we will prove that the optimal circuit structure for implementing transposition operation $U \mapsto U^T\ketbra{\psi}{\psi}\overline{U}$ is given in terms of a superinstrument in which the superinstrument element corresponding to success has the Choi operator given by 
\begin{align}
    S_{\I\O\F, \textup{trans}}^{\textup{opt}} = \ketbra{\phi^{+}}{\phi^{+}}_{\I\M} * \left(\ketbra{\psi}{\psi}_{\O} \otimes \dketbra{\mathbb{1}}{\mathbb{1}}_{\M\F}\right).
\end{align}
The quantum circuit corresponding to this superinstrument element is presented in Fig.~\ref{fig:transpositionKNpure}. This circuit structure resembles remote state preparation (RSP) \cite{Bennett_2001}, except with the difference that in this case we will have the action of the unitary as an additional step. Namely, we can imagine Alice and Bob sharing a maximally entangled state. Alice has complete knowledge of the state $\ketbra{\psi}{\psi}$ and the access to the unknown unitary $U$. This enables her to prepare $U^{T} \ketbra{\psi}{\psi} \overline{U}$ on Bob's side by querying the unitary $U$ with her part of maximally entangled state, after which she performs the $M_{\psi}= \ketbra{\psi}{\psi}^{T}$ measurement. 
\begin{figure}[t!]
    \centering
\includegraphics[width=0.7\textwidth]{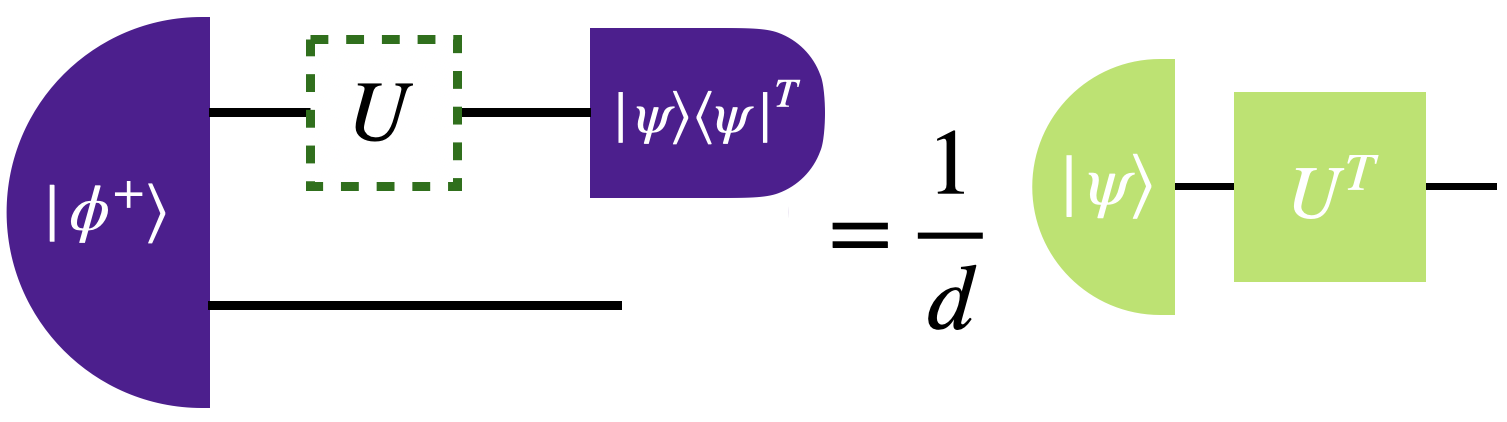}
    \caption{Optimal protocol for universal unitary transposition applied to a known arbitrary input state. The probability of successful implementation is $p = \frac{1}{d}$.}
\label{fig:transpositionKNpure}
\end{figure}
One can also have clear understanding of why this protocol works from the circuit diagram perspective. Following the Fig.~\ref{fig:transpositionKNpure}, we can imagine circuit elements getting transposed when going over the maximally entangled state, as the equation $(A \otimes \mathbb{1})\ket{\phi^{+}} = ( \mathbb{1}\otimes A^T)\ket{\phi^{+}}$ holds, where $A$ corresponds to arbitrary matrix. Alternatively, we could imagine a protocol where Bob has access to a single query of the unitary instead of Alice. In this case, we would need to implement a SWAP gate after the unitary.

Let us remark that the encoder part is the same as for the optimal unknown case of transposition, $E = \ketbra{\phi^{+}}{\phi^{+}}$ (cf. \cite{Quintino2019prob}) and it can be understood as the resource for the optimal unitary storing. Surprisingly, although both the encoder and the decoder could, in principle, exploit the information about the known state, the optimal encoder for transposition does not depend on the known state.
However, the decoder $D=\ketbra{\psi}{\psi}_{\O} \otimes \dketbra{\mathbb{1}}{\mathbb{1}}_{\M\F}$ does, compared to the unknown case, where the Choi operator of the optimal supermap was $S_{\P\I\O\F}= \ketbra{\phi^{+}}{\phi^{+}}_{\I\M} * \left(\ketbra{\phi^{+}}{\phi^{+}}_{\P\O} \otimes \dketbra{\mathbb{1}}{\mathbb{1}}_{\M\F}\right)$, and had a probability of $1/d^2$~\cite{Quintino2019prob}. 

\begin{restatable}{theorem}{ThmProbTransSingle}\label{thm:trans_single}
    Given a single call of the $d$-dimensional unitary operation $U \in \mathrm{SU}(d)$, in the case of a supermap realised on a known pure input state $\ketbra{\psi}{\psi}$, the transposition $f(U) = U^{T}$ can be implemented with a maximal success probability\footnote{For comparison, we denote the values of the corresponding optimal universal state protocol in light gray.}
\begin{equation}\label{p_T}
    p^{\psi}_{\textup{max,trans}} = \frac{1}{d} \quad \quad \quad \textcolor{gray}{ \textup{cf. } \quad p_{\textup{max,trans}} = \frac{1}{d^2}}
\end{equation}
\end{restatable}

Compared to the unknown state protocol, where the success probability is $p_{\textup{max,trans}} = \frac{1}{d^2}$, we see that knowing the state leads to a quadratic improvement. This can be easily seen from the circuit representation: each maximally entangled state contributes a factor of $\frac{1}{d}$ to the probability. In the unknown state case, the decoder required a Bell measurement, introducing an additional $\frac{1}{d}$ factor and yielding an overall probability of $\frac{1}{d^2}$. In contrast, when the state is known, we can directly use the knowledge of the state to design the measurement, achieving the desired result with probability $\frac{1}{d}$.

We will encounter this protocol again in the following paragraph, where the superchannel for the deterministic task will correspond to the sum of the probabilistic task and the corresponding correction. In general, an optimal deterministic protocol cannot be derived from a probabilistic strategy. However, in the case of known state transposition, this becomes possible due to the unique symmetry of the problem.

\paragraph{Optimal deterministic approximate transposition protocol on a pure known unentangled state}\label{sec:transp_det_single}
In Thm.~\ref{thm:trans_det}, we will prove that the optimal strategy for implementing the transposition operation $U \mapsto U^T\ketbra{\psi}{\psi}\overline{U}$ on a known state in a deterministic approximate way is given by a superchannel with Choi operator
\begin{align}\label{sup_trans_pur_det}
S_{\I\O\F, \textup{trans}}^{\textup{opt}} = \ketbra{\phi^{+}}{\phi^{+}}_{\I\M} * \left(\ketbra{\psi}{\psi}_{\O} \otimes \dketbra{\mathbb{1}}{\mathbb{1}}_{\M\F} 
+  \frac{(\mathbb{1}_{\O} - \ketbra{\psi}{\psi}_{\O})\otimes (\mathbb{1}_{\I\F} -   \ketbra{\phi^{+}}{\phi^{+}}_{\I\F})}{d^2 - 1} \right).
\end{align}
The quantum circuit corresponding to this superchannel is presented in Fig.~\ref{Fig:TransDetpure}.  
\begin{figure}[b!]
    \centering
\includegraphics[width=0.9\textwidth]{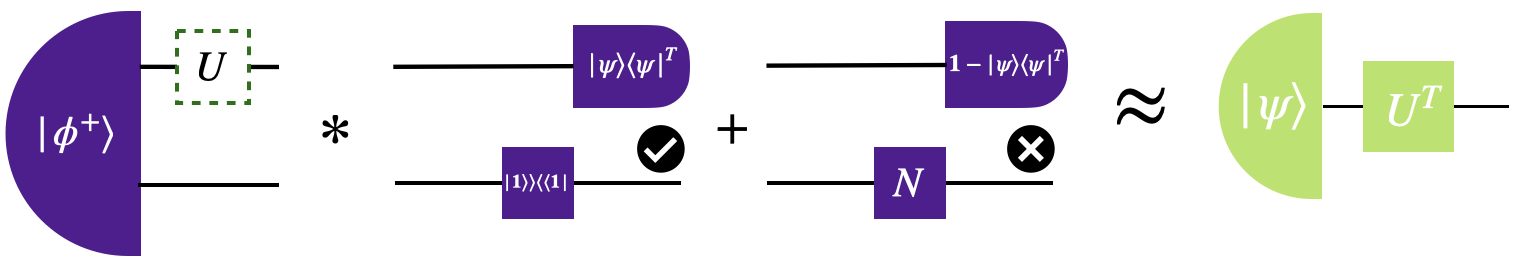}
    \caption{Optimal protocol for implementing the transposition operation in a deterministic and approximate manner on an arbitrary known pure state. Here, “$*$” denotes the linking operation. The protocol consists of a success branch ($\checkmark$) and a failure branch ($\times$). On the failure branch, we perform the optimal CPTP approximation of the reduction map (which corresponds to the optimal universal NOT for qubits, as correction channel.
}
    \label{Fig:TransDetpure}
\end{figure}

Let us now observe that the superchannel of Eq.~\eqref{sup_trans_pur_det} can be constructed via the optimal probabilistic protocol introduced in the previous section. As explained in Sec.~\ref{sec:probabilistic}, probabilistic protocols are described by superinstruments that split into a success branch \( \smap{S}_{\checkmark} \) and a failure branch \( \smap{S}_{\times} \). When analysing the probabilistic exact case, since exactness is required, the failure branch is discarded entirely, resulting in a protocol that succeeds only with a certain probability.

A natural question then arises: is it possible to retain the failure branch and attempt to correct the outcome, similar to the correction step in standard quantum teleportation?
Such a construction leads to an approximate realisation of the protocol, since the lack of knowledge of the unitary prevents perfect correction. The fidelity of such protocol simply corresponds to:
\begin{equation}\label{Eq:fid_decomp}
    F_{\textup{tot}} = p_{s}\cdot 1 + (1- p_{s})F_{\textup{fail}},
\end{equation}
where $p_a$ is the probability of performing the exact transformation, hence, $F_\textup{success}=1$.

While the intuitive construction involving a success branch with unit fidelity and a failure branch provides a valid perspective on deterministic approximate realisations, it does not capture the most general class of such protocols. Not all deterministic strategies decompose in this way. Nonetheless, the superchannel in Eq.~\eqref{sup_trans_pur_det} does exhibit this structure and corresponds to the fidelity expression given in Eq.~\eqref{Eq:fid_decomp}. That this strategy turns out to be optimal in this case is a peculiar feature of the transposition function within the known input state setting. In contrast, for unknown input states, the optimal deterministic implementation of unitary transposition does \emph{not} generally fall within this class of strategies~\cite{Quintino_2022det}.

Let us now analyse more closely why the strategy based on a probabilistic approach combined with a correction in the failure branch performs so well in the known input state transposition task. When the protocol succeeds, it does so with a probability of \(p_s = \frac{1}{d}\).
Upon failure, the resulting (unnormalised) state is \(U^T (\id - \ketbra{\psi}{\psi}) \overline{U}\), whereas our goal is \(U^T \ketbra{\psi}{\psi} \overline{U}\). To achieve this, we may introduce a \emph{correction channel}. A desirable property of such a correction is that, first, it commutes with the unknown unitary, and second, it effectively flips the state \(\id - \ketbra{\psi}{\psi}\mapsto \ketbra{\psi}{\psi}\). That is, firstly, the desired channel \(\mathcal{D}\) commutes with the arbitrary unitary operator $U^T$. That is, it is desirable that the correction channel satisfies the covariance property%
\footnote{In principle, full unitary covariance may not be necessary; that is, we may impose that $ \map{D}(U \rho U^{\dagger}) = U \map{D}(\rho) U^{\dagger} \quad \forall U\in \mathrm{SU}(d)$ only on the known fixed state $\rho=\ketbra{\psi}$. But, as we will prove, the optimal performance is attainable with a correction channel that is covariant on all input states.}
\begin{equation}
    \map{D}(U \rho U^{\dagger}) = U \map{D}(\rho) U^{\dagger} \quad \forall U\in \mathrm{SU}(d), \; \forall \rho\in\L(\mathbb{C}^d).
\end{equation}
The unique channel, up to mixing parameter $\eta$, satisfying this condition for $\forall U$ is the depolarising channel\footnote{We remark that Refs.~\cite{Keyl_2002,Gschwendtner_2021Covariant} also imposes that $\eta\geq0$, however, the linear map $\mathcal{D}_\eta$ is CPTP in the range $\eta\in \left[-\frac{1}{d^2 - 1}, 1\right].$ }~\cite{Keyl_2002,Gschwendtner_2021Covariant},
\begin{equation}
    \mathcal{D_\eta}(\rho) = \eta \rho + (1 - \eta) \frac{\mathbb{1}}{d},
\end{equation}
which is unital, trace-preserving, and covariant under the unitary group. Here, $\eta$ is visibility factor that describes the white noise added to the state $\rho$. 

As previously mentioned, a second desired property is the implementation of a flip operation, mapping \(\id - \ketbra{\psi}{\psi} \mapsto \ketbra{\psi}{\psi}\)
. While the depolarising channel does not achieve this transformation exactly, it does approximate it in the regime of negative visibility, \(\eta < 0\). For \(d = 2\), negative visibility corresponds to a ``flip'' of the Bloch vector, accompanied by a reduction in its length\footnote{As is well known, a perfect universal flip, i.e., a universal NOT operation on arbitrary mixed states, is not physically achievable~\cite{Buzek_NOT_1999} in the exact sense. Instead, the universal NOT channel~\eqref{eq:NOTgate} offers the best possible approximation, flipping the Bloch vector direction while simultaneously contracting it.}. 

Therefore, in the optimal case corresponding to the superchannel described in Eq.~\eqref{sup_trans_pur_det}, the correction channel is described by a depolarising channel with the smallest visibility $\eta$ such that the map $\map{D_\eta}$ is completely positive. Hence, we set the negative visibility
\begin{equation}\label{Eq:negvis}
    \eta = -\frac{1}{d^2 - 1},
\end{equation}
which corresponds to the optimal CPTP approximation of the reduction map~\cite{Horodecki1998Reduction} 
\begin{align}\label{eq:NOTgate}
    \map{N}(\rho) = \frac{d\Tr(\rho)\, \mathbb{1} - \rho}{d^2 - 1},
\end{align}
which, for the case of qubits, corresponds to the optimal completely positive approximation of the universal NOT \cite{Buzek_NOT_1999}. 

In summary, the optimal superchannel presented in Eq.~\eqref{sup_trans_pur_det} corresponds to a strategy that accepts the success outcome of the probabilistic protocol while correcting the failure branch by commuting the universal NOT gate with the unknown unitary, as illustrated in Fig.~\ref{Fig:TransDetpure}.

\begin{restatable}{theorem}{FidTransPure}\label{thm:trans_det}
Given a single call of the $d-$dimensional unitary operation $U \in \mathrm{SU}(d)$, in the case of a supermap realised on a known pure input state $\ket{\psi}\in\mathbb{C}^d$, the transposition operation $U^{T}$ can be implemented, achieving optimal average fidelity 
    \begin{equation}
    \expval{F^{\psi}}_{\textup{trans}, k=1} = \frac{2d + 1}{d^2 + d} \quad \quad \quad \quad \textcolor{gray}{ \textup{cf. } \quad \expval{F}_{\textup{trans}, k=1} = \frac{2}{d^2}}
\end{equation}
\end{restatable}

As in probabilistic case, we see quadratic improvement compared to the unknown input state case.

\paragraph{Transposition on a part of a known bipartite state: probabilistic exact}

\begin{figure}[b!]
    \centering
\includegraphics[width=0.7\textwidth]{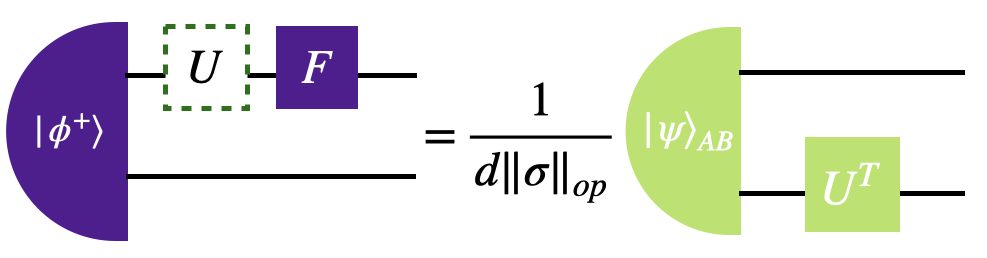}
    \caption{Protocol implementing optimal transposition of an unknown unitary acting on part of a known, arbitrary bipartite state. The success probability is given in terms of the maximal eigenvalue of the reduced density matrix $\sigma = \Tr_A(\ketbra{\psi}{\psi}_{AB})$, and is given by $p= \frac{1}{d\norm{\sigma}_{\textup{op}}}$. }
    \label{fig:bipartTrans}
\end{figure}

We now consider the task of implementing the transposition $U \mapsto (U_{\F}^T\otimes \mathbb{1}_{\M})\ketbra{\psi}{\psi}_{\F\M}(\overline{U}_{\F}\otimes \mathbb{1}_{\M})$ on a part of a bipartite state. The optimal supermap is given by
\begin{align}\label{SDP_bip}
     S^{\textup{opt}}_{\I\O\F\A, trans} = \frac{1}{\norm{\Tr_{\A}(\ketbra{\psi}{\psi}_{AB} )}_{\textup{op}}}\ketbra{\psi}{\psi}_{\O\A} \otimes \ketbra{\phi^{+}}{\phi^{+}} _{\I\F}
\end{align}
as depicted in Fig.~\ref{fig:bipartTrans}.

To understand this protocol operationally, we introduce the concept of a quantum filter, which is an instrument element with a single Kraus operator. A quantum filter $F\in\mathcal{L}(\mathcal{H})$ is the Kraus operator, corresponding to the instrument element $\map{F}$ where $\map{F}(\cdot) =  F(\cdot) F^\dagger$. Since $\map{F}$ is an instrument element, it must be trace non-increasing, hence its Kraus operator $F$ must respect $ F F^\dagger \leq \mathbb{1}$. Quantum filter can be used to convert a maximally entangled state into an arbitrary bipartite state with a certain probability:
\begin{align}
    (\mathbb{1} \otimes F) \ket{\phi^{+}} = \gamma \ket{\psi}.
\end{align}
Ideally, we would like to have $\gamma = 1$, which would correspond to a 
deterministic and exact realisation of the state $\ket{\psi}$ by applying 
the filter to a maximally entangled state, however, this is not possible 
due to the condition $F^{\dagger}F \leq \mathbb{1}$. Hence, the desired 
realisation can be implemented only probabilistically. More explicitly, 
consider an arbitrary density matrix $\sigma = \sum_i \alpha_i \ketbra{i}{i}$ 
with $\sqrt{\sigma} = \sum_i \sqrt{\alpha_i} \ketbra{i}{i}$. The purification of $\sigma$ can be written as $\ket{\psi} = \sum_i \sqrt{\alpha_i }\ket{ii} 
= \sum_i \sqrt{\alpha_i } \sqrt{d}\frac{\ket{ii}}{\sqrt{d}}$. This leads to the equation $(\mathbb{1} \otimes \sqrt{\sigma d})\ket{\phi^{+}} = 
\ket{\psi}$. The term $\sqrt{\sigma d}$ here would be a perfect filter in 
the sense of achieving $\ket{\psi}$ exactly and deterministically, however, 
it fails the condition of a quantum filter.

Hence, we need to find the optimal normalization that satisfies the condition, while still achieving the best possible probability of implementing the state $\ket{\psi}$. In particular, we found that it is given by the following filter,
\begin{equation}
    F\coloneqq \frac{\sqrt{\sigma}}{\sqrt{\norm{\sigma}}},
\end{equation}
where $\sigma = \Tr_A(\ketbra{\psi}{\psi}_{AB})$, which obtain the desired protocol for unitary transposition. Interestedly, as proven in the App.~\ref{append:bipart_proofs} this filtering-based protocol is optimal for unitary transposition. 

\begin{restatable}{theorem}{ThmBipTrans}\label{bip_trans_thm}
Given a single call of the $d$-dimensional unitary operation $U \in \mathrm{SU}(d)$, in the case of a supermap realised on a subsystem of a known \emph{bipartite} pure input state $\ket{\psi}_{AB}\in \mathcal{H}_{A}\otimes\mathcal{H}_{B}$, the transposition operation $U^{T}$ can be implemented with the maximal success probability,
    \begin{align}
    p^{\psi_{AB}}_\textup{max,trans}=\frac{1}{d\norm{\Tr_{A}(\ketbra{\psi}{\psi}_{AB} )}_{\textup{op}}}
    \end{align} 
\end{restatable}

Let us notice now some special cases:
\begin{itemize}
    \item For a maximally entangled state, we have $p(\ketbra{\phi^{+}}{\phi^{+}}_{AB}) = 1$. This is not surprising, as it follows directly from the property of the maximally entangled state: $(\mathbb{1} \otimes A)\ket{\phi^{+}} = (A^{T} \otimes \mathbb{1})\ket{\phi^{+}}$, after which we simply do nothing in the decoder part. 
    \item For the single-subsystem pure state case, we recover the result proved in the previous section, $p(\ketbra{\psi}{\psi}_{AB} = \ketbra{\phi}{\phi}_{A} \otimes \frac{\mathbb{1}_B}{d}) = \frac{1}{d}$
\end{itemize}

\paragraph{Deterministic and exact transposition operation on a known-mixed state}

It follows from Thm.~\ref{thm:Mix_visib} and Thm.~\ref{thm:trans_det} that deterministic and exact implementation of $U \mapsto U^{T}\rho \overline{U}$, known state of the form $\rho = \eta \ketbra{\psi}{\psi} + \frac{1-\eta}{d}\mathbb{1}$ can be achieved for the values of $\eta\leq \eta^*_{\textup{trans}}$, where $\eta^*_{\textup{trans}}$ equals,
\begin{equation}\label{Eq:trans_det_exact}
    \eta^*_{\textup{trans}} = \frac{d}{d^2 - 1}.
\end{equation}
We note that this visibility coincides with that of the optimal universal NOT channel defined in Eq.~\eqref{Eq:negvis}. This is consistent with the optimal protocol for pure unentangled states shown in Fig.~\ref{Fig:TransDetpure}, where we implement the probabilistic scheme using a measurement with POVM elements $\ketbra{\psi}^T$ and $\id - \ketbra{\psi}^T$. When the outcome corresponding to $\ketbra{\psi}^T$ is obtained, the resulting state is $U^T\ketbra{\psi}\overline{U}$. When the outcome corresponding to $\id - \ketbra{\psi}^T$ is obtained, we apply the optimal universal NOT channel $\map{N}$ to produce a state that approximates $U^T\ketbra{\psi}\overline{U}$.

\begin{figure}[b!]
    \centering
\includegraphics[width=0.7\textwidth]{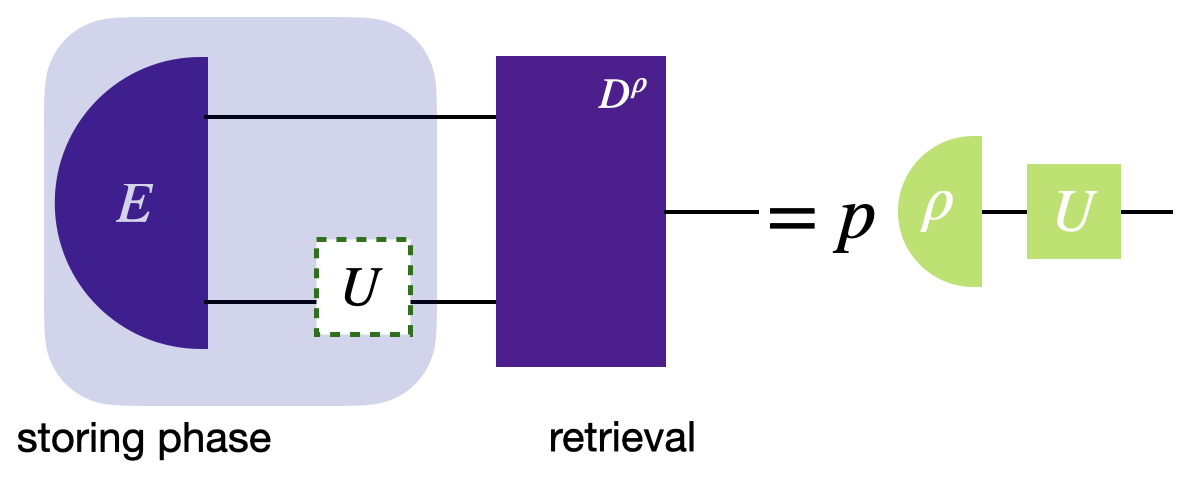}
    \caption{ SAR protocol with a general encoder and decoder. While the decoder depends on the classical information specifying \(\rho\), the encoder, and thus the storage phase, is independent of \(\rho\).}
    \label{fig:SARdelay}
\end{figure}

\subsection{Unitary storage-and-retrieval}\label{sec:SAR}

Imagine being given access to an unknown quantum operation, but only later, after losing access to it, you obtain the input state on which you wish to implement the operation. This is the setting of the storage-and-retrieval (SAR) protocol, which consists of storing the action of an unknown operation on a chosen storing state and later retrieving it on the desired input state, as depicted in Fig.~\ref{fig:SARdelay}.

SAR corresponds to the delayed input state protocol for the function \( f(U) = U \). As discussed in Sec.~\ref{sec:delay_KN}, we define delayed input protocols as those in which access to the classical description of the input state is postponed; that is, rather than delaying access to the input state itself, we delay access to the information about it. Consequently, the decoder channel corresponding to the retrieval step can be adapted to the knowledge of the input state $\rho$, but, importantly, the encoder state $E\in\mathcal{L}(\mathcal{H}_\I\otimes\mathcal{H}_\A)$ cannot. Note that if the encoder state $E$ could depend on the input state $\rho$, the task would be trivial, as one could simply set $E=\rho$. To ensure that the encoder has no such dependences, we impose the covariance relation $[U,\Tr_\A(E_{\I\A})]=0$ for every\footnote{When $k$ calls of the input operation $U$ are available, we impose that $[U^{\otimes k},\Tr_\A(E_{\I\A})]=0$.} $U\in\mathrm{SU}(d)$. This condition also guarantees that the overall protocol is input-state covariant, meaning that the success probability of the SAR protocol does not depend on the choice of the known input state.

Similarly to the unknown input state case~\cite{Quintino2019prob,Quintino_2022det}, there exists an equivalence relation between the task of SAR and parallel unitary transposition. More precisely, we have the following theorem.

\begin{restatable}{theorem}{ThmSARsingle}\label{thm:SAR_transp_Eq}
Let $\ket{\psi}_{AB} \in \mathcal{H}_{A} \otimes \mathcal{H}_{B}$ be an arbitrary pure bipartite state. Every probabilistic exact or deterministic approximate storage-and-retrieval (SAR) task that transforms $k$ calls to an unknown input unitary $U \in \mathrm{SU}(d)$ into $(\id_A \otimes U_B)\ket{\psi}_{AB}$ can be converted, without loss in performance, into a task that implements parallel unitary transposition, transforming $k$ calls of the same unitary $U$ into $(\id_A \otimes U^{T}_B)\ket{\psi}_{AB}$, and vice versa. Hence, these two tasks are in one-to-one correspondence.
\end{restatable}
The proof for arbitrary $k$ is provided in App.~\ref{sec:proof_SAR}, and we present here the proof for a single call of the input unitary ($k=1$). 
\begin{proof} Let us begin by considering one direction of the equivalence, namely that the construction for transposition protocol implies the construction for SAR protocol. As established in Thm.~\ref{thm:trans_single}, the encoder state may be taken to be the maximally entangled state without any loss of performance. Using the well-known identity $(\mathbb{1} \otimes U)\ket{\phi^+} = (U^{T} \otimes \mathbb{1})\ket{\phi^+},$ this immediately yields a construction of an SAR protocol with the same performance. For the converse direction, we require that an SAR protocol achieve the same performance for all pure states $\ket{\psi}$, arbitrary but known state. To guarantee this, we must impose the covariance condition $[\Tr_{\M}(E_{\I \M}), V_{\I}] = 0,$ for all $V \in \mathrm{SU}(d)$. Under this symmetry requirement, it follows that the encoder is the purification of the maximally mixed state, hence a maximally entangled state, just as in the unknown state SAR~\cite{Sedl_k_2019}. Finally, by the same identity $(\mathbb{1} \otimes U)\ket{\phi^+} = (U^{T} \otimes \mathbb{1})\ket{\phi^+},$ any covariant SAR protocol can be converted into a unitary transposition protocol, again without loss of performance. \end{proof}
\begin{figure}[t!]
    \centering
\includegraphics[width=0.7\textwidth]{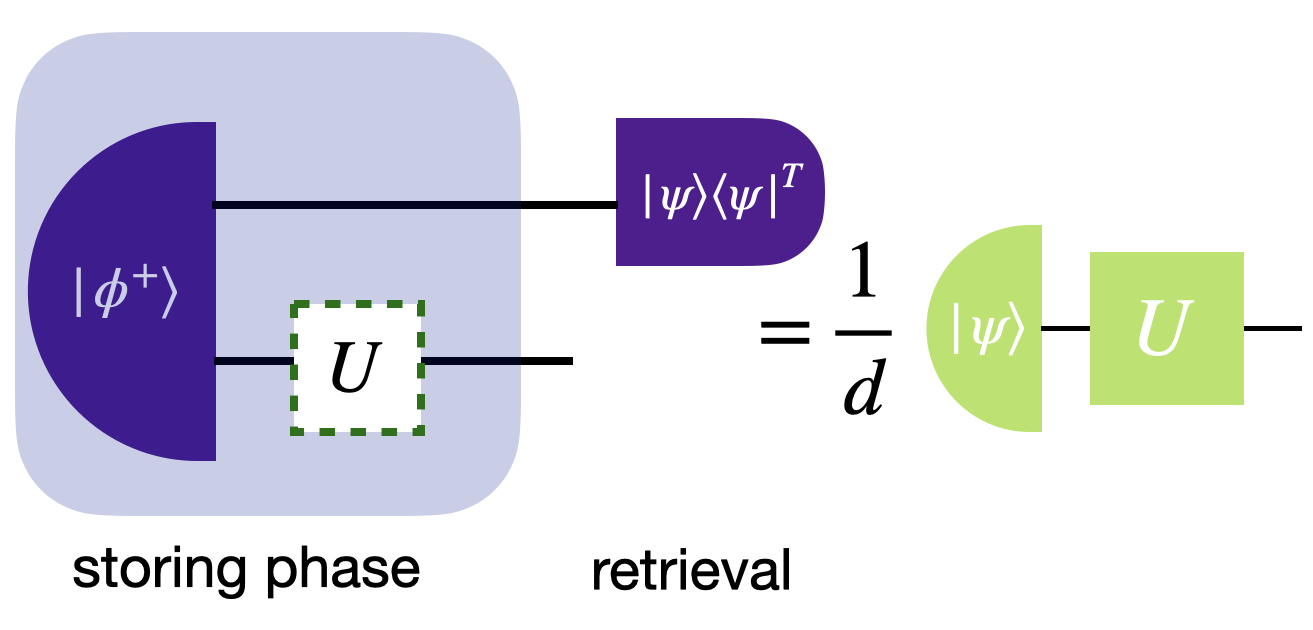}
    \caption{ Optimal realisation of a storage-and-retrieval protocol designed to implement an operation on a known program state. The probability of successful implementation is $p = \frac{1}{d}$.
}
    \label{fig:SARknpure}
\end{figure}

From Thm.~\ref{bip_trans_thm} and the equivalence established above, we conclude that, given a single call to a $d$-dimensional unitary operation $U \in \mathrm{SU}(d)$, a SAR protocol implemented on a known pure bipartite input state $\ket{\psi}_{AB}\in \mathcal{H}_{A}\otimes\mathcal{H}_{B}$ can achieve a maximal success probability
\begin{equation}
    p^{\psi_{AB}}_{\mathrm{max,\,trans}} = \frac{1}{d \, \norm{\Tr_{A}(\ketbra{\psi}{\psi}_{AB})}_{\mathrm{op}}}
\end{equation}
for probabilistic exact realisation. 

We now notice that the probabilistic protocols presented here, as well as its proofs of optimality, admit a straightforward generalisation from unitary operations to general quantum channels. To see this, note that the SAR optimal protocol described in Fig.~\ref{fig:SARknpure} succeeds with probability $p = \frac{1}{d}$ even when a non-unitary quantum channel is applied. Therefore, for arbitrary qudit-to-qudit quantum channels, a success probability of $p = \frac{1}{d}$ is attainable. Since it is impossible to exceed this probability for $d$-dimensional unitary operations, it follows that one also cannot exceed it for general qudit channels.

Hence, the success probability $p=\frac{1}{d}$ is optimal even when considering channels that are not unitary\footnote{This argument also applies to the transposition of unital channels, 
i.e., completely positive trace-preserving maps $\map{C}$ that satisfy 
$\map{C}(I) = I$. To see this, consider a Kraus representation 
$\map{C}(\rho) = \sum_i K_i \rho K_i^\dagger$. The transposed channel is 
defined by transposing each Kraus operator, $\map{C}^T(\rho) = \sum_i K_i^T \rho (K_i^T)^\dagger$. Using the identity 
$(\mathbb{1} \otimes K_i)\ket{\phi^+} = (K_i^T \otimes \mathbb{1})\ket{\phi^+}$, we see that the transpose of a channel is not, in general, trace-preserving, 
unless the channel is unital, which in Kraus form requires $\sum_i K_i K_i^\dagger = \mathbb{1}$. Therefore, when considering channel transposition, our conclusions for the case of unitaries also hold for unital channels.}.

For the deterministic case, we present results only for a pure state unentangled with the reference. From Thm.~\ref{thm:trans_det}, we have that, given a single call to a $d$-dimensional unitary operation $U \in \mathrm{SU}(d)$, a SAR protocol on a known pure input state $\ketbra{\psi}{\psi}$ achieves an optimal average fidelity
\begin{equation}
    \expval{F^{\psi}}_{\mathrm{trans}} = \frac{2d + 1}{d^2 + d}
\end{equation}
for deterministic approximate realisation.
As we can see, these results represent quadratic improvement compared to the results of the unknown input state case \cite{Bisio_2010}.

Finally, in Sec.~\ref{sec:Garazi}, we analyse the relationship of the task of SAR with known input states and protocols for port-based state preparation introduced in~\cite{Muguruza_2024}.

\subsection{Unitary complex conjugation}

We now explore the task of universal unitary complex conjugation on known input states, that is, the scenario where \( f(U) = \overline{U} \).

\paragraph{Optimal probabilistic exact conjugation protocol on a pure known unentangled state}

Probabilistic exact complex conjugation on known input states can be performed deterministically and exactly when \( d = 2 \), whereas for \( d > 2 \), the success probability is necessarily zero. In the qubit case, deterministic and exact unitary conjugation can be implemented simply by applying the Pauli \( Y \) operation before and after the arbitrary unitary~\cite{Miyazaki_2017}, as illustrated in Fig.~\ref{fig:conjgKNpure}. For \( d > 2 \), Ref.~\cite{Quintino2019prob} shows that in the case of unknown input states, any probabilistic exact protocol must have vanishing success probability, i.e., \( p = 0 \). Our result thus extends this no-go theorem to the case of known pure input states, showing that the success probability remains zero even with full knowledge of the state.

\begin{figure}[t!]
    \centering
\includegraphics[width=0.7\textwidth]{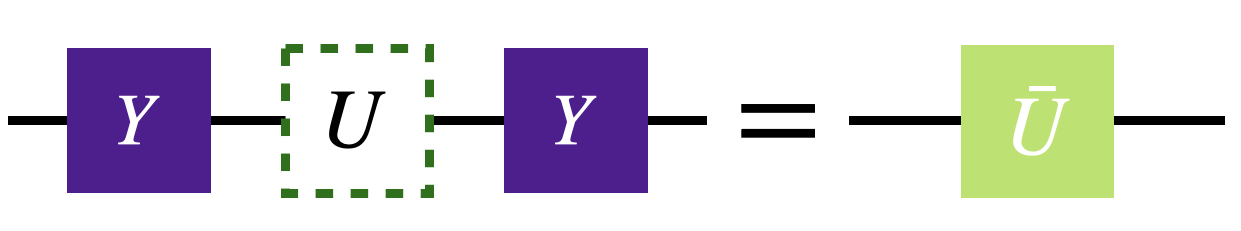}
    \caption{Protocol implementing the conjugation transformation for a qubit system, as proposed in Ref.~\cite{Miyazaki_2017}. The implementation is deterministic, exact, and universal with respect to both the input state and the input unitary. In Thm.~\ref{single_conjg} we show that for $d>2$ system, universal unitary conjugation cannot be implemented, i.e., $p=0$, in even in the known input scenario. Furthermore, in Thm.~\ref{thm:conjgBip}, we extend this result to the case of known pure bipartite states, showing that universal unitary conjugation is likewise impossible when \( d > 2 \).
}
    \label{fig:conjgKNpure}
\end{figure}

\begin{restatable}{theorem}{ProbConjgSingl}\label{single_conjg}
Given a single call of the $d$-dimensional unitary operation $U \in \mathrm{SU}(d)$, in the case of a supermap realised on a known pure input state $\ket{\psi}\in\mathbb{C}^d$, the conjugation operation $\overline{U}$ can be implemented with maximal success probability
\begin{equation}
 p^{\psi}_{\textup{max,conj}}  =
\begin{cases}
       1,& \textup{if } d = 2\\
    0,             &  \textup{if } d > 2
\end{cases} \quad \quad \textcolor{gray}{ \textup{cf. } \quad p_{\textup{max,conj}} = \begin{cases}
       1,& \textup{if } d = 2\\
    0,             &  \textup{if } d > 2.
\end{cases}}
\end{equation}

\end{restatable}

\paragraph{Optimal deterministic approximate conjugation protocol on a pure known unentangled state}

\begin{figure}[b!]
    \centering
\includegraphics[width=0.6\textwidth]{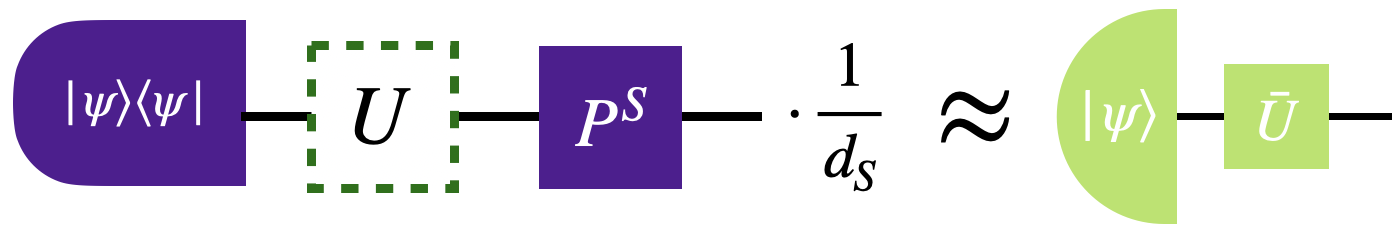}
    \caption{Optimal protocol for implementing the conjugation operation in a deterministic and approximate manner on an arbitrary known pure state. $P^{S}$ denotes the projector onto the symmetric subspace and $\frac{1}{d_S}$ denotes its normalisation.}
    \label{fig:conjgDetpure}
\end{figure}

We now consider deterministic approximate conjugation protocol on a pure known unentangled state. As discussed in the previous section, when $d=2$, deterministic an exact unitary complex conjugation is possible. Now, in Thm.~\ref{thm:det_conj} stated below, we prove that the optimal superchannel has a Choi operator given by,
\begin{align}
     S^{\textup{opt}}_{\I\O\F, \textup{conj}} =\ketbra{\psi}{\psi}_{\I} \otimes   \frac{2}{d+1} P^{S}_{\O\F} 
\end{align}
where $P^{S}_{\O\F}\coloneqq\frac{\id_{\O\F}+\textup{SWAP}_{\O\F}}{2}$ is the projector onto the symmetric subspace of $\mathbb{C}^d\otimes \mathbb{C}^d$ and $\textup{SWAP}\in\mathcal{L}(\mathbb{C}^d\otimes\mathbb{C}^d)$ is the swap operator, defined via $\textup{SWAP}\ket{\psi}\ket{\phi}=\ket{\phi}\ket{\psi}$ for every $\ket{\psi},\ket{\phi}\in\mathbb{C}^d$. We remark that, $\frac{2}{d+1} P^{S}_{\O\F}$ is the Choi operator of the channel $\map{C}(\rho)=\frac{\id\Tr(\rho)+\rho^T}{d+1}$,
which known as the  is the optimal completely positive trace-preserving approximation for state transposition~\cite{Buscemi_2003TranspMap, Kalev_2013Qdesigns, Dong_2019PosMaps}. The protocol is depicted in Fig.~\ref{fig:conjgDetpure}.

Finally, we remark that, unlike unitary transposition, a deterministic strategy based on a probabilistic protocol followed by correction is not suitable in this context. This is because for \( d = 2 \), no correction is needed, since the task can be performed perfectly; whereas for \( d > 2 \), there is no successful branch at all, as the success probability is zero. In other words, a correction strategy is only meaningful when the probabilistic protocol has a non-zero probability of success.

\begin{restatable}{theorem}{DetConjgSingl}\label{thm:det_conj}
 Given a single call of the $d-$dimensional unitary operation $U \in \mathrm{SU}(d)$, in the case of a supermap realised on a known pure input state $\ket{\psi}\in\mathbb{C}^d$,   
    the conjugation operation $\overline{U}$  can be implemented, attaining optimal average fidelity 
    \begin{equation}
 \expval{F^{\psi}}_{\textup{conj}}  =
\begin{cases}
       1,& \textup{if } d = 2\\
    \frac{2}{d + 1},              &  \textup{if } d >  2 
\end{cases} \quad \quad \textcolor{gray}{\textup{cf. } \quad \expval{F}_{\textup{conj,} k = 1}  = \frac{2}{d(d-1)}}
\end{equation}
\end{restatable}
Comparing the optimal average fidelities for the known and unknown cases, we observe a quadratic improvement in the known case.

\paragraph{Conjugation on a part of a known bipartite state: probabilistic exact}

We now consider the task of probabilistic exact complex conjugation on pure bipartite quantum states. We show that the situation for known pure bipartite states is analogous to that of pure product states: when \( d = 2 \), unitary complex conjugation can be performed deterministically and exactly, whereas for \( d > 2 \), the success probability must necessarily be zero.

\begin{restatable}{theorem}{ConjgBipThm}\label{thm:conjgBip}
    
Given a single call of the $d$-dimensional unitary operation $U \in \mathrm{SU}(d)$, in the case of a supermap realised on a subsystem of a known \emph{bipartite} pure input state $\ket{\psi}_{AB}\in \mathcal{H}_{A}\otimes\mathcal{H}_{B}$, the conjugation operation $\overline{U}$ can be implemented with maximal success probability
\begin{equation}
 p^{\psi_{AB}}_{\textup{max,conj}}  =
\begin{cases}
       1,& \textup{if } d = 2\\
    0,             &  \textup{if } d > 2.
\end{cases}
\end{equation}
\end{restatable}

\paragraph{Deterministic and exact conjugation operation on a known-mixed state}

By combining Thm.~\ref{thm:Mix_visib} with Thm.~\ref{thm:det_conj}, it follows that deterministic and exact implementation of $U \mapsto \overline{U}\rho U^{T}$ on an arbitrary, known state of the form $\rho = \eta \ketbra{\psi}{\psi} + \frac{1-\eta}{d}\mathbb{1}$ is possible if and only if $\eta\leq \eta^*_{\textup{conj}}$, where $\eta^*_{\textup{conj}}$ equals,
\begin{equation}
    \eta^*_{\textup{conj}} =  \begin{cases}
       1,& \textup{if } d = 2\\
    \frac{1}{d + 1}              &  \textup{if } d >  2 .
\end{cases} 
\end{equation}

\subsection{Unitary inversion}\label{sec:inversion}

We now explore the task of universal unitary inversion on known input states, that is, the case where $f(U)=U^{-1}$.
The inversion, as a map, is the composition of transposition and conjugation $f_{\textup{inv}} = f_{\textup{trans}} \circ f_{\textup{conj}} $. However, this does not imply a straightforward relationship between the supermaps implementing the corresponding protocols in the known input state scenario. While in the unknown input state case it is meaningful to consider compositionality of the superchannels, in the known state scenario, we cannot simply plug one known state supermap into another, since our supermaps transform operations into quantum states, and cannot be composed. Hence, in general, we do not have compositionality of the known state supermaps. 

\paragraph{Optimal probabilistic exact inversion protocol on a pure known unentangled state}

When $d=2$, we can use the relation $YU_2 Y=\overline{U_2}$ that holds for all $U_2\in\mathcal{SU}(2)$ to recognise that the unitary inversion problem is unitarily equivalent to unitary transposition. Also, below we prove that, similarly to unitary conjugation, when $d>2$ any attempt to obtain an exact unitary inversion necessarily has a zero success probability. 
\begin{restatable}{theorem}{InvProbSingl}\label{thm:InvProbSingl}
 Given a single call of the $d$-dimensional unitary operation $U \in \mathrm{SU}(d)$, in the case of a supermap realised on a known pure input state $\ket{\psi}\in\mathbb{C}^d$, 
 the inversion operation $U^{-1}$ can be implemented with maximal success probability
\begin{equation}
 p^{\psi}_{\textup{max}, inv}  =
\begin{cases}
       \frac{1}{2},& \textup{if } d = 2\\
    0,             &  \textup{if } d > 2
\end{cases} \quad \quad  \textcolor{gray}{ \textup{cf. } \quad p_{\textup{max,inv}} = \begin{cases}
       \frac{1}{4},& \textup{if } d = 2\\
    0,             &  \textup{if } d > 2.
\end{cases}}
\end{equation}
\end{restatable}
The proof of this theorem is presented in App.~\ref{subsec:proof_inv}.

\paragraph{Optimal deterministic approximate inversion protocol on a pure known unentangled state}

\begin{figure}[b!]
    \centering
\includegraphics[width=0.9\textwidth]{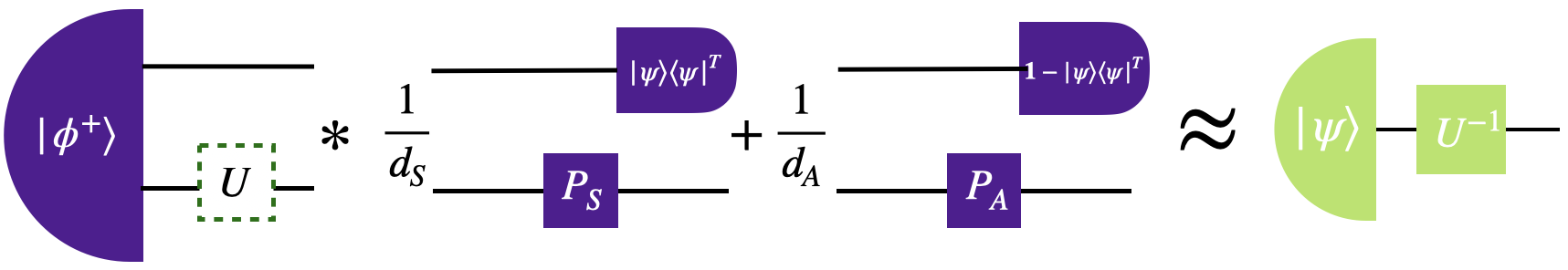}
    \caption{Optimal protocol for implementing the inversion operation in a deterministic and approximate manner on an arbitrary known pure state. $P^{S}$ and $P^{A}$ denote the projectors onto the symmetric and antisymmetric subspaces, respectively.}
    \label{fig:inversion_det}
\end{figure}

We now analyse deterministic approximation protocols for unitary inversion protocol on a pure known unentangled states. In Thm.~\ref{thm:det_inv} stated below, we prove that for any dimension $d$, the optimal average fidelity for unitary inversion on known input states is $\expval{F^{\psi}}_{\textup{inv}}  =
\frac{d+3}{d(d+1)} $. This optimal performance is attained by superchannel with a Choi operator
\begin{align}
     S^{\textup{opt}}_{\I\O\F, \textup{inv}} = \frac{\ketbra{\psi}{\psi}_{\O} }{d_{S}}\otimes P^{S}_{\I\F} + \frac{(\mathbb{1}_{\O} - \ketbra{\psi}{\psi}_{\O})}{d_{A}}\otimes P^{A}_{\I\F}
\end{align}
where $P^{S}_{\O\F}\coloneqq\frac{\id_{\O\F}+\textup{SWAP}_{\O\F}}{2}$ and $P^{A}_{\O\F}\coloneqq\frac{\id_{\O\F}-\textup{SWAP}_{\O\F}}{2}$ are the projectors onto the symmetric and anti-symettric space and $d_{S}=\tr(P^S)$ and $d_{A}=\tr(P^A)$ their respective dimensions. The optimal protocol is illustrated in Fig.~\ref{fig:inversion_det}, where, similarly to the case of unitary transposition illustrated in Fig.~\ref{Fig:TransDetpure}, the supermap splits into two branches conditioned on the measurement in the \( \mathcal{L}(\mathcal{H}_{\O}) \) space. When the measurement element is \( \ketbra{\psi}{\psi}_{\O} \), the supermap applies the channel with Choi operator proportional to the projector onto the symmetric space \( P_{S} \). Conversely, when the measurement element is \( \mathbb{1}_{\O} - \ketbra{\psi}{\psi}_{\O} \),  channel with Choi operator proportional to the projector onto the anti-symmetric space \( P_{A} \) is applied.

\begin{restatable}{theorem}{InvDetSing}\label{thm:det_inv}
 Given a single call of the $d-$dimensional unitary operation $U \in \mathrm{SU}(d)$, in the case of a supermap realised on a known pure input state $\ket{\psi}\in\mathbb{C}^d$, the inversion operation $U^{-1}$ can be implemented, attaining optimal average fidelity 
\begin{equation}
 \expval{F^{\psi}}_{\textup{inv}}  =
\frac{d+3}{d(d+1)} \quad \quad \quad \textcolor{gray}{ \textup{cf. } \quad \expval{F}_{\textup{inv}, k=1}  = \frac{2}{d^2}}.
\end{equation}
\end{restatable}

We again see the quadratic improvement compared to the unknown case.

\paragraph{Inversion on a part of a known bipartite state: probabilistic exact}

Similarly to the pure state case, unitary inversion protocol will exhibit the properties of both conjugation and transposition, even though we cannot directly compose the corresponding supermaps. Specifically, we again have zero probability for $d>2$, while for the qubit case, the protocol is equivalent to transposition, shown in Fig.~\ref{fig:bipartTrans}. 

\begin{restatable}{theorem}{InvBipProb}\label{bip_inv_thm}
 
 Given a single call of the $d$-dimensional unitary operation $U \in \mathrm{SU}(d)$, in the case of a supermap realised on a subsystem of a known \emph{bipartite} pure input state $\ket{\psi}_{AB}\in \mathcal{H}_{A}\otimes\mathcal{H}_{B}$, the inversion operation $U^{-1}$ can be implemented with maximal success probability
\begin{equation}
 p^{\psi_{AB}}_{\textup{max,inv}}  =
\begin{cases}
        \frac{1}{2\norm{\Tr_{A}(\ketbra{\psi}_{AB})}_{\textup{op}}},& \textup{if } d = 2\\
    0,             &  \textup{if } d > 2.
\end{cases} 
\end{equation}
\end{restatable}

For $d=2$, we know that inversion protocol is the same as transposition, while for $d>2$ conjugation, we provide an independent proof in App.~\ref{append:bipart_proofs}. 

\paragraph{Deterministic and exact inversion operation on a known-mixed state}

As in the unitary transposition and unitary conjugation, we can combine Thm.~\ref{thm:Mix_visib} and Thm~\ref{thm:det_inv} to show that deterministic and exact implementation of $U \mapsto U^{-1}\rho U$ on known state of the form $\rho = \eta \ketbra{\psi}{\psi} + \frac{1-\eta}{d}\mathbb{1}$ is possible if and only if $\eta\leq \eta^*_{\textup{inv}}$, where 
\begin{equation}
    \eta^*_{\textup{inv}} =  \frac{2}{d^2 - 1}.
\end{equation}

\section{Higher-order quantum computing with multiple calls of the input operation}\label{sec:k>1}

Previously, we have focused our analysis on tasks where the input operation we aim to transform is called only once, that is, one has access to a single use of the unitary operation. In this section, we consider a variation of such tasks where multiple calls to the input operation are available. Specifically, for a given function \( f : \mathrm{SU}(d) \to \mathrm{SU}(d) \) and a fixed state \( \rho \in \mathcal{L}(\mathbb{C}^d) \), we aim to transform an arbitrary unitary operation \( \map{U}(\rho) = U \rho U^\dagger \) into the state \( f(U)\rho f(U)^\dagger \) using \( k \) calls to \( U \).

Mathematically, having \( k \) calls to an operation \( \map{U} \) is equivalent to having access to \( k \) copies of the linear map \( \map{U} \). Hence, we look for supermaps that implement
\begin{equation}
\smap{S}^{\rho} : \map{U}^{\otimes k} \mapsto f(U)\rho f(U)^{\dagger}.
\end{equation}
Here, the total input space is described by $\mathcal{H}_{\I} = \bigotimes_{i=1}^{k} \mathcal{H}_{I_i},$
and the total output space by $\mathcal{H}_{\O} = \bigotimes_{i=1}^{k} \mathcal{H}_{O_i}.$
For each \( i \in \{1, \ldots, k\} \), the spaces \( \mathcal{H}_{I_i} \) and \( \mathcal{H}_{O_i} \) are isomorphic to \( \mathbb{C}^d \), so we have $\mathcal{H}_{\I} \cong \mathcal{H}_{\O} \cong (\mathbb{C}^d)^{\otimes k} \cong \mathbb{C}^{d^k}.$ From this perspective, the single call tasks analysed in the previous sections correspond to the particular case where \( k = 1 \).

In order to address problems where multiple calls are available, we will make use of superchannels and and superinstruments with multiple slots. These are transformations where one can use to transform $k$ independent channels into another quantum channel~\cite{taranto2025review} . The structure of superchannels and superinstruments with multiple slots is considerably richer then its one-slot version, in particular, the input operations may be called in different causal orders. In this work, we focus our analysis on three different classes that arise naturally when analysing multiple slots superchannels; parallel superchannels - where all operations are called in parallel as illustrated in Fig.~\ref{fig:paralel}, sequential superchannels - also known as quantum combs~\cite{Chiribella2007architecture}, where the input operations are called sequentially in a quantum circuit which allows adaptive strategies, as illustrated in Fig.~\ref{fig:sEq}, and general superchannels - where the input operations can be called without respecting a definite causal order~\cite{Araujo_2015Witns,Chiribella_2013Switch, Ognyan_2012OCB}. We note that, when considering multiple calls of the same unitary, such as the tasks analysed in this work, quantum circuits with classical control (QC-CC) and quantum circuits with quantum control (QC-QC) cannot outperform sequential circuits~\cite{Bavaresco2022,Abbott2025}.

In the following subsections, we present a precise definition and mathematical characterisation of these three classes of superchannels and superinstruments.

\begin{figure}[b!]
    \centering
\includegraphics[width=0.7\textwidth]{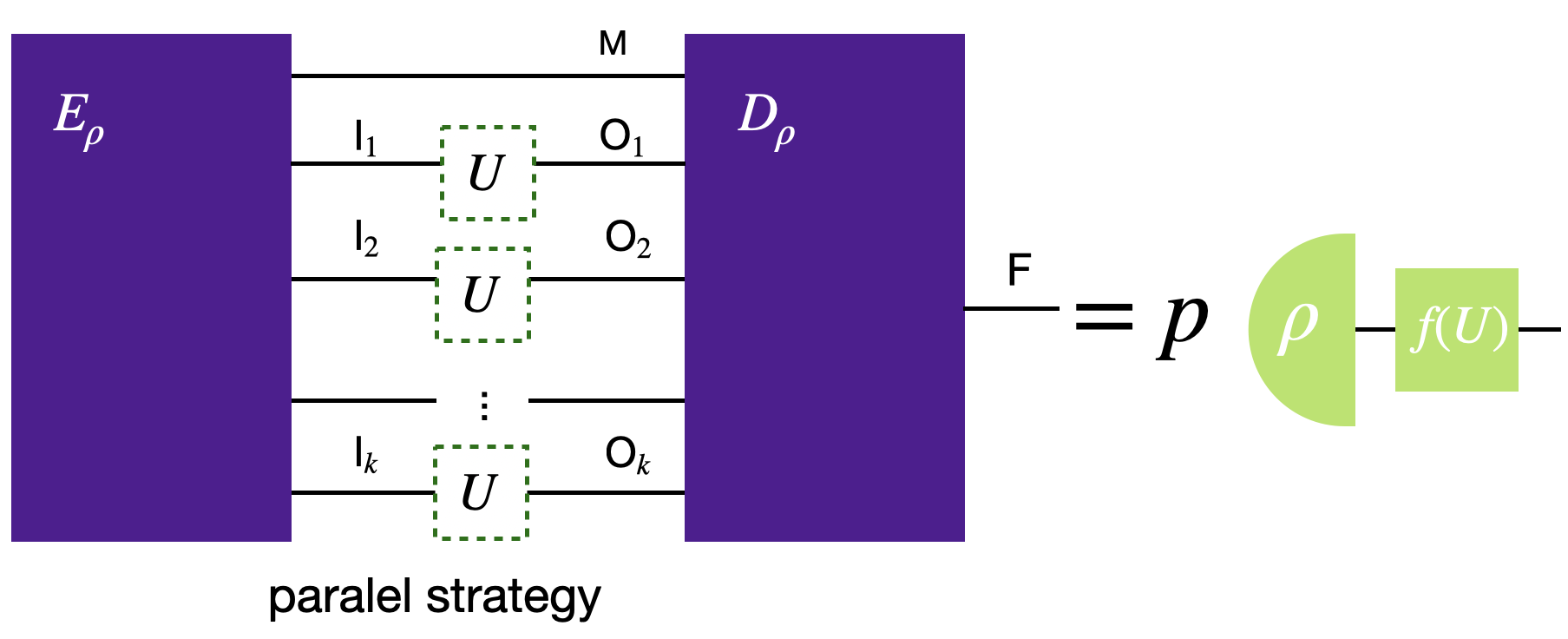}
    \caption{ Known state supermap protocol realising $f(U)\rho f(U)^{\dagger}$ using $k$ calls to the unknown unitary in a parallel strategy. }
    \label{fig:paralel}
\end{figure}

\subsection{Parallel strategies}

We now introduce the concept of parallel strategies, which are protocols designed to transform quantum operations through parallel superchannels or parallel superinstruments. 

In its most general form, a superchannel with \( k \) slots describes a transformation from \( k \) independent quantum channels to a single quantum channel. Let \( \{\map{C}_i\}_{i=1}^k \) be a set of arbitrary, independent quantum channels. To act on these \( k \) channels in parallel, we may treat them collectively as a single channel, defined by $
\map{C} \coloneqq \bigotimes_{i=1}^k \map{C}_i.$ From this perspective, parallel superchannels can be understood as single-call superchannels (as discussed in Sec.~\ref{sec:Knstate_intro}) that act on the joint channel \( \map{C} \). More directly, in the known input state scenario, parallel superchannels correspond to supermaps that can be described as
\begin{equation}
    \smap{S}^{\rho} \left( \bigotimes_{i=1}^k \map{C}_i \right) =  \left( \map{D}^{\rho} \circ \left[ \bigotimes_{i=1}^k \map{C}_i \otimes \map{\mathbb{1}} \right] \right) (E^{\rho}),
\end{equation}
where \( E^{\rho} \in \mathcal{L}\left( \mathcal{H}_{\I} \otimes \mathcal{H}_{\M} \right) \) is a quantum state, and \( \map{D}^{\rho} : \mathcal{L}\left( \mathcal{H}_{\O} \otimes \mathcal{H}_{\M} \right) \rightarrow \mathcal{L}\left( \mathcal{H}_{\F} \right) \) is a quantum channel.

When analysing \( k \) calls to the same unitary operation, all channels \( \map{C}_i \) are equal to the same map, that is, $\map{C}_i \coloneqq \map{U} : \mathcal{L}(\mathcal{H}_{\I_{i}}) \rightarrow \mathcal{L}(\mathcal{H}_{\O_{i}})$ for every \( i \). In this case, illustrated in Fig.~\ref{fig:paralel}, the action of the parallel superchannel becomes
\begin{equation}
    \smap{S}^{\rho} \left( \map{U}^{\otimes k} \right) = \left( \map{D}^{\rho} \circ \left[ \map{U}^{\otimes k} \otimes \map{\mathbb{1}} \right] \right) (E^{\rho}).
\end{equation}
A linear operator \( S_{\I \O \F} \) is the Choi operator of a parallel superchannel if and only if it satisfies the following constraints:
\begin{align}
S_{\I \O \F} & \geq 0 \label{paral_SDP} \\
\Tr_{\F}(S_{\I \O \F}) & = \Tr_{\O \F} (S_{\I \O \F}) \otimes \frac{\mathbb{1}_\O}{d_{\O}} \\
\Tr\left(S_{\I \O \F}\right) & = d_{\O} \label{paral_SDP3}.
\end{align}
These constraints are equivalent to those characterising a single-slot superchannel, as given in Eqs.~\eqref{Eq:superch_known_constr}--\eqref{Eq:superch_known_constr3}. The only difference lies in the tensor product structure of the input and output spaces, which are now given by 
\(
\mathcal{H}_{\I} = \bigotimes_{i=1}^{k} \mathcal{H}_{I_{i}}
\)
and 
\(
\mathcal{H}_{\O} = \bigotimes_{i=1}^{k} \mathcal{H}_{\O_{i}}.
\)

Finally, as in the single-slot case, probabilistic transformations are described by superinstruments. Analogously, a parallel superinstrument is a collection of completely positive supermaps that sum to a parallel superchannel.

\subsection{Sequential strategies}

\begin{figure}[b!]
    \centering
    \includegraphics[width=0.9\textwidth]{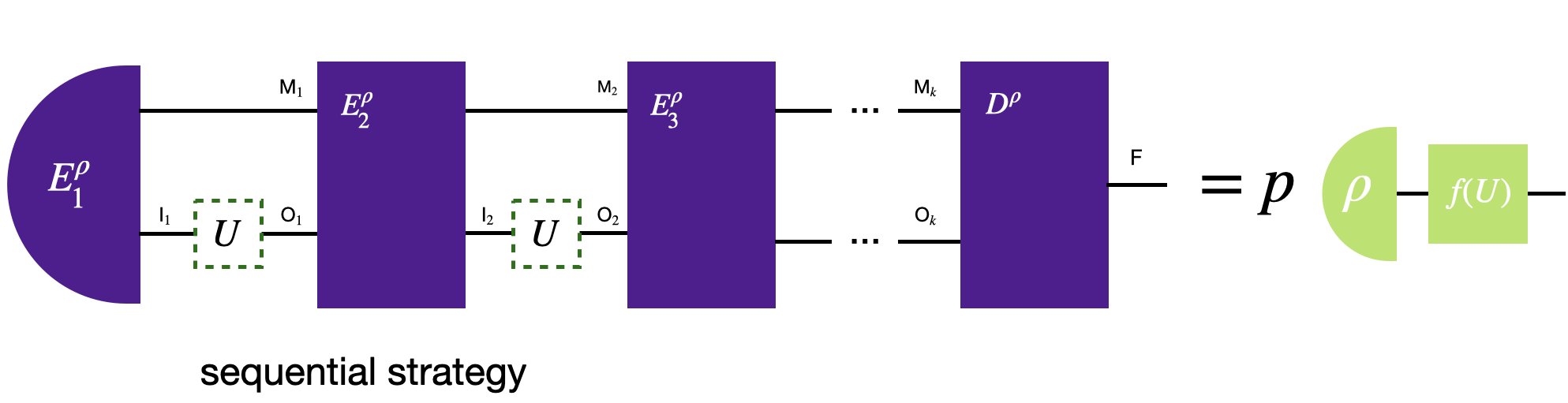}
    \caption{ Known input state supermap protocol realising $f(U)\rho f(U)^{\dagger}$ using $k$ calls to the unknown unitary in a sequential strategy.}
    \label{fig:sEq}
\end{figure}

As the name suggests, in sequential strategies the input operations are called one after another, and intermediate channels can be applied between these calls. This structure gives rise to what is often referred to as adaptive protocols, depicted in Fig.~\ref{fig:sEq}. 

Sequential strategies are mathematically formalised by sequential superchannels, which are also known as quantum combs~\cite{Chiribella_2008}.

A sequential superchannel with \( k \) slots can be understood as a quantum circuit with \( k \) open slots. As in the parallel case, we define the input and output spaces as  
\(
\mathcal{H}_{\I} = \bigotimes_{i=1}^{k} \mathcal{H}_{\I_{i}}
\) 
and 
\(
\mathcal{H}_{\O} = \bigotimes_{i=1}^{k} \mathcal{H}_{O_{i}},
\)
respectively. The action of a sequential superchannel  
\(
\smap{S} : \left (\mathcal{L}(\mathcal{H}_{\I}) \to \mathcal{L}(\mathcal{H}_{\O}) \right) \to \mathcal{L}(\mathcal{H}_{\F})
\)  
on \( k \) calls to the same unitary operation can then be written as
\begin{equation}
    \smap{S}\left(\map{U}^{\otimes k}\right) = \left( \map{D}^{\rho} \circ \left[ \map{U} \otimes \map{I}_{\M} \right] \circ \map{E}_k^{\rho} \circ \cdots \circ \left[ \map{U} \otimes \map{I}_{\M} \right] \right)(E_1^{\rho}),
\end{equation}
where \( E_1^{\rho} \in \mathcal{L}(\mathcal{H}_{I_1} \otimes \mathcal{H}_{\M}) \) is a quantum state,  
and for each \( i \in \{2, \ldots, k\} \),  
\(
\map{E}_i^{\rho} : \mathcal{L}(\mathcal{H}_{O_{i-1}} \otimes \mathcal{H}_{\M}) \to \mathcal{L}(\mathcal{H}_{I_i} \otimes \mathcal{H}_{\M})
\)  
are quantum channels inserted between calls to the unitary.  
Finally,  
\(
\map{D}^{\rho} : \mathcal{L}(\mathcal{H}_{O_k} \otimes \mathcal{H}_{\M}) \to \mathcal{L}(\mathcal{H}_{\F})
\)  
is a quantum channel representing the final decoding.

As in the parallel case, sequential superchannels admit a convenient characterisation in terms of their corresponding Choi operator. Before presenting this characterisation, we introduce the trace-and-replace notation~\cite{Araujo_2015Witns}. For an operator  
\( A \in \mathcal{L}(\mathcal{H}_{\I} \otimes \mathcal{H}_{\O}) \), we define
\begin{equation}
    {}_{\O} A \coloneqq \Tr_{\O}(A) \otimes \frac{\mathbb{1}_{\O}}{d_{\O}}.
\end{equation}

A linear operator \( S_{\I \O \F} \) is the Choi operator of a sequential superchannel if and only if it satisfies the following constraints~\cite{Chiribella_2009,Chiribella2007architecture,taranto2025review}:
\begin{equation}
    S_{\I \O \F} \geq 0,
\end{equation}
\begin{equation}
    {}_{\F} S_{\I \O \F} = {}_{\O_k \F} S_{\I \O \F},
\end{equation}
\begin{equation}
    {}_{\I_k \O_k \F} S_{\I \O \F} = {}_{\O_{k-1} \I_k \O_k \F} S_{\I \O \F},
\end{equation}
\begin{equation}
    \vdots
\end{equation}
\begin{equation}
    {}_{\I_2 \O_2 \ldots \I_k \O_k \F} S_{\I \O \F} = {}_{\O_1 \I_2 \O_2 \ldots \I_k \O_k \F} S_{\I \O \F},
\end{equation}
\begin{equation}
    \Tr(S_{\I \O \F}) = d_{\O}.
\end{equation}

The set of sequential superchannels is strictly larger than the set of parallel ones. As we will see later, for certain tasks—such as unitary transposition and unitary inversion—sequential strategies strictly outperform parallel ones. In contrast, for unitary conjugation, the optimal strategy can also be implemented using a parallel scheme.

\subsection{ General strategies}

\begin{figure}[b!]
    \centering
\includegraphics[width=0.6\textwidth]{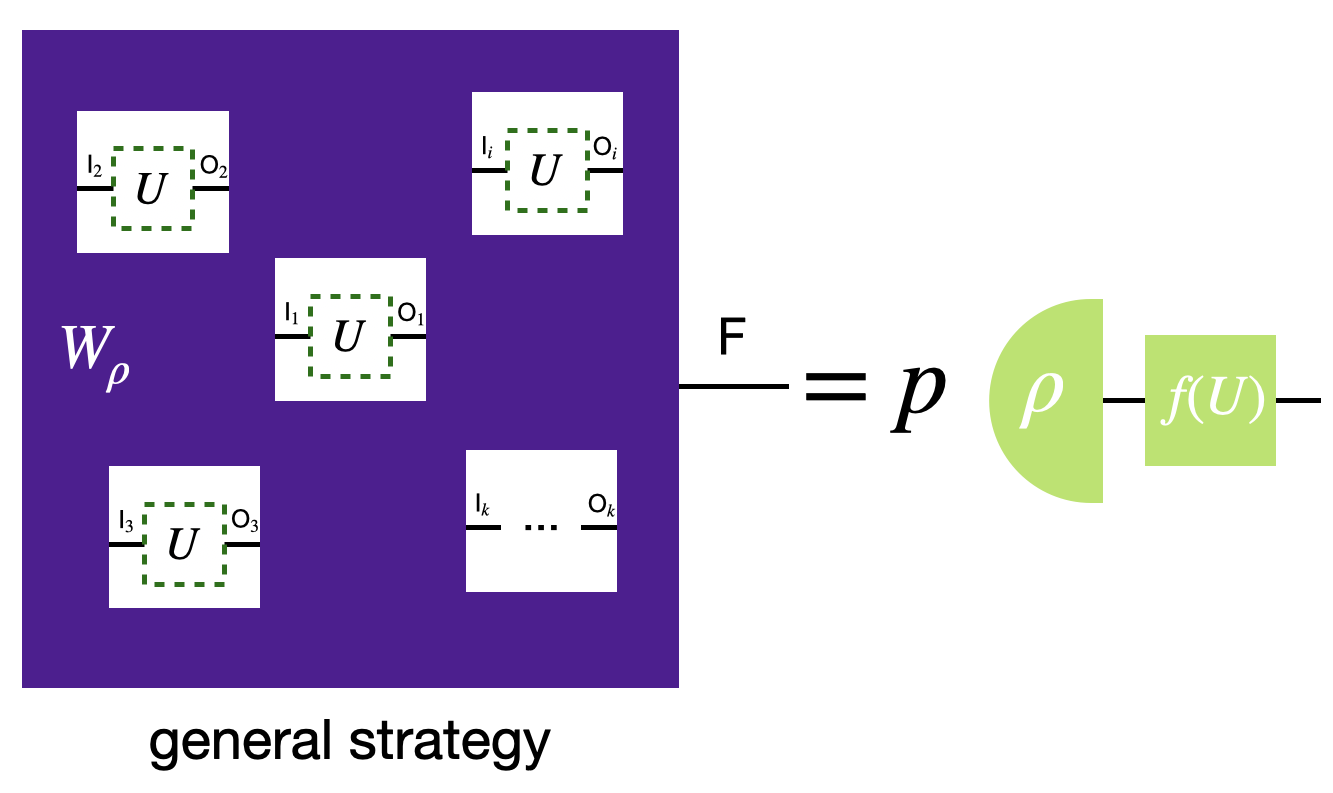}
    \caption{ Known state supermap protocol realising $f(U)\rho f(U)^{\dagger}$ using $k$ calls to the unknown unitary in a general strategy. The supermap corresponding to a general strategy is often denoted by $W$, as indicated in the figure. }
    \label{fig:general}
\end{figure}

General strategies are protocols that go beyond circuit structures and the constraints of causal order~\cite{Quintino2019prob,Quintino_2022det,Ognyan_2012OCB,Araujo_2015Witns,Chiribella_2013Switch}. They are the most general way of processing $U^{\otimes k}$ unitaries into some $f(U)\rho f(U)^{\dagger}$, such that it still satisfies constraints of quantum theory. These structures can also be characterized with the set of semidefinite constraints and formulated into an SDP problem, as detailed for instance in Ref.~\cite{marco_maps, Hoffreumon_2022}. For simplicity we will present only example of $k=2$ case, where we will use the trace and replace notation. A Choi operator $S$ that represents a \emph{general superchannel} transforming $k=2$ independent channels into a quantum state, is characterized with the following set of constraints, 
\begin{align}\label{general_SDP}
S_{\I \O \F } & \geq 0 ; \\
{}_{\I_1 \O_1 \F} S_{\I \O \F } & ={}_{\I_1 \O_1O_{2} \F} S_{\I \O \F } \\
{}_{\I_2 \O_2 \F} S_{\I \O \F } & ={}_{\O_1 \I_{2} \O_{2} \F} S_{\I \O \F } \\
{}_{\F} S_{\I \O \F }+{}_{\O_1 \O_2 \F} S_{\I \O \F } & ={}_{\O_1 \F} S_{\I \O \F }+{}_{\O_2 \F} S  \\
\Tr(S_{\I \O \F }) & = d_{O_1} d_{O_\O} 
\end{align}

where now $\I = \I_1\otimes \I_2 $ and $ \O = \O_1 \otimes \O_2 $. As one might expect, the general strategy is, in principle, the most powerful and can provide an advantage over both sequential and parallel strategies. Any task achievable with a sequential strategy can also be performed using a general strategy, but the converse does not hold. In other words, the following hierarchy holds:  $\textup{PAR} \leq \textup{SEQ} \leq \textup{GEN}$ \cite{Bavaresco_2021Hierarchy,Quintino_2022det,Quintino2019prob}.

\subsection{Probabilistic exact realisations}

Similarly to discussion in Sec.~\ref{sec:probabilistic}, the problem of finding the optimal supermap that transforms $\smap{S}^{\rho} : \map{U}^{\otimes k} \rightarrow f(\map{U})\rho f(\map{U})^{\dagger}$ in probabilistic and exact manner,  achieving maximal possible probability is given by the following set of constraints, 
\begin{align} 
\max &\; p \\ \text { s.t. }  \smap{S}^{\rho}\left(\map{U}^{\otimes k}\right)&=p f(\map{U})(\rho), \quad \forall \map{U} :\mathrm{SU}(d) \rightarrow \mathrm{SU}(d) \\  \{\map{\map{S}}^{\rho}, \map{\map{F}}^{\rho}\}& \text { is a valid superinstrument}\\
&\text { corresponding to parallel, sequential or general strategy}
\end{align}
Notice again that this problem has infinitely many constraints, due to the $\forall$ quantifier. As stated in Thm.~\ref{thm:twirl_prob_constr}, we can use the symmetries of the problem ~\eqref{symmetriesK} to reduce the number of constraints and thereby show that the problem is an SDP, with a finite number of constraints,
\begin{align} 
 \max &\; p \\
\text { s.t. }  S_{\I\O\F} *\dketbra{\mathbb{1}}{\mathbb{1}}_{\I\O} ^{\otimes k} &= p \ketbra{0} \quad \textup{where $S$ is covariant under corresponding symmetry}\\
0  \leq S_{\I\O \F} &\leq S^{\textup{ch}}_{\I\O \F} \\
S \geq 0 \; \textup{and}& \; S^{\textup{ch}} \text { is a valid superchannel}\\
&\text { corresponding to parallel, sequential or general strategy}
\end{align}

\subsubsection{Symmetries of the problem}\label{sym_k}

Similarly to the symmetries discussed in Sec.~\ref{symmetries}, requiring the supermap to act as a universal machine over all input unitaries implies that the defining equation of the protocol exhibits certain symmetry properties,
\begin{itemize}
    \item \textbf{Independence from the choice of known state (pure state case):} as stated in Thm.~\ref{thm:independ}.
    \item \textbf{Input state independent covariance :} as stated in Thm.~\ref{thm:twirl_prob_constr}, without loss of generality, we can restrict our attention to the covariant supermaps $S^{c}_{\I \O \F}$ that exhibit symmetry properties, $[S^{c}_{\I \O \F}, U_{f}] = 0$, where $U_f$ is a total unitary corresponding to the symmetry in case of $f$ transposition, conjugation and inversion. Concretely, we have:
\begin{itemize}
\label{symmetriesK}
    \item \textbf{Transposition and SAR} exhibit the symmetry
    $[S_{\I\O\F}, U^{\otimes k}_{\I} \otimes \mathbb{1}^{\otimes k}_{\O} \otimes \overline{U}_{\F}] = 0 \quad \forall U.$
    Given this symmetry, the supermap admits a decomposition of the form
    \begin{equation}
    S = \sum_{ij\lambda} A^{ij\lambda}_{\O} \otimes B^{ij\lambda}_{\I \F},
    \end{equation}
    where \(B^{ij\lambda}_{\I \F}\) are elements of the orthogonal basis given in terms of Brauer algebra, spanning the subspace \(\mathcal{L}(\mathcal{H}_{\I} \otimes \mathcal{H}_{\F})\).

    \item \textbf{Conjugation} exhibits the symmetry $[S_{\I\O\F}, \mathbb{1}^{\otimes k}_{\I} \otimes U^{\otimes k}_{\O} \otimes U_{\F}] = 0 \quad \forall U,$
    which implies the decomposition
    \begin{equation}
    S = \sum_{ij\lambda} A^{ij\lambda}_{\I} \otimes E^{ij\lambda}_{\O \F},
    \end{equation}
    where \(E^{ij\lambda}_{\O \F}\) are elements of the Yamanouchi basis within the subspace \(\mathcal{L}(\mathcal{H}_{\O} \otimes \mathcal{H}_{\F})\).

    \item \textbf{Inversion} exhibits the symmetry
    $[S_{\I\O\F}, \mathbb{1}^{\otimes k}_{\O} \otimes  U^{\otimes k}_{\I} \otimes U_{\F}] = 0 \quad \forall U,
  $ leading to the decomposition
    \begin{equation}
    S = \sum_{ij\lambda} A^{ij\lambda}_{\O} \otimes E^{ij\lambda}_{\I \F},
    \end{equation}
    where \(E^{ij\lambda}_{\I \F}\) are elements of the Yamanouchi basis within the subspace \(\mathcal{L}(\mathcal{H}_{\I} \otimes \mathcal{H}_{\F})\).

\end{itemize}
 
    \item \textbf{Input state dependent covariance :} as stated in Thm.~\ref{thm:KN_sym}.
\end{itemize}

By exploiting the underlying symmetries, we employ representation-theoretic methods to derive analytical solutions to the SDP. These methods, in particular, provide a natural basis for decomposing the Choi operator of the supermap $S^{c}_{\I \O \F}$, enabling us to determine the optimal solution explicitly by solving for its coefficients. In other words, we can identify the subspaces in which the supermap decomposes and determine how much each contributes to the optimal strategy. Starting with the simpler case of $[S, U^{\otimes k}] = 0$, as it appears for conjugation and inversion. From the Schur-Weyl duality it follows that $S \in \operatorname{Span}\{P_\pi\}_{\pi \in S_k},$ where \(P_\pi\) are the permutation operators acting on the 
\(k\)-fold tensor product space, and \(S_k\) is the symmetric group of permutations of \(\{1,\dots,k\}\). From this span, one can construct the Young-Yamanouchi basis \(\{E^{\mu}_{ij}\}\), which serves as a basis for the supermap operator \(S\) 
in the case of supermaps implementing conjugation or inversion functions.

In the case of transposition, where $[S, U^{\otimes k-1} \otimes \overline{U}] = 0$, 
Schur–Weyl duality is replaced by its extension with the walled Brauer algebra, that is, $S \in \operatorname{Span}\{P_\pi^{T_k}\}_{\pi \in S_k}$, where $P_\pi^{T_k}$ is the partial transposition of the permutation operator $P_\pi$ on its last subsystem.

\subsection{Deterministic approximate realisations} 

Similarly to the discussion in Sec.~\ref{sec:deterministic}, the problem of finding the optimal supermap that transforms  $
\smap{S}^{\rho} : \map{U}^{\otimes k} \mapsto f(\map{U}) \rho f(\map{U})^{\dagger}$ 
in a deterministic approximate manner, while achieving optimal average fidelity, can be formulated through the following constraints:  
\begin{align} \label{eq:det_k}
\max \quad & \expval{F^{\psi}} = \Tr(\Omega^\psi S^\psi)\mathrm{d} U. \\
\textup{s.t.} \; & \map{S}^{\psi} \;\textup{is a valid superchannel }\\
&\textup{corresponding to a parallel, sequential, or general strategy.}
\end{align}
where $\Omega^\psi := \int_{\text {Haar }} \dketbra{\overline{U}}{\overline{U}}^{\otimes k}_{\I\O} \otimes f(U) \ketbra{\psi} f(U)^\dagger$  (see Sec.~\ref{sec:SDP} for more discussion on how do evaluate $\Omega^\psi$.
Deterministic approximate $k > 1$, realisations again fall into three classes of protocols: sequential, parallel, and general, with a performance hierarchy $\textup{PAR} \;\leq\; \textup{SEQ} \;\leq\; \textup{GEN},$ 
quantified in terms of optimal average fidelity.\footnote{That is, any task achievable with a sequential strategy can also be performed using a general strategy, but the converse does not hold.}

Additionally, as discussed in Sec.~\ref{sec:deterministic}, the definition of the performance 
operator allows us to identify the problem in Eq.~\eqref{eq:det_k} as an SDP. We also refer to Sec.~\ref{sec:deterministic} for a discussion on state fidelity and its relationship with white noise robustness as figure of merit. 
\subsection{The repeat until success parallel approach}\label{sec:RUS}

Let us briefly comment on a class of strategies known as Repeat-Until-Success (RUS) protocols \cite{Dong_2021}. 
A RUS protocol is a probabilistic procedure using $k>1$ calls of the input unitary in which, upon failure, the initial state can be reconstructed and the protocol repeated. 

For unknown input state protocols the RUS strategy cannot in general be applied. Namely, in an unknown input supermap protocol, when the protocol fails, the state is consumed by the failure operation. Since this operation is unknown, the state cannot be reconstructed and the protocol cannot be repeated\footnote{A supermap acting on an unknown input state can allow repetition only in specially designed cases, such as the success-or-draw protocol \cite{Dong_2021}, where the input state is preserved in case of failure.}. On the other hand, for known input protocols, on the other hand, RUS is always possible since we have a perfect classical description of the state \(\rho\) on which the transformed unitary \(f(U)\) is applied. This classical information allows us to re-prepare \(\rho\) and thus repeat the strategy upon failure of the supermap, up to the number of available copies of the input unitary.

The overall success probability of the RUS strategy can be expressed in terms of the success probability of a single run, denoted by 
\(p_S\):
\begin{align}
    p^{\textup{RUS}}_{S} = 1 - (1 - p_S)^{k},
\end{align}
where \(k\) denotes the number of repetitions of the task. As we show in Sec.~\ref{sec:SAR}, the RUS construction for transposition provides an exponential advantage compared to the unknown input case, and we conjecture it to be the optimal strategy. Importantly, RUS strategies have strong operational appeal: they rely on single-call protocols and do not require global quantum resource states for the encoder and decoder of a \(k\)-call protocol. Thus, even if the strategy is not strictly optimal, simplicity of implementation makes it particularly appealing.

\subsection{Exponential advantage in unitary storage-and-retrieval}\label{sec:SAR_k}
\begin{figure}[t!]
    \centering
\includegraphics[width=0.8\textwidth]{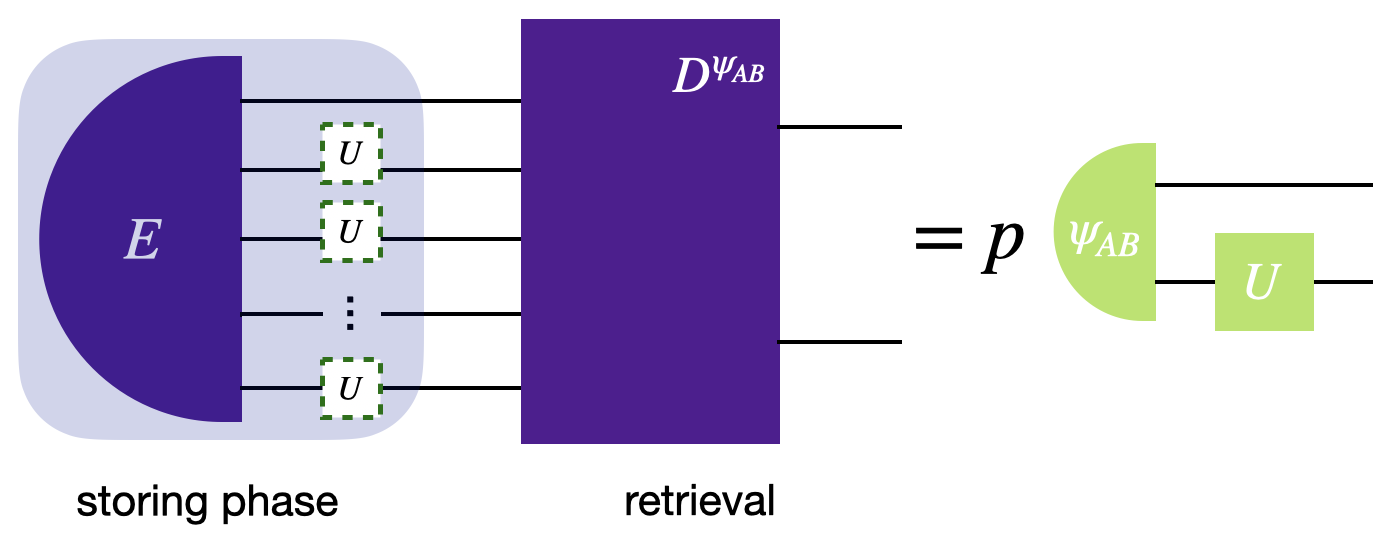}
    \caption{ Storage-and-retrieval protocol with known input state using $k$ parallel queries to the unknown unitary. }
    \label{fig:SAR_k'}
\end{figure}

As in the single call scenario, the multiple call SAR problem is the task where one stores $k$ calls of an unknown unitary operation on an storing quantum state, then, in a later moment, we aim to retrieve the action of $U$ on some known quantum state $\rho$ via a joint quantum operation, as illustrated in Fig.~\ref{fig:SAR_k'}. Previous literature~\cite{Bisio_2010,Sedl_k_2019} considered the case where the retrieval operation does not depend on the input state we aim to retrieve the desired operation and proved the the optimal success probability for any dimension $d$ and number of calls $k$ is given by
\begin{equation}
    p_{\textup{SAR}} = 1 - \left(1 - \frac{k}{k-1+d^2}\right).
\end{equation}

We now show that when the input state is known, the retrieval operation can be tailored to this information. In other words, the storing part of the protocol is completely universal, while the retrieved unitary implementation is tailored for an input known to the user. 
\SARk*
%
Note that while the probability of success presented in Eq.~\eqref{p_SAR_k*} approaches one exponentially in the number of calls $k$, the optimal performance in the unknown input state case approaches one linearly in $k$.

Furthermore, in the unknown input state case, the storing state acts globally across all calls of the input unitaries and is followed by an entangling measurement that consumes all available calls~\cite{Sedl_k_2019, Bisio_2010}. By contrast, the RUS protocol simply repeats the single-shot strategy, which means that we can have $k$ independent storing procedures. This is operationally easier to implement and also avoids consuming all uses of the unitaries if the early success is achieved. Although our protocol can formally be viewed as a parallel realisation, it does not require all unitary calls to be used simultaneously.

Finally, the RUS protocol presented here has a close relation with the port-based remote preparation protocol analysed in~\cite{Muguruza_2024}, although with operational and conceptual differences. The relationship of SAR with known input states and port-based remote preparation~\cite{Muguruza_2024} is discussed in Sec.~\ref{sec:Garazi}. 

\subsection{Probabilistic unitary transposition, conjugation, and inversion with multiple calls}\label{sec:application_k>1}

In the following, we will focus solely on the case of probabilistic exact implementations. Moreover, we will restrict to parallel strategies, as they offer the simplest, yet operationally appealing, case. As previously mentioned, parallel strategies also provide a lower bound for all other strategies. Therefore, the following results can be considered as lower bounds for sequential and general tasks as well.

\paragraph{Parallel transposition} As stated in Thm.~\ref{thm:SAR_transp_Eq}, transposition is equivalent to SAR and so we have a construction for transposition also in terms of RUS strategy.

\paragraph{Parallel conjugation} For the unknown state case, Ref.~\cite{Quintino2019prob} proved that if \(k < d-1\), any probabilistic protocol necessarily has zero success probability. Here we establish a strictly stronger result: even when the input state is known, if \(k < d-1\), any probabilistic protocol still has zero success probability. Moreover, since deterministic and exact unitary conjugation is possible for \(k = d-1\), even in the unknown state case~\cite{Miyazaki_2017}, we see a fundamental contrast: when \(k < d-1\) the success probability is \(p=0\), whereas for \(k = d-1\) it is \(p=1\).

\begin{restatable}{theorem}{ThmConjgk}\label{thm:conj_k}
Any probabilistic protocol (parallel, sequential or general) transforming $k< d-1$ calls of the $d-$dimensional unitary operation $U \in \mathrm{SU}(d)$ into conjugation operation $\overline{U}$, with some probability $p$ independent of $U$, on a known (but arbitrary) input state $\ket{\psi}\in\mathbb{C}^d$ necessarily has probability $p=0$. 
\end{restatable}

\paragraph{Parallel inversion} Similarly to the unitary conjugation case,  Ref.~\cite{Quintino2019prob} proved that when \(k < d-1\), any probabilistic inversion protocol necessarily has zero success probability. Here we prove a strictly stronger result: even when the input state is known, if \(k < d-1\), any probabilistic protocol still has zero success probability. We remark that the proof for the unknown state case in Ref.~\cite{Quintino2019prob} relies on the fact that unitary inversion can be achieved by composing a unitary conjugation protocol with a unitary transposition protocol. In the known input case considered here, such a composition is not possible, since the protocols transform unitary operations into quantum states (as opposed to unitary operations into other unitary operations). This necessitates a different proof method from the one used in Ref.~\cite{Quintino2019prob}.

\begin{restatable}{theorem}{ThmInvgk}
Any probabilistic protocol (parallel, sequential or general) transforming $k< d-1$ calls of the $d-$dimensional unitary operation $U \in \mathrm{SU}(d)$ into inversion operation $U^{-1}$, with some probability $p$ independent of $U$, on a known (but arbitrary) input state $\ket{\psi}\in\mathbb{C}^d$ necessarily has probability $p=0$. 
\end{restatable}

Note that when \(d=2\) (qubits), an single-call RUS strategy can be directly applied, yielding the same probability as in the transposition and SAR cases. As mentioned in Sec.~\ref{sec:inversion}, this follows from the fact that the inversion protocol for qubits is equivalent to transposition: for \(d=2\), conjugation can always be achieved perfectly by applying \(Y\) gates around the unitary, as illustrated in Fig.~\ref{fig:conjgKNpure}, thereby converting the transposition strategy into inversion. On the other hand, when \(d>2\) and \(k \geq d-1\), the a direct implementation of the single-call RUS strategy for inversion yields \(p_{\textup{RUS,inv}} = 0\), since the single-call strategy for inversion with \(d>2\) has zero success probability. 
However, for \(k \geq d-1\) with \(d>2\), a slight modification of the RUS strategy can be employed. Specifically, since deterministic and exact conjugation is possible with \(k' = d-1\) calls (where \(k'\) denotes the number of calls required for conjugation), the available calls for unitary inversion can be grouped into sets of \(d-1\). Hence, instead of repeating single call strategy, we proceed by repeating $d-1$ strategy, obtaining the probability:
\begin{align}
    p_{\textup{inv,k}} = 1 - \left(1- \frac{1}{2\norm{\Tr_{A}(\ketbra{\psi}_{AB} )}}\right)^{\left\lfloor \tfrac{k}{d-1} \right\rfloor}.
\end{align}
This probability can be regarded as a lower bound for the inversion protocol 
with multiple calls, with $k\geq d-1$.

\subsection{Numerical results}\label{sec:numerics_main}

In this section, we present numerical results obtained by implement the SDP formulation of the problems in Matlab using the interpreter cvx~\cite{cvx} and the solver MOSEK solver~\cite{mosek}, all code used in this work is available on at the online repository of Ref.~\cite{MTQ_github_HOQO_KS}. 
Table~\ref{tab:transposition} and Table~\ref{tab:inversion} show the optimal success probabilities and optimal average fidelities, respectively, for the transposition and inversion protocol, considering $k \in \{2,3\}$ uses of the input operation and dimensions $d \in \{2,3\}$. 
\begin{table}[!ht]
\centering
\renewcommand{\arraystretch}{1.4}
\setlength{\arrayrulewidth}{0.3mm}
\begin{subtable}[t]{0.46\textwidth}
\centering
\begin{tabular}{c|c|c|c}
\multicolumn{4}{c}{\( \boldsymbol{p^{\textup{max}}_{\textup{trans}}(k = 2, d)}\)} \\
\hline
\(d\) & \textbf{Parallel} & \textbf{Adaptive} & \textbf{General} \\
\hline
2 & 0.7500 & 0.8334  & 0.8334 \\
3 & 0.5556 & 0.6112 & 0.6112\\
\specialrule{0.8pt}{0.8pt}{0.8pt}
\multicolumn{4}{c}{\(\boldsymbol{p^{\textup{max}}_{\textup{trans}}(k = 3, d)}\)} \\
\hline
\(d\) & \textbf{Parallel} & \textbf{Adaptive} & \textbf{General} \\
\hline
2 &0.8750 & 1& 1 \\
\hline
\end{tabular}
\subcaption{Optimal success probabilities for the transposition protocol}
\label{tab:numerics_prob}
\end{subtable}
\hspace{2em}
\begin{subtable}[t]{0.46\textwidth}
\centering
\begin{tabular}{c|c|c|c}
\multicolumn{4}{c}{\(\boldsymbol{\expval{F^{\textup{max}}_{\textup{trans}}}(k = 2, d)}\)} \\
\hline
\(d\) & \textbf{Parallel} & \textbf{Adaptive} & \textbf{General} \\
\hline
2 & 0.9396 & 0.9671 & 0.9671 \\
3 & 0.7530 & 0.8190 & 0.8190 \\
\specialrule{0.8pt}{0.8pt}{0.8pt}
\multicolumn{4}{c}{\(\boldsymbol{\expval{ F^{\textup{max}}_{\textup{trans}}}(k = 3, d)}\)} \\
\hline
\(d\) & \textbf{Parallel} & \textbf{Adaptive} & \textbf{General} \\
\hline
2 & 0.9769 & 1& 1 \\
\hline
\end{tabular}
\subcaption{Optimal average fidelities for the transposition protocol}
\label{tab:numerics_det}
\end{subtable}
\caption{Numerical results for the transposition protocol with $k \in \{2,3\}$ calls and $d \in \{2,3\}$}
\label{tab:transposition}
\end{table}


\begin{table}[!ht]
\centering
\renewcommand{\arraystretch}{1.4}
\setlength{\arrayrulewidth}{0.3mm}

\begin{subtable}[t]{0.46\textwidth}
\centering
\begin{tabular}{c|c|c|c}
\multicolumn{4}{c}{\( \boldsymbol{p^{\textup{max}}_{\textup{inv}}(k = 2, d)}\)} \\
\hline
\(d\) & \textbf{Parallel} & \textbf{Adaptive} & \textbf{General} \\
\hline
2 & 0.7500 & 0.8334  & 0.8334 \\
3 & 0.3333 & 0.3333 & 0.3333 \\
\specialrule{0.8pt}{0.8pt}{0.8pt}
\multicolumn{4}{c}{\(\boldsymbol{p^{\textup{max}}_{\textup{inv}}(k = 3, d)}\)} \\
\hline
\(d\) & \textbf{Parallel} & \textbf{Adaptive} & \textbf{General} \\
\hline
2 & 0.8750 & 1 & 1\\
\hline
\end{tabular}
\subcaption{Optimal success probabilities for the inversion protocol}
\label{tab:numerics_prob2}
\end{subtable}
\hspace{2em}
\begin{subtable}[t]{0.46\textwidth}
\centering
\begin{tabular}{c|c|c|c}
\multicolumn{4}{c}{\(\boldsymbol{\expval{ F^{\textup{max}}_{\textup{inv}}}(k = 2, d)}\)} \\
\hline
\(d\) & \textbf{Parallel} & \textbf{Adaptive} & \textbf{General} \\
\hline
2 & 0.9396 & 0.9671 & 0.9671\\
3 & 0.6710 & 0.7443 & 0.7443 \\
\specialrule{0.8pt}{0.8pt}{0.8pt}
\multicolumn{4}{c}{\(\boldsymbol{\expval{F^{\textup{max}}_{\textup{inv}}}(k = 3, d)}\)} \\
\hline
\(d\) & \textbf{Parallel} & \textbf{Adaptive} & \textbf{General} \\
\hline
2 & 0.9769 & 1 & 1 \\
\hline
\end{tabular}
\subcaption{Optimal average fidelities for the inversion protocol}
\label{tab:numerics_det2}
\end{subtable}

\caption{Numerical results for the inversion protocol with $k \in \{2,3\}$ calls and $d \in \{2,3\}$}
\label{tab:inversion}
\end{table}

As discussed in Sec.~\ref{sec:probabilistic}, in order to cast the maximisation problem for probabilistic exact transformations presented in Eq.~\eqref{Eq:prob_SDP1} as an SDP, we followed the steps of Ref.~\cite{Quintino2019prob} and replaced the infinite set of unitaries by a finite spanning set. Concretely, one can always find a finite collection of unitaries $\{U_{d,i}\}$ such that  
\begin{align}
\operatorname{Span}\!\left(\dketbra{U_d^{\otimes k}}{U_d^{\otimes k}} \,\middle|\, U_d \ \mathrm{unitary}\right)
= \operatorname{Span}\!\left(\dketbra{U_{d,i}^{\otimes k}}{U_{d,i}^{\otimes k}} \,\middle|\, U_{d,i} \in \{U_{d,i}\}\right).
\end{align}
Explicitly constructing such a basis is in general not straightforward, but for numerical purposes one can sample unitaries~$U_d$ uniformly at random with respect to the Haar measure. Because the subspace is finite dimensional, drawing enough unitaries at random from the Haar measure will almost surely yield a linearly independent spanning set. Checking independence is efficient, so we can iteratively sample and collect unitaries until we have a spanning set. This random sampling method lets us approximate the universal constraints well enough for numerical exploration and suffices for the purposes of this paper. Also, from the same argument present in Refs.~\cite{Quintino2019prob,Quintino_2022det}, we may assume without loss of performance that the Choi operator of the superchannels and superinstruments in the optimisation have real numbers. Finally, in principle on may obtain a more efficient code formulation of the problem by explicitly making use of the symmetries and by writing the constraints in a block diagonal form, similarly to what is done in Refs.~\cite{yoshida_exact,grinko2023gelfandtsetlinbasispartiallytransposed}. Since our main purpose here is to understand the behaviour of the problem for small values of $k$ and $d$ to develop intuition for proving analytical results, we leave this numerical optimisation possibility for future research.

As listed in Sec.~\ref{sec:numerics_intro}, the numerical tables allow us to draw several conclusions.  
First, let us notice that for $k=3$ and $d=2$, both unitary transposition and inversion on an arbitrary known state can be implemented deterministically and exactly, even with a pure input. This is in parallel to the universal input scenario where $k=4$ queries were required at $d=2$ to achieve deterministic exact realisation on a pure input state~\cite{yoshida_exact}. Secondly, for $k\in\{2,3\}$ and $d\in\{2,3\}$, we observe that the repeat-until-success strategy, discussed in Sec.~\ref{sec:SAR_k}, with success probability $p_{\mathrm{RUS}} = 1 - \left(1 - \tfrac{1}{d}\right)^{k}$ is optimal for both transposition and SAR. Thirdly, within the tested regime, general strategies seem to provide no advantage over sequential ones, unlike in the unknown input case where indefinite causal order can outperform sequential protocols~\cite{Quintino_2022det,Quintino2019prob}. Finally, while for unknown inputs deterministic parallel transposition and inversion are equivalent and attainable via estimation~\cite{Quintino_2022det,Bisio_2010}, this equivalence breaks in the known input setting. Moreover, unlike the unknown input case where sequential protocols cannot beat parallel ones when $k<d$~\cite{Yoshida2024}, here deterministic sequential protocols can indeed surpass parallel ones even in that regime.

\section{Comparison between storage-and-retrieval with known states and port-based state preparation}\label{sec:Garazi}
\begin{figure}[t!]
  \centering
  \begin{subfigure}[b]{0.45\textwidth}
    \includegraphics[width=\linewidth]{figures/SAR_k.png}
    \captionsetup{width=0.9\textwidth}
    \caption{ Circuit representation of probabilistic unitary SAR with $k$ calls of the input unitary and known bipartite input state, as considered and described in this work.}
    \label{fig:SAR_k_compar}
  \end{subfigure}
    \quad \quad 
  \begin{subfigure}[b]{0.45\textwidth}
    \includegraphics[width=\linewidth]{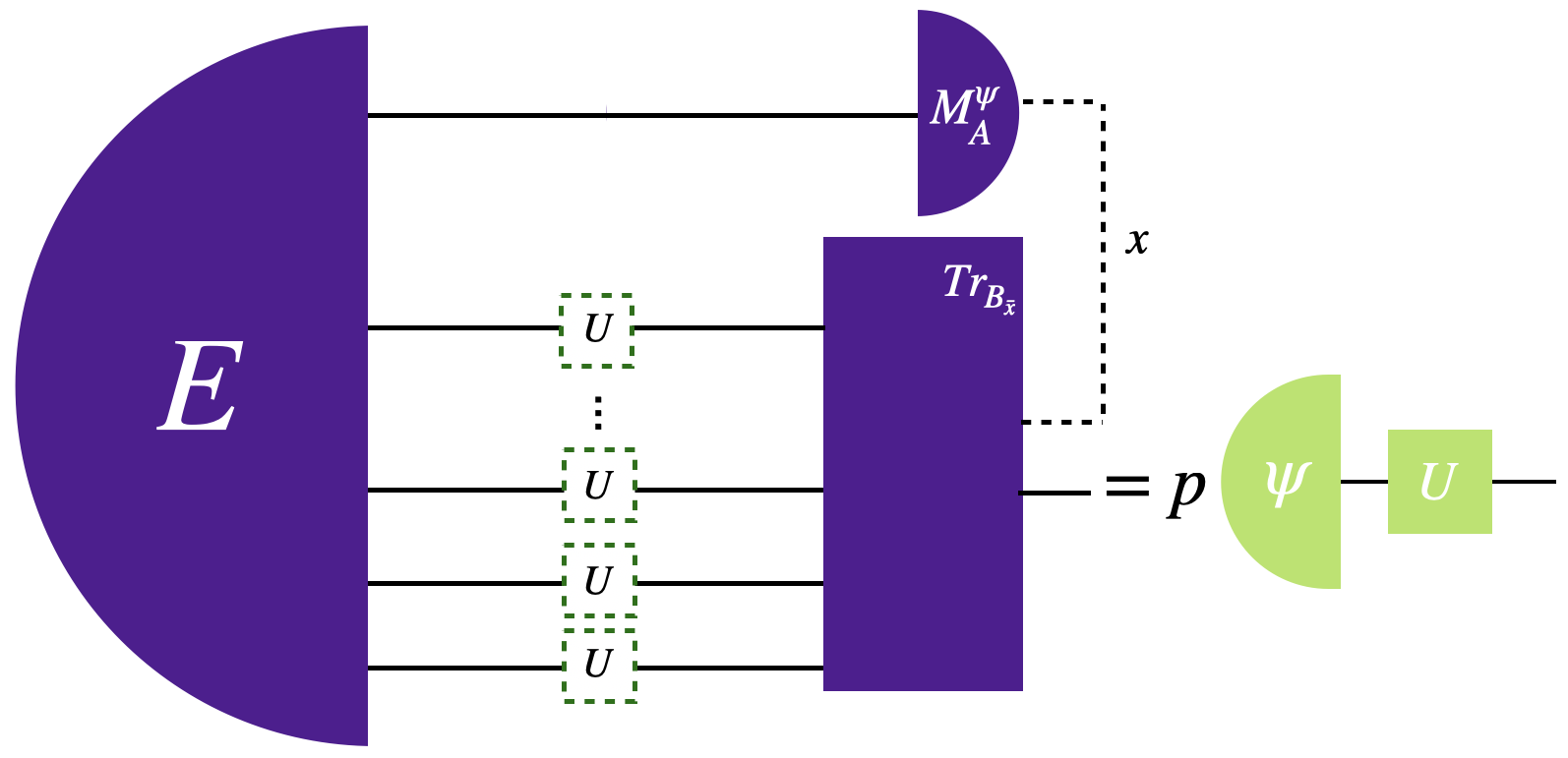}
    \captionsetup{width=0.9\textwidth}
    \caption{ Circuit representation of PBSP~\cite{Muguruza_2024} where Bob performs a unitary operation before the partial trace. The top wire where a measurement is performed represents Alice’s system and the others wires represent Bob's.}
    \label{fig:PBSPgeneral}
  \end{subfigure}
  \captionsetup{width=0.9\textwidth}
  \caption{Comparison between general SAR protocol with known bipartite input state and port-based state preparation as discussed in~\cite{Muguruza_2024}. We see from Fig.~\ref{fig:PBSPgeneral} that PBSP can be viewed as a particular instance of SAR, Fig.~\ref{fig:SAR_k_compar}. Namely, the decoder is restricted to a specific structure and the known bipartite input state used by our SAR protocol is separable. }
  \label{fig:comparisonSAR_PBSP}
\end{figure}

We now discuss the relationship between our results on SAR and Ref.~\cite{Muguruza_2024}, where the authors introduce the task of port-based state preparation (PBSP). Inspired by remote state preparation~\cite{Bennett_2001}, PBSP is a variant of port-based teleportation (PBT) where the state to be teleported is known to the party performing the teleportation. In the PBSP scenario, two parties, Alice and Bob, share an entangled qudit state $E \in \mathcal{L}\big((\mathbb{C}^d \otimes \mathbb{C}^d)^{\otimes k}\big)$, where each party holds $k$ subsystems, referred to as ports. To teleport the known target state $\ket{\psi}$ to Bob, Alice performs a joint measurement on her share of the system, where the measurement can now depend on $\ket{\psi}$. After obtaining the outcome, she communicates the classical result to Bob. As in port-based teleportation (PBT)~\cite{Mozrzymas2018PBT, Ishizaka2008TelepProgr}, Bob uses this classical information to select one of his $k$ qudits: he discards the rest and keeps the chosen port. Also as in PBT, Bob does not apply any additional correction to the teleported state; the protocol is designed so that the state appears directly in the selected port. As illustrated in Fig.~\ref{fig:comparisonSAR_PBSP} and noted in Ref.~\cite{Muguruza_2024}, any PBSP protocol can be seen as a special case of a SAR protocol in which the decoder depends on the known input state. The key difference is that a general SAR protocol allows the decoding operation to act jointly on all subsystems, as shown in Fig.~\ref{fig:SAR_k_compar}; therefore, not every SAR protocol can be reduced to a PBSP one.

Under the assumption that the parties share $k$ copies of maximally entangled qudits\footnote{Reference~\cite{Muguruza_2024} leaves as an open question whether a non-maximally entangled state can outperform maximally entangled ones for this task. Note that for PBT it is known that non-maximally entangled states do outperform maximally entangled ones~\cite{Studzinski_2017PBT, Mozrzymas2018PBT}.}, $\phi = \ketbra{\phi^{+}}{\phi^{+}}^{\otimes k}$, Ref.~\cite{Muguruza_2024} shows that the optimal probability of success is $p = 1 - \left(1 - \frac{1}{d}\right)^{k}$. Recall that the RUS strategy we present for retrieving the action of an arbitrary $d$-dimensional unitary operation on a bipartite state $\ket{\psi}_{AB}$ is
\begin{equation}
    p_{\textup{SAR}}^{\textup{RUS}}(k,d) = 1 - \left(1 - \frac{1}{d\norm{\Tr_{\A}(\ketbra{\psi}{\psi}_{AB} )}_{\textup{op}}}\right)^{k} \geq 1 - \left(1 - \frac{1}{d}\right)^{k},
\end{equation}
a value that coincides with the optimal PBSP for maximally entangled states used in the encoding phase and when $\ket{\psi_{AB}}$ is separable.  
\begin{figure}[t!]
    \centering
\includegraphics[width=0.58\textwidth]{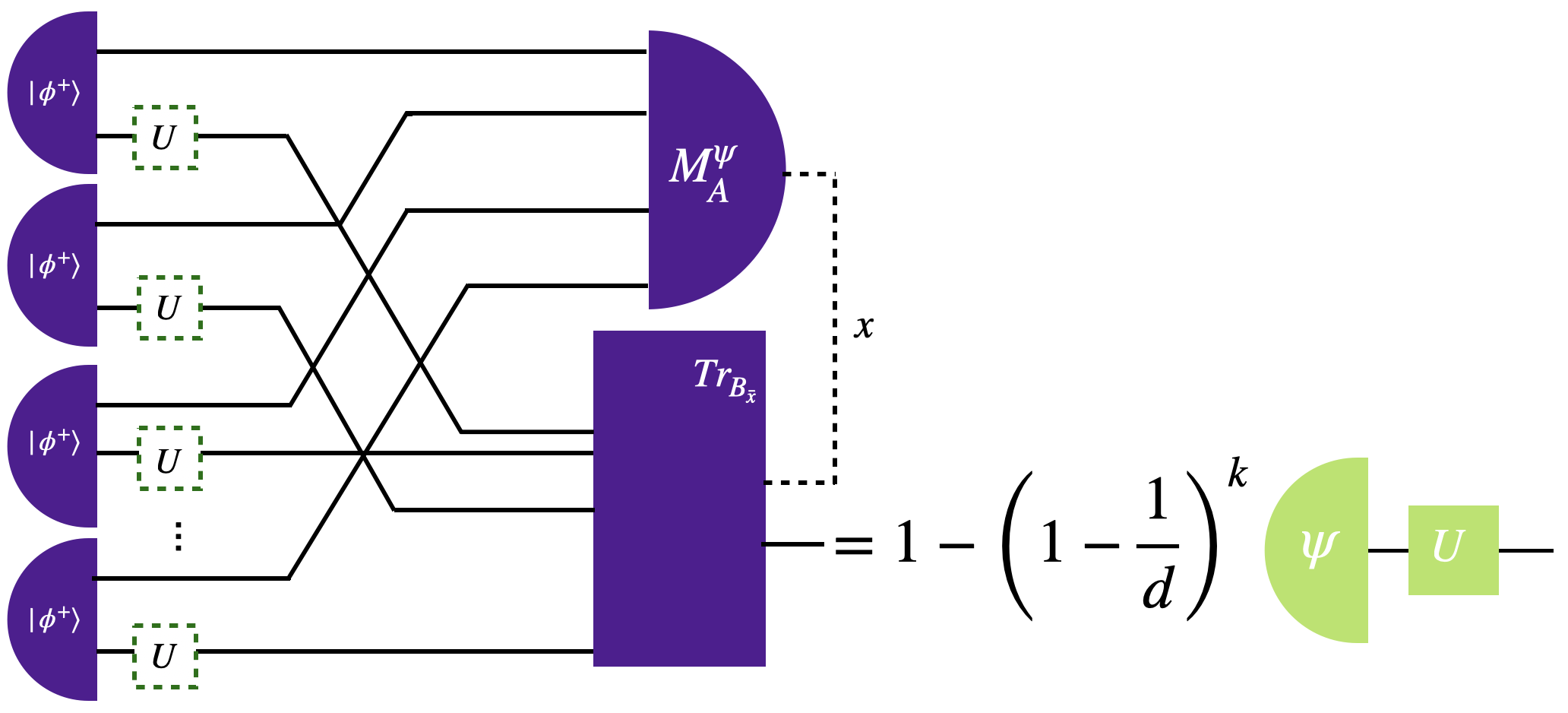}
    \caption{ Port-based state preparation protocol with maximally entangled encoding state presented in Ref.~\cite{Muguruza_2024}. Alice has the classical description of a pure state $\psi$ that she wishes to prepare on Bob's side with output probability $p$. Alice performs a measurement parametrised by the knowledge of $\psi$, followed by a single round of classical communication to Bob. Bob then traces out all subsystems except the $x$ register. }
    \label{fig:PBSP}
\end{figure}

We observe that, while the success probability for the pure input state RUS presented here matches that of PBSP in Ref.~\cite{Muguruza_2024}, there is a major operational difference between these two tasks. The PBSP protocol of~\cite{Muguruza_2024} consists of a single joint measurement that “consumes’’ all $k$ resource states $(\id\otimes U)\ket{\phi^+}$, whereas in the RUS strategy proposed here we perform $k$ independent qudit measurements. If success is achieved, part of the storing states remains available for a possible subsequent protocol. This makes our protocol simpler to implement and less resource-consuming. First, it avoids joint measurements, meaning that less entanglement is required in the decoding step. Second, because the success probability grows exponentially with the number of trials, for large $k$ only a small portion of the available states is typically consumed.

Since probabilistic PBSP can be viewed as a restriction of probabilistic unitary SAR, any upper bound on probabilistic SAR also applies to PBSP. 
In Sec.~\ref{sec:SAR}, we showed that in the single-call case the optimal probability of success for unitary SAR is exactly 
$p_{\textup{SAR}}^{\textup{RUS}}(k=1,d)=\frac{1}{d}$. 
This immediately implies that, when $k=1$, the optimal performance of PBSP is also $p_{\textup{PBSP}}(k=1,d)=\frac{1}{d}$, regardless of the initial resource state shared by the parties. 
Furthermore, in Sec.~\ref{sec:numerics_main}, numerical evidence suggests that 
$p_{\textup{SAR}}(k,d) \leq 1 - \left(1 - \frac{1}{d}\right)^{k}$. 
Since the success probability of unitary SAR upper bounds that of PBSP, our numerical results also suggest that 
$p_{\textup{PBSP}}(k,d) \leq 1 - \left(1 - \frac{1}{d}\right)^{k}$, again independently of the initial shared resource state.

Finally, we note that Ref.~\cite{Muguruza_2024} also considers a deterministic approximate version of PBSP, using the average and worst-case fidelity as figures of merit. However, due to the different nature of the problems and how they are formulated, the optimal average fidelities for deterministic PBSP are strictly smaller than those for unitary SAR. For instance, Ref.~\cite{Muguruza_2024} proves that when Alice and Bob share a maximally entangled state, the relation $p=\expval{F}$ holds: the probability of success coincides exactly with the average fidelity of the protocol. In our unitary SAR setting, it is straightforward to see that this relation does not hold. Namely, by analogy with the deterministic single-call case discussed in Sec.~\ref{sec:transp_det_single}, we can always convert a probabilistic protocol into a deterministic one by attempting a correction whenever failure occurs. This construction leads to an average fidelity
\begin{equation}
    \expval{F_{\textup{tot}}} = p_s \cdot 1 + (1 - p_s) \expval{F_{\textup{fail}}},
\end{equation}
which is strictly greater than $p_s$ whenever $p_s<1$.

\section{Discussion}

In this work, we explored higher-order quantum operations in a setting where the input operation is fully universal and unknown, but the state on which the transformation is to be implemented is assumed to be completely known. This allows the structure of the supermap, specifically, the encoder and decoder in the protocol, to be tailored to the given state. We show that exploiting this classical information enables enhanced performance in certain tasks, while in others, somewhat surprisingly, it provides no advantage.

At the level of formalism, this setting gives rise to a distinct class of symmetries compared to the universal input state case. We distinguished between input state independent and input state dependent symmetries (Secs.~\ref{sec:sym_prob} and \ref{sec:sym_det}), showing that classical knowledge of the input state fundamentally reshaped the structure of the problem. 

Our main findings can be summarised as follows:
\begin{itemize}
    \item For the probabilistic exact protocols that involve only a single call to the input unitary, we established that the implementations of transposition, storage-and-retrieval (SAR), and inversion exhibit a quadratic improvement compared to the case where the input state is universal and unknown. Similarly, in the deterministic approximate setting, all of the $\mathrm{SU}(d)$ homomorphisms and antihomomorphisms exhibit quadratic improvement in optimal average fidelity, compared to the universal case.
    \item Interestingly, we found that not all tasks in probabilistic exact realisation exhibit performance improvement. In particular, for protocols implementing probabilistic exact unitary conjugation, knowledge of the input state offers no advantage. This negative result shows that state knowledge is not a universal resource: structural obstructions rooted in the representation-theoretic properties of the unitary group remain, and no amount of state tailoring can overcome them in this case.
    \item In the case of protocols that implement a desired transformation on part of a known bipartite state, the amount of entanglement between the two subsystems directly influences the attainable performance. We proved that the success probability of the probabilistic exact protocol increases monotonically with the degree of entanglement, interpolating smoothly between the single-state result (with success probability $1/d$) and the maximally entangled case, which allows unit probability of success.
    \item When considering mixed states of the form $\rho = \eta \ketbra{\psi} + \frac{1-\eta}{d}\mathbb{1}$, we showed that it is possible to construct deterministic and exact protocols, provided that the visibility $\eta$ falls below a critical threshold. Interestingly, this threshold is determined by the optimal average fidelity achieved in the corresponding pure-state protocol. In other words, the quality of approximate pure-state implementations sets the boundary for exact realisations on noisy, depolarized states. This connection between approximate and exact regimes deepens the structural understanding of HOQO and provides a practical criterion for exact implementability.  
    \item In the storage-and-retrieval problem with multiple calls to the input unitary, the known state assumption allows for a repeat-until-success (RUS) strategy, which provides an exponential advantage compared to the corresponding unknown input state SAR protocol. Specifically, the success probability grows as $1 - (1-1/d)^{k}$, approaching unity exponentially fast in the number of calls $k$. This result is important from an operational perspective, since the RUS strategy relies only on the single shot protocol and is therefore simpler to implement than more general strategies that jointly process $k$ uses of the input unitary. Moreover, it allows us to preserve unused calls to the unitary whenever the protocol succeeds before all $k$ calls are consumed, providing a clear advantage from the resource point of view.
    \item While the previous conclusions follow from analytical results, we additionally present two conclusions based solely on numerical results: 
    \begin{itemize}
        \item Numerically, we saw that, for the case $k=3$ and $d=2$, both unitary transposition and inversion on an arbitrary known state can be achieved deterministically and exactly, even in the pure input state setting. This is in sharp parallel with the universal input state scenario, where $k=4$ queries and $d=2$ was required to reach deterministic and exact implementation. 
        \item Our numerical results indicate that for \(k \in \{2,3\}\) and \(d \in \{2,3\}\), the repeat-until-success strategy is optimal in both the transposition and SAR tasks.
        \item In the tested regime of \(k \in \{2,3\}\) calls to the input unitary and state dimensions \(d \in \{2,3\}\), general strategies offer no advantage over sequential ones. This contrasts with the unknown input state case, where indefinite causal order can outperform sequential protocols~\cite{Quintino_2022det,Quintino2019prob}.
    \end{itemize}   
\end{itemize}

From these results we see that, on the one hand, classical knowledge of the input state fundamentally changes the problem and allows for, in principle, higher performances of higher-order quantum operations. Moreover, the role of entanglement and noise demonstrates that the advantage depends crucially on the structure of the known input. 

These findings open several directions for future research. They represent a limiting case of higher-order quantum operations (HOQO) under varying non-universal assumptions. Other non-universal assumptions within HOQO include partial knowledge of the input state, gray-box protocols with partial information about the input unitary, and fine-tuned combinations of relative knowledge about the input state and the input unitary. For instance, Ref.~\cite{Grosshans2024Multicopy} considers a port-based teleportation scenario in which Alice has \(n\) copies of the state she aims to teleport. The case \(n=1\) corresponds to standard PBT, while the limit \(n \to \infty\) corresponds to port-based state preparation~\cite{Muguruza_2024}. In our terminology, this multi-copy regime, can be interpreted as a scenario where the \(n\) available copies provide partial knowledge of the input state, with the known input state case recovered in the limit where the number of unknown copies tends to infinity. 

Let us also remark that the assumption of a known input state is the particular case of the setup where the input wire is restricted to act on a certain subset. Namely, a wire in a quantum circuit usually represents the full set of states $\rho \in \{ \sigma \in \mathcal{L}(\mathcal{H}) \mid \sigma \ge 0, \mathrm{Tr}(\sigma)=1 \}$. Our framework can instead be understood as the limiting case where the input wires in the circuit are restricted to a subset of possible states $\sigma \in \{\rho_i\}_i$. In the present work, this set was restricted to only a single element $\sigma \in \{\rho\}$; however, generalisations to larger sets constitute an interesting direction for future investigation.

Equally crucial is the problem of compositionality: once universality is abandoned, the output of one supermap may no longer serve as a valid input for another unless the associated side information is carefully tracked. Developing a systematic theory of composition for non-universal HOQO therefore remains an open challenge. 

In conclusion, our work demonstrates that combining the universality of the input operation with classical knowledge of the input state unlocks a rich new regime of higher-order quantum computation. It provides both concrete performance gains and conceptual insights into the structure of HOQO, while also raising foundational questions about varying universality assumptions and the interplay between quantum and classical resources in higher-order frameworks.

\section{Acknowledgments}
We are grateful to Akihito Soeda for useful discussions regarding the proof of Thm.~1 of Ref.~\cite{Quintino2019prob}.
VB and MTQ acknowledge support by QuantEdu France, a State aid managed by the French National Research Agency for France 2030 with the reference ANR-22-CMAS-0001. MTQ acknologes the funding Tremplins nouveaux entrants \& nouvelles entrantes - Edition 2024, project HOQO-KS.
SY acknowledges support by Japan Society for the Promotion of Science (JSPS) KAKENHI Grant Number 23KJ0734, FoPM, WINGS Program, the University of Tokyo, and DAIKIN Fellowship Program, the University of Tokyo. MM acknowledges the MEXT Quantum Leap Flagship Program (MEXT QLEAP) JPMXS0118069605, JPMXS0120351339, the JSPS KAKENHI Grant Numbers 21H03394 and 23K21643, and IBM Quantum.
We acknowledge the Japanese-French Laboratory for Informatics (JFLI) for
the support on organising the Japanese-French Quantum
Information 2023 workshop.

\clearpage
\appendix
\addcontentsline{toc}{section}{Appendix}
\section*{Appendix}

\section{Proofs for the formalism of probabilistic exact realisation}

In this section we will prove the main points discussed in Sec.~\ref{sec:sym_prob}.

\IndependenceThmProb*

\begin{proof}
Let us first consider the case of $f$ being a homomorphism. Given $U \mapsto f(U)\ket{\psi}$ for all $U \in \mathrm{SU}(d)$ and some $\ket{\psi}\in \mathbb{C}^{d}$, it follows that $VU \mapsto f(VU)\ket{\psi} = f(V)f(U)\ket{\psi}$. Since $f$ is surjective, for all $\ket{\psi} \in \mathbb{C}^{d}$ there exists a unitary $V \in \mathrm{SU}(d)$ such that $\ket{\psi} = f(V)\ket{0}$. Therefore, $S^{\ket{0}}$ satisfies,
\begin{equation}
\begin{aligned}
 S^{\ket{0}}_{\I\O\F} * \dketbra{VU}{VU}^{\otimes k} &= p f(UV)\ketbra{0}{0 } f(UV)^{\dagger}\\
    &= p f(U)f(V)\ketbra{0}f(V)^{\dagger} f(U)^{\dagger} \\
    &= p f(U)\ketbra{\psi}{\psi} f(U)^{\dagger} =  S^{\ket{\psi}}_{\I\O\F} * \ChoiU^{\otimes k}.
\end{aligned}
\end{equation}

Analogously, we consider the case of $f$ being an antihomomorphism. Given $U \mapsto f(U)\ket{\psi}$ for all $U$ and some $\ket{\psi}$, it follows that $UV \mapsto f(UV)\ket{\psi} = f(U)f(V)\ket{\psi}$. Therefore, we obtain
\begin{equation}
\begin{aligned}
 S^{\ket{0}}_{\I\O\F} * \dketbra{UV}{UV}^{\otimes k} &= p f(VU)\ketbra{0}{0 } f(VU)^{\dagger}\\
    &= p f(U)f(V)\ketbra{0}{0} f(V)^{\dagger} f(U)^{\dagger} \\
    &= p f(U)\ketbra{\psi}{\psi} f(U)^{\dagger} =  S^{\ket{\psi}}_{\I\O\F} * \ChoiU^{\otimes k}.
\end{aligned}
\end{equation}

\end{proof}

\TwirlProbConstr*

\begin{proof}
 We begin by showing that, the constraint
\begin{equation}\label{eqv:proof_prob_twirl1}
     S^{c}_{\I\O\F}*\dketbra{U}{U}^{\otimes k}_{\I\O} = p\, f(U)\ketbra{\psi}{\psi}_{\F}f(U)^{\dagger} \quad \forall\, U\in\mathrm{SU}(d),
\end{equation}
implies the constraint
\begin{equation}   
   \map{\tau}(S)_{\I\O\F}*\dketbra{U}{U}^{\otimes k}_{\I\O} = p\, f(U)\ketbra{\psi}{\psi}_{\F}f(U)^{\dagger} \quad \forall\, U\in\mathrm{SU}(d),
\end{equation}
where the map \(\map{\tau}\) denotes a \emph{twirled supermap}.
The twirling is defined as follows:
\begin{itemize}
    \item For homomorphic \(f\),
    \begin{align}
\map{\tau}(S)\coloneqq\int {\left[ \mathbb{1}_{\textup{I}} \otimes \overline{B}^{\otimes k}_{\textup{O}} \otimes f(B)_{\textup{F}}\right] S }  {\left[ \mathbb{1}_{\textup{I}} \otimes \overline{B}^{\otimes k}_{\textup{O}} \otimes f(B)_{\textup{F}}\right]^{\dagger}  \,\mathrm{d} B },
\end{align}

\item For antihomomorphic \(f\),
\begin{align}
\map{\tau}(S)\coloneqq\int {\left[ \overline{B}^{\otimes k}_{\textup{I}} \otimes \mathbb{1}_{\textup{O}} \otimes f(B)_{\textup{F}}\right] S } {\left[ \overline{B}^{\otimes k}_{\textup{I}} \otimes \mathbb{1}_{\textup{O}} \otimes f(B)_{\textup{F}}\right]^{\dagger} \,\mathrm{d} B },
\end{align}
\end{itemize}
where the integral is taken with respect to the Haar measure, here and throughout the paper.

Let us focus first on the case of transposition function; the reasoning for other functions is analogous. Since Eq.~\eqref{eqv:proof_prob_twirl1} holds for all \(U\), we can make the change of variables \(U \mapsto UB\):
\begin{align}\label{eqv:proof9_1}
     S_{\I\O\F}*\dketbra{UB}{UB}^{\otimes k}_{\I\O} &= p\, B^T U^T\ketbra{\psi}{\psi}_{\F}\overline{U}\,\overline{B} \quad \forall\, U, B\in\mathrm{SU}(d).
\end{align}
Let us rewrite the right-hand side explicitly:
\begin{align}
     S_{\I\O\F}*\dketbra{UB}{UB}^{\otimes k}_{\I\O} &= \Tr_{\I \O}\!\left(S_{\I\O\F}\left( \mathbb{1}_{\I} \otimes  (\overline{U^{ \otimes k}B^{ \otimes k}})_{\O} \right)\dketbra{\mathbb{1}}{\mathbb{1}}^{\otimes k}_{\I\O}\left( \mathbb{1}_{\I} \otimes  \left(U^{ \otimes k}B^{ \otimes k}\right)^T_{\O} \right) \otimes \mathbb{1}_{\F}\right)\\
     &= \Tr_{\I \O}\!\left( \big( B^{ \otimes k}_{\I} \otimes  \mathbb{1}_{\O} \big)S_{\I\O\F}\big( (B^{ \otimes k}_{\I})^{\dagger} \otimes  \mathbb{1}_{\O} \big)\big(  \dketbra{\overline{U}}{\overline{U}}_{\I\O}^{ \otimes k} \otimes \mathbb{1}_{\F} \big)\right).
\end{align}
Substituting this back into Eq.~\eqref{eqv:proof9_1}, we obtain
\begin{align}
    \Tr_{\I \O}\!\left( \big( B^{ \otimes k}_{\I} \otimes  \mathbb{1}_{\O} \big)S_{\I\O\F}\big( (B^{ \otimes k}_{\I})^{\dagger} \otimes  \mathbb{1}_{\O} \big)\big(  \dketbra{\overline{U}}{\overline{U}}_{\I\O}^{ \otimes k} \otimes \mathbb{1}_{\F} \big)\right) &= p\, B^TU^T\ketbra{\psi}{\psi}_{\F}\overline{U}\,\overline{B}.
\end{align}
By multiplying from the left the equation above by $\overline{B}$ and right with the $B^{T}$, we obtain,
\begin{align}
    \Tr_{\I \O}\!\left(\big( B^{ \otimes k}_{\I} \otimes  \mathbb{1}_{\O} \otimes \overline{B}_{\F} \big)S_{\I\O\F} \big( (B^{ \otimes k}_{\I})^{\dagger} \otimes  \mathbb{1}_{\O} \otimes B^T_{\F} \big)\dketbra{\overline{U}}{\overline{U}}_{\I\O}^{ \otimes k} \right) &= p\, U^T\ketbra{\psi}{\psi}_{\F}\overline{U}.
\end{align}
By taking the Haar integral of both sides over $B$, we obtain
\begin{align}
\map{\tau}(S)_{\I\O\F}*\dketbra{U}{U}^{\otimes k}_{\I\O} = p\, f(U)\ketbra{\psi}{\psi}_{\F}f(U)^{\dagger} \quad \forall\, U\in\mathrm{SU}(d),
\end{align}
where the twirled supermap is defined by
\begin{align}\label{eqv:twirl_T}
\map{\tau}(S)\coloneqq\int & {\left[ B^{\otimes k}_{\textup{I}} \otimes \mathbb{1}_{\textup{O}} \otimes \overline{B}_{\textup{F}}\right] S } {\left[ (B_{\textup{I}}^{ \otimes k})^{\dagger} \otimes \mathbb{1}_{\textup{O}}^{T } \otimes B_{\textup{F}}^{T}\right]  \mathrm{d} B }.
\end{align}
Due to the left- and right-invariance of the Haar measure, the twirled supermap satisfies the covariance condition:
\begin{align}
    [\map{\tau}(S), U^{\otimes k}_{\I} \otimes \mathbb{1}_{\O} \otimes \overline{U}_{\F}] = 0, \quad \forall U\in \mathrm{SU}(d),
\end{align}
which completes the proof of the first part of the theorem.

For the second part of the proof, we have the constraint
\begin{equation}   
   S^{c}_{\I\O\F}*\dketbra{\mathbb{1}}{\mathbb{1}}^{\otimes k}_{\I\O} = p\, \ketbra{\psi}{\psi}_{\F},
\end{equation}
where \(S^{c} \coloneqq \map{\tau}(S)\) is again the twirled supermap defined in Eq.~\eqref{eqv:twirl_T}, and hence $[S^{c}, U^{\otimes k}_{\I} \otimes \mathbb{1}_{\O} \otimes \overline{U}_{\F}] = 0, \quad \forall U\in \mathrm{SU}(d)$. We want to show that this implies the constraint,
\begin{equation}
     S^{c}_{\I\O\F}*\dketbra{U}{U}^{\otimes k}_{\I\O} = p\, f(U)\ketbra{\psi}{\psi}_{\F}f(U)^{\dagger}.
\end{equation}
We will again consider the case of transposition map \(f(U) = U^{T}\). 
\begin{align}
    &S^c_{\I\O\F} * \dketbra{U}{U}_{\I\O} 
=   \Tr_{\I\O}\!\left( S^c_{\I\O\F} \dketbra{U}{U}^{T}_{\I\O} \otimes \mathbb{1}_{\F})\right)  \\
&=   \Tr_{\I\O}\!\left( (A_{\I} \otimes \mathbb{1}_{\O} \otimes \overline{A}_{\F})S^c_{\I\O\F}(A_{\I} \otimes \mathbb{1}_{\O} \otimes \overline{A}_{\F})^{\dagger}(\dketbra{\overline{U}}{\overline{U}}_{\I\O} \otimes \mathbb{1}_{\F})\right)  \\
&=\Tr_{\I\O}\!\left( (A_{\I} \otimes \mathbb{1}_{\O} \otimes \overline{A}_{\F})S^c_{\I\O\F}(A^{\dagger}_{\I} \otimes \mathbb{1}_{\O} \otimes A^T_{\F})((\mathbb{1}_{\I} \otimes \overline{U}_{\O})\dketbra{\mathbb{1}}{\mathbb{1}}_{\I\O}(\mathbb{1}_{\I} \otimes U^{T}_{\O}) \otimes \mathbb{1}_{\F})\right) \\
&= \Tr_{\I\O}\!\left((A_{\I} \otimes \mathbb{1}_{\O} \otimes \overline{A}_{\F})S^c_{\I\O\F}(A^{\dagger}_{\I} \otimes \mathbb{1}_{\O} \otimes A^T_{\F})((U^{\dagger}_{\I} \otimes \mathbb{1}_{\O})\dketbra{\mathbb{1}}{\mathbb{1}}_{\I\O}(U_{\I} \otimes \mathbb{1}_{\O}) \otimes \mathbb{1}_{\F})\right)\\
&= \Tr_{\I\O}\!\left( ((UA)_{\I} \otimes \mathbb{1}_{\O} \otimes \overline{A}_{\F})S^c_{\I\O\F}((UA)^{\dagger}_{\I} \otimes \mathbb{1}_{\O} \otimes A^T_{\F})(\dketbra{\mathbb{1}}{\mathbb{1}}_{\I\O} \otimes \mathbb{1}_{\F})\right).
\end{align}
Since the equality above holds for every $A$, we may just set $A=U^\dagger$, so that \begin{align}
    S^c_{\I\O\F} * \dketbra{U}{U}_{\I\O} &= U^{T}_{\F}\Tr_{\I\O}\!\left( S^c_{\I\O\F}(\dketbra{\mathbb{1}}{\mathbb{1}}_{\I\O} \otimes \mathbb{1}_{\F})\right) \overline{U}_{\F}\\
    &= U^{T}_{\F}(p \ketbra{\psi}{\psi}) \overline{U}_{\F}.
\end{align}
This closes the proof for transposition map. 

The argument extends analogously to other homomorphic and antihomomorphic functions \(f\), and which proves the theorem.
\end{proof}

\KNsym*

\begin{proof}
Let us start with transposition function $f(U) = U^{T}$. We are going to construct a covariant operator $S^c$ which respects the commutation relation $[S^c,V_f]=0$ and that attains the same performance of $S$. For that, we define the linear operator 
\begin{equation}
    S^{c}_{\I\O\F} \coloneqq \int_{\mathcal{V}\subseteq SU(d)} \mathrm{d}V \,(\mathbb{1}_{\I} \otimes V^{T \otimes k}_{\O} \otimes \id_F S_{\I\O\F})   (\mathbb{1}_{\I} \otimes \overline{V}^{\otimes k}_{\O} \otimes \id_F ), \quad \quad \textup{ where } V^T\ket{\psi}=\ket{\psi}.
\end{equation}
Here $\mathcal{V} \subseteq \mathrm{SU}(d)$ denotes the subset of unitary operators 
$V \in \mathrm{SU}(d)$ satisfying $V^{T}\ket{\psi} = \ket{\psi}$ for a given $\ket{\psi} \in \mathbb{C}^{d}$.
Since the elements of this subset themselves 
form a compact group, $\mathcal{V}$ carries its own Haar measure, which we use for integration. Direct calculation shows that $S^c$ respects the commutation relation $[S^c,V_f]=0$. We now verify that for $S^{c}_{IOF}$ we obtain
\begin{align}
     S^{c}_{\I\O\F} * \ChoiU^{\otimes k} &= \Tr_{\I\O}\big((\mathbb{1}_{\I} \otimes V^{T \otimes k}_{\O}) S_{\I\O\F} (\mathbb{1}_{\I} \otimes \overline{V}^{\otimes k}_{\O})\dketbra{\overline{U}}{\overline{U}}^{\otimes k}_{\I \O} \otimes \mathbb{1}_{\F}\big) \\
     &= S_{\I\O\F} * \dketbra{V^{\otimes k}U^{\otimes k}}{V^{\otimes k}U^{\otimes k}}_{\I \O} \\
     &= p U^T V^T\ketbra{\psi} \overline{V} \overline{U} \\
     & = p U^{T}\ketbra{\psi}{\psi } \overline{U}.
\end{align}
The argument extends analogously to other homomorphic and antihomomorphic functions \(f\). This proves the theorem.

\end{proof}

\section{Proofs for the formalism of deterministic approximate realisation}\label{append:proofs_det}

In this section, we first demonstrate that the performance operator introduced for known input states in Sec.~\ref{sec:deterministic} serves as a valid tool for evaluating the average performance of a superchannel.

\BipDet*
\begin{proof}

We begin with the definition of the state fidelity between two bipartite states of the form
\begin{equation} 
    F(\ketbra{\phi_{U}}{\phi_{U}}_{\A\F},f(U)_{\F}\ketbra{\psi}{\psi}_{\A\F}f(U)_{\F}^{\dagger} = \bra{\phi_{U}}_{\A\F} f(U)_{\F}\ketbra{\psi}{\psi}_{\A\F}f(U)_{\F}^{\dagger} \ket{\phi_{U}}_{\A\F}. 
\end{equation} 
where $\ketbra{\phi_{U}}{\phi_{U}}_{\A\F}$ denotes a state that depends on the unitary $U$. The average fidelity is then given by
\begin{align} 
    \expval{F^{\psi}}\coloneqq \int \mathrm{d}U \,F(\ketbra{\phi_{U}}{\phi_{U}}_{\A\F},f(U)_{\F}\ketbra{\psi}{\psi}_{\A\F}f(U)_{\F}^{\dagger}). 
\end{align} 
In the context of our theorem, we have $\ketbra{\phi_{U}}{\phi_{U}}_{\A\F} = S_{\I\O\F\A}*\dketbra{U}{U}_{\I\O}^{\otimes k}$ and hence,
\begin{equation}
\begin{aligned}
    \expval{ F^{(\psi)}} &= \int F( S_{\I\O\F\A}*\dketbra{U}{U}_{\I\O}^{\otimes k}, f(U)_{\F} \ketbra{\psi}{\psi }_{\A\F} f(U)^{\dagger}_{\F}) dU\\
    &= \int S_{\I\O\F\A}*\dketbra{U}{U}_{\I\O}^{\otimes k}* f(\overline{U})_{\F} \ketbra{\psi}{\A\F} f(U)^{T}_{\F} \mathrm{d}U\\
    &= \int \Tr( S_{\I\O\F\A}(\dketbra{\overline{U}}{\overline{U}}_{\I\O}^{\otimes k} \otimes f(U)_{\F} \ketbra{\psi}{\psi }_{\A\F} f(U)^{\dagger}_{\F})) \mathrm{d}U
\end{aligned}
\end{equation}
Form here we see that 
\begin{equation}
\begin{aligned}
    \Omega^{\psi}_{A\I\O\F} &= \int\dketbra{\overline{U}}{\overline{U}}_{\I\O}^{\otimes k} \otimes f(U)_{\F} \ketbra{\psi}{\psi }_{\F\A} f(U)^{\dagger}_{\F}) \mathrm{d}U \\
    &= \int\dketbra{\overline{U}}{\overline{U}}_{\I\O}^{\otimes k} \otimes \dketbra{f(U)}{f(U)}_{\F\F^{\prime}} * \ketbra{\psi}{\psi }_{\F^{\prime}\A}) \mathrm{d}U \\
    &= \ketbra{\psi}{\psi }_{\F^{\prime}\A} * \Omega_{\I\O\F\F^{\prime}} = \ketbra{\psi}{\psi }_{\P\A} * \Omega_{\P\I\O\F}
\end{aligned}   
\end{equation} 

\end{proof}

\IndependenceThmDet*

\begin{proof}
For the homomorphism, given $ S^{\ket{0}}_{\I\O\F}$ we can construct $S^{\ket{\psi}}_{\I\O\F} = S^{\ket{0}}_{\I^{\prime}\O\F} * \dketbra{V}{V}_{\I^{\prime}\I}$, where $V \in \mathrm{SU}(d)$ is s.t. satisfies $\ket{\psi} = f(V)\ket{0}$. Then it follows
\begin{equation}
\begin{aligned}
\expval{F}^{\ket{\psi}} &= \int \Tr((S^{\ket{\psi}}_{\I\O\F}*\ChoiU_{\I\O})f(U)\ketbra{\psi}{\psi}_{\F} f(U)^{\dagger}) \mathrm{d}U\\
 &= \int \Tr((S^{\ket{0}}_{\I^{\prime}\O\F} * \dketbra{V}{V}_{\I^{\prime}\I}*\ChoiU_{\I\O})f(U)f(V)\ketbra{0}_{\F} f(V)^{\dagger}f(U)^{\dagger}) dU\\
 &= \int \Tr((S^{\ket{0}}_{\I^{\prime}\O\F} *\dketbra{VU}{VU}_{\I^{\prime}\O})f(UV) \ketbra{0}_{\F} f(UV)^{\dagger}) \mathrm{d}U\\
 &= \int \Tr((S^{\ket{0}}_{\I^{\prime}\O\F} * \ChoiU_{\I\O})f(U)\ketbra{0}_{\F} f(U)^{\dagger}) \mathrm{d}U =\expval{F}^{\ket{0}}.
\end{aligned}
\end{equation}

For the antihomomorphism, given $ S^{\ket{0}}_{\I^{\prime}\O\F}$ we can construct $S^{\ket{\psi}}_{\I\O\F} = S^{\ket{0}}_{\I\O^{\prime}\F} * \dketbra{V}{V}_{\O^{\prime}\O}$, which satisfies
\begin{equation}
\begin{aligned}
\expval{F}^{\ket{\psi}} &= \int \Tr((S^{\ket{\psi}}_{\I\O\F}*\ChoiU_{\I\O})f(U)\ket{\psi}  \bra{\psi}_{\F} f(U)^{\dagger}) \mathrm{d}U\\
 &= \int \Tr((S^{\ket{0}}_{\I\O^{\prime}\F} * \dketbra{V}{V}_{\O^{\prime}\O}*\ChoiU_{\I\O})f(U)f(V)\ketbra{0}_{\F} f(V)^{\dagger}f(U)^{\dagger}) \mathrm{d}U\\
 &= \int \Tr((S^{\ket{0}}_{\I\O^{\prime}\F} * \dketbra{UV}{UV}_{\I\O})f(VU)\ketbra{0}_{\F} f(VU)^{\dagger}) \mathrm{d}U\\
 &= \int \Tr((S^{\ket{0}}_{\I^{\prime}\O\F} * \ChoiU_{\I\O})f(U)\ketbra{0}_{\F} f(U)^{\dagger}) dU =\expval{F}^{\ket{0}}.
\end{aligned}
\end{equation}

\end{proof}

\Omegacovariancef*

\begin{proof}   
The symmetries of a superchannel corresponding to the deterministic task are tightly related to those of the performance operator, as we will show below. We therefore begin by establishing the covariance properties of the performance operator, from which the covariance properties of the superchannel immediately follow.

Let us start with the performance operator for unitary transposition is given by
\begin{align}
    \Omega^{\psi}_{\I\O\F}
    &= \int\!\dketbra{\overline{U}}{\overline{U}}_{\I\O}^{\otimes k}
      \otimes \big(U^{T}_{\F}\,\ketbra{\psi}{\psi}\,\overline{U}_{\F}\big)\,\mathrm{d}U.
\end{align}

For transposition we prove $[\Omega^\psi, A_f] = 0$ for all $A\in \mathrm{SU}(d)$ with
\begin{align}
    A_f &= A_{\I}^{\otimes k}\otimes \mathbb{1}_{\O}^{\otimes k}\otimes \overline{A}_{\F}.
\end{align}
Then we have 
\begin{align}
A_f\,\Omega^{\psi}\,A_f^{\dagger}
   &= \int \dketbra{\overline{U}A^{T}}{\overline{U}A^{T}}^{\otimes k}
      \otimes \Big[\overline{A}\,U^{T}\ketbra{\psi}{\psi}\overline{U}\,A^{T}\Big]\,\mathrm{d}U.
\end{align}
By changing the variable $U$ to $W\coloneqq UA^\dagger$, we have
\begin{align}
A_f\,\Omega^{\psi}\,A_f^{\dagger}
    &= \int \dketbra{\overline{W}}{\overline{W}}^{\otimes k}
      \otimes \Big[W^{T}\ketbra{\psi}{\psi}\overline{W}\Big]\,\mathrm{d}W
    = \Omega^{\psi},
\end{align}
since Haar measure is invariant under right multiplication. This concludes the first part of the proof.

Now let us show the other part of the theorem, namely that $[\Omega^{\psi},V_f]=0$ for,
\begin{align}
V_f &= \mathbb{1}_{\I}^{\otimes k}\otimes \overline{V}_{\O}^{\otimes k}\otimes \mathbb{1}_{\F},
\quad \forall V \in \mathcal{S} \subset \mathrm{SU}(d),\\
\textup{where} \quad \mathcal{S} &\coloneqq
    \big\{
        V \in \mathrm{SU}(d)
        \;\big|\;
        V^{ T}\ket{\psi}=\ket{\psi}
    \big\}.
\end{align}
The action of $V_f$ on the performance operator is given by 
\begin{align}
V_f\,\Omega^{\psi}\,V_f^{\dagger}&= (\mathbb{1}_{\I}^{\otimes k}\otimes \overline{V}_{\O}^{\otimes k}\otimes \mathbb{1}_{\F})\left(\int \dketbra{\overline{U}}{\overline{U}}_{\I\O}^{\otimes k}
   \otimes \big(U^{T}\ketbra{\psi}{\psi}\overline{U}\big)\,\mathrm{d}U\right) (\mathbb{1}_{\I}^{\otimes k}\otimes \overline{V}_{\O}^{\otimes k}\otimes \mathbb{1}_{\F})^{\dagger}\\ 
&= \int \dketbra{\overline{V}\,\overline{U}}{\overline{V}\,\overline{U}}_{\I\O}^{\otimes k}
   \otimes \big(U^{T}V^T\ketbra{\psi}{\psi}\overline{V}\overline{U}\big)\,\mathrm{d}U.
\end{align}
By changing the variable $U$ to $W\coloneqq VU$ and invoking the invariance of Haar measure, we have
\begin{align}
V_f\,\Omega^{\psi}\,V_f^{\dagger}
&= \int \dketbra{\overline{W}}{\overline{W}}_{\I\O}^{\otimes k}
   \otimes \big(W^{T}\ketbra{\psi}{\psi}\overline{W}\big)\,\mathrm{d}W
= \Omega^{\psi},
\end{align}
which shows the theorem for the unitary transposition. The argument extends analogously to other homomorphic and antihomomorphic functions \(f\), and hence follows the claim.

Since $\Omega^{\psi}$ is invariant under these symmetries, we can define the group-averaging twirling map $\map{\tau}(\cdot)\coloneqq\int_U \int_V V_fU_f (\cdot) U_f^\dagger V_f^\dagger \mathrm{d}U\mathrm{d}V$, 
for which $\map{\tau}(\Omega^{\psi})=\Omega^{\psi}$. As proven in~\cite{Ebler_2023}, 
if $S$ is a valid superchannel, then $\map{\tau}(S)$ is also a valid superchannel. 
It then follows that for any given superchannel $S$, the covariant superchannel $\map{\tau}(S)$ attains the same performance. We can see this from the following chain of equalities 
$\Tr(\Omega^{\psi} S) = \Tr(\map{\tau}(\Omega^{\psi}) S) = \Tr(\Omega^{\psi}\map{\tau}(S))$. 
First equality holds because $\Omega^{\psi}$ is invariant under $\map{\tau}$, 
and the second because $\map{\tau}$ is self-adjoint with respect to the 
Hilbert–Schmidt inner product. 

\end{proof}

\subsection{Equivalence of the figures of merit} \label{sec:proofs_fig_mer}

We will prove that, under the assumptions of homomorphic and antihomomorphic functions considered here, the worst-case and average-case fidelities are equivalent. 
\begin{theorem} \label{thm:symm_det}
     Let $S \in \mathcal{L}\left(\mathcal{H}_I \otimes \mathcal{H}_O \otimes \mathcal{H}_F\right)$ be a parallel/sequential/general superchannel that transforms $k$ uses of a unitary operator $U \in \mathcal{S U}(d)$ into $f(U) \in \mathcal{S U}\left(d\right)$ with average fidelity $\expval{F}$. If $[S, U_f]=0$ holds where $U_f = U_{\I} \otimes \overline{U}_{\F}$ for transposition, $U_f = U_{\O} \otimes U_{\F} $ for conjugation and $U_f = U_{\I} \otimes U_{\F}$ for inversion, then the worst-case fidelity is given by $F_{\mathrm{wc}}=\expval{F}$, where $F_{\mathrm{wc}} = \underset{U \in \mathrm{SU}(d)}{\mathrm{inf}}F(S^{\psi}_{\I\O\F}*\ChoiU^{\otimes k}_{\I\O}, f(U)\ketbra{ \psi}{ \psi}_{\F}f(U)^{\dagger}_{\F} )$,as given by Eq.~\eqref{eq:def_F_wc}.
\end{theorem}

\begin{proof}
 Let us start with transposition case where corresponding supermap satisfies symmetry $[S, U_{\I} \otimes \overline{U}_{\F} ]=0$. Then we have
\begin{align}
    S^{\prime}_{\I\O\F}*\ChoiU_{\I\O} &= \Tr_{\I\O}((U_{\I} \otimes \overline{U}_{\F})S_{\I\O\F}(U_{\I} \otimes \overline{U}_{\F})^{\dagger}(\ChoiU^{T}_{\I\O} \otimes \mathbb{1}_{\F}))\\
    &= \overline{U}_{\F} \Tr_{\I\O}(U_{\I}S_{\I\O\F}U^{\dagger}_{\I} ((\mathbb{1}_{\I} \otimes \overline{U}_{\O})\dketbra{\mathbb{1}}{\mathbb{1}}_{\I\O} (\mathbb{1}_{\I} \otimes U_{\O}^{T})))\overline{U}_{\F}^{\dagger}\\
    &= \overline{U}_{\F} \Tr_{\I\O}((\mathbb{1}_{\I} \otimes U_{\O}^{T})^{\dagger})U_{\I}S_{\I\O\F}U^{\dagger}_{\I} ((U_{\I} \otimes \mathbb{1}_{\O})\dketbra{\mathbb{1}}{\mathbb{1}}_{\I\O})\overline{U}_{\F}^{\dagger}\\
    &= \overline{U}_{\F} \Tr_{\I\O}(S_{\I\O\F} \dketbra{\mathbb{1}}{\mathbb{1}}_{\I\O})\overline{U}_{\F}^{\dagger}.
\end{align}

We can now calculate state fidelity between implemented state $S^{\prime}_{\I\O\F}*\ChoiU_{\I\O}$ and the perfect realisation $f(U)\ket{\psi}$, where $f(U) = U^T$:
\begin{align}
    F( S^{\prime}_{\I\O\F}*\ChoiU_{\I\O} ) &= \bra{\psi} \overline{U}_{\F} S^{\prime}_{\I\O\F}*\ChoiU_{\I\O} U_{\F}^{T}\ket{\psi}\\
    &= \bra{\psi}S_{\I\O\F} * \dketbra{\mathbb{1}}{\mathbb{1}}_{\I\O} \ket{\psi}.
\end{align}

Similar hold for other functions and their respective symmetries.  

Since fidelity for the case of particular $U$ is independent of $U$, then we have that $\expval{F}= F = F_{\mathrm{wc}}$. 

\end{proof}

\section{A useful lemma}
\begin{lemma}\label{lemma:identities}
Let $\rho$ be an operator $\rho_{12} \in\mathcal{L}(\mathcal{H}_1\otimes\mathcal{H}_2)$ and 
$P^{S}_{23}$, $P^{A}_{23}$ projectors onto the symmetric and antisymmetric subspaces of 
$\mathcal{L}(\mathcal{H}_2\otimes\mathcal{H}_3)$, given by
\begin{align}
P^{S}_{12} &= \frac{\mathbb{1}_{12} + \mathrm{SWAP}_{12}}{2}, \\[1mm]
P^{A}_{12} &= \frac{\mathbb{1}_{12} - \mathrm{SWAP}_{12}}{2},
\end{align}
where $\mathrm{SWAP}_{12}= \sum_{i,j=1}^{d} \ketbra{i}{j}_1 \otimes \ketbra{j}{i}_2 = \dketbra{\mathbb{1}}{\mathbb{1}}^{T_2}$.

Then, the following identities hold:
\begin{align}
\rho_{12} * P^{S}_{23} &=
\frac{\Tr_{2}(\rho_{12})\otimes \mathbb{1}_3 + \rho_{13}^{T_{3}}}{2}, \label{eq:Ps}\\[1mm]
\rho_{12} * P^{A}_{23} &= 
\frac{\Tr_{2}(\rho_{12})\otimes \mathbb{1}_3 - \rho_{13}^{T_{3}}}{2}. \label{eq:Pa}
\end{align}

\end{lemma}

\begin{proof}
    Let us start from the identity involving the symmetric and antisymmetric projector,
   \begin{align}\label{eq:proof_lem}
       \rho_{12} * P^{S,A}_{23} &= \rho_{12} * \left( \frac{\mathbb{1}_{23} \pm \mathrm{SWAP}_{23}}{2} \right)
   \end{align}
   We now analyse the two terms separately. For the first term, we have
   \begin{align}
       \rho_{12}* \mathbb{1}_{23}= \Tr_2((\rho_{12} \otimes \mathbb{1}_{3})\otimes \mathbb{1}_{123})=\Tr_{2}(\rho_{12})\otimes \mathbb{1}_{3}.
   \end{align}
For the second term, we have
\begin{align}
       \rho_{12}* \dketbra{\mathbb{1}}{\mathbb{1}}^{T_3}_{23} &= 
       \Tr_{2}((\rho_{12} \otimes \mathbb{1}_3)(\dketbra{\mathbb{1}}{\mathbb{1}}_{12}^{T} \otimes \mathbb{1}_{3}))\\
       &= \sum_{ij}\Tr_{2}((\rho_{12} \otimes \mathbb{1}_3)(\ketbra{i}{j}_{1}\otimes \ketbra{j}{i}_{2}  \otimes \mathbb{1}_{3}))\\
       &= \sum_{ij}\bra{i}\rho_{12}\ket{j}\otimes \ketbra{j}{i}_{2}  \otimes \mathbb{1}_{3}=\rho_{13}^{T_{3}}.
   \end{align}
Substituting these results into Eq.~\eqref{eq:proof_lem} immediately yields the desired relations, completing the proof.
\end{proof}

\section{Proofs for the pure input state (without auxiliary system)}

\subsection{Transposition}\label{subsec:proof_trans}

\ThmProbTransSingle*

\begin{proof}  
To prove the theorem, we determine the supermap that, given a single use of the unitary operation $U \in \mathrm{SU}(d)$ and a known input state $\ketbra{\psi}{\psi}$, implements the transposition $f(U) = U^{T}$ with maximal success probability. Thm.~\ref{thm:twirl_prob_constr} implies that, without loss in
performance, we may restrict our analysis to supermaps whose Choi operators respect the
corresponding commutation relations 
\begin{equation}\label{project_basis_trans}
    [S_{\I\O\F},  U_{\I} \otimes \mathbb{1}_{\O} \otimes \overline{U}_{\F}] = 0,
\end{equation}
which ensure that $S$ can be written as
\begin{align}\label{eq:S_trans_prob}
    S_{\I\O\F} = A_{\O}  \otimes \ketbra{\phi^{+}}{\phi^{+}}_{\I\F} + B_{\O} \otimes (\mathbb{1}_{\I\F} -   \ketbra{\phi^{+}}{\phi^{+}}_{\I\F}).
\end{align}
This leads to the following SDP:
\begin{align}
    &\textup{max} \; \; p \\
    &\textup{s.t.} \; \; (A_{\O}  \otimes \ketbra{\phi^{+}}{\phi^{+}}_{\I\F} + B_{\O} \otimes (\mathbb{1}_{\I\F} -   \ketbra{\phi^{+}}{\phi^{+}}_{\I\F}))*\dketbra{\mathbb{1}}{\mathbb{1}}_{\I\O} = p \ketbra{0}{0 }_{\F}, \label{eq:proof_trans_SDP1}\\ 
    &A_{\O}, B_{\O} \geq 0, \\
    &(A_{\O}  \otimes \ketbra{\phi^{+}}{\phi^{+}}_{\I\F} + B_{\O} \otimes (\mathbb{1}_{\I\F} -   \ketbra{\phi^{+}}{\phi^{+}}_{\I\F})) \leq S_{\textup{ch}}, \label{eq:proof_trans_SDP2}\\
    &\Tr_{\F}(S_{\textup{ch}}) = \frac{\mathbb{1}_{\I}\otimes \mathbb{1}_\O}{d}.\label{eq:proof_trans_SDP3}
\end{align}
From Eq.~\eqref{eq:proof_trans_SDP1} we obtain
\begin{align}
     \frac{1}{d} A + \left(\Tr(B) \mathbb{1} - \frac{1}{d} B\right) = p\ketbra{0}{0}.
\end{align}
We can get two constraints by projecting on $\ket{0}$
\begin{align}
  \bra{0} A \ket{0} + d\Tr(B)\bra{0} \mathbb{1}\ket{0}  - \bra{0} B\ket{0} = dp,
\end{align}
and on $\ket{i} \perp \ket{0}$ in the orthogonal complement,
\begin{align}
   \bra{i} A \ket{i} = \bra{i} B\ket{i} - d\Tr(B)= b_{ii} - d\sum_j b_{jj}, 
\end{align}
where we used $\bra{i} B\ket{i} \coloneqq b_{ii}$ notation. Since $A \geq 0$, we know that $\bra{i} A \ket{i}\geq 0$. However, since $B \geq 0$ we have that $b_{ii} - d\sum_j b_{jj} \leq 0 $. Hence,
\begin{align}\label{eq:eqv_A}
   B = 0, \quad \textup{and} \quad A = pd\ketbra{0}{0}.
\end{align}
Using Eqs.~\eqref{eq:proof_trans_SDP2} and \eqref{eq:proof_trans_SDP3} and setting $B=0$, we obtain:
\begin{align}
    A \leq \mathbb{1}.
\end{align}
Plugging now Eq.~\eqref{eq:eqv_A}
\begin{equation}
    p \leq \frac{1}{d} 
\end{equation}
which further implies that the maximal possible value is,
\begin{equation}
    p_{\textup{max}} = \frac{1}{d}
\end{equation}

In conclusion, the optimal supermap for probabilistic, single-call unitary transposition is given by
\begin{equation}\label{eq:S_trans_pure_opt}
    S^{\textup{max}}_{\I\O\F} = \ketbra{0}{0 }_{\O} \otimes \ketbra{\phi^{+}}{\phi^{+}}_{\I\F}.
\end{equation}
\end{proof}
In order to have a better understanding of the supermap in terms of quantum circuits, we solve for the form of the encoder and decoder. We can write the supermap as  
\begin{align}
S_{\I\O\F} = E_{\I\M} * D_{\M\O/\F}
\end{align}
From here, we have
\begin{align}
\Tr_{\O\F}(S_{\I\O\F}) 
&=\Tr_\M(E_{\I\M}). 
\end{align}
With this and Eq.~\eqref{eq:S_trans_pure_opt}, we can find that, 
\begin{align}
    \Tr_{\O\F}(S_{\I\O\F}) = \frac{\mathbb{1}_{\I}}{d} = \Tr_{\M}(E_{\I\M}).
\end{align}
Hence, we may set the encoder state as a purification of the maximally mixed state, that is,
\begin{align}
 E_{\I\M}= \ketbra{\phi^{+}}{\phi^{+}}_{\I\M}.
\end{align}
Now for the decoder we have,
\begin{align}
    \Tr_{\I}(S_{\I\O\F}) &= \Tr_{\I}\left(\ketbra{\phi^{+}}{\phi^{+}}_{\I\M}*D_{\M\O/\F}\right)\\
    &= \frac{\Tr_{\M}(D_{\M\O/\F})}{d}.
\end{align}
Using this and Eq.~\eqref{eq:S_trans_pure_opt}, we get
\begin{align}
    \Tr_{\I}(S_{\I\O\F}) &= \ketbra{0}{0 }_{\O} \otimes \frac{\mathbb{1}}{d}\\
&=\frac{\Tr_{\I}(D_{\M\O/\F})}{d}
\end{align}
from which we have,
\begin{align}
 D_{\M\O/\F}&= \ketbra{0}{0 }_{\O} \otimes\dketbra{\mathbb{1}}{\mathbb{1}}_{\M\F}.
\end{align}
Thus, we conclude that the supermap has the following circuit form
\begin{align}
    S_{\I\O\F} = \ketbra{\phi^{+}}{\phi^{+}}_{\I\A} * \left(\ketbra{0}{0 }_{\O} \otimes \dketbra{\mathbb{1}}{\mathbb{1}}_{\M\F}\right)
\end{align}
which is illustrated in Figure~\ref{fig:transpositionKNpure}, or equivalently,
\begin{align}
    S^{\psi}_{\I\O\F} = \ketbra{\phi^{+}}{\phi^{+}}_{\I\A} * \left(\ketbra{\psi}{\psi }_{\O} \otimes \dketbra{\mathbb{1}}{\mathbb{1}}_{\M\F}\right)
\end{align}
due to the symmetry of the supermap.

\FidTransPure*
\begin{proof}\label{proof:trans_sing_det}
To prove the theorem, we construct the supermap \( S \) that maximizes the average fidelity 
\(\expval{F^{\psi}} = \mathrm{Tr}(S\,\Omega^{\psi})\), and show that it attains the stated optimal value.
For the unknown input state case, Appendix E.4 of Ref.~\cite{Quintino_2022det} shows that the performance operator for unitary transposition is 
\begin{equation}
    \Omega = \ketbra{\phi^{+}}{\phi^{+}}_{\P\O} \otimes \ketbra{\phi^{+}}{\phi^{+}}_{\I\F}+ \frac{\big(\mathbb{1}_{\P\O} - \ketbra{\phi^{+}}{\phi^{+}}_{\P\O}\big) \otimes \big(\mathbb{1}_{\I\F} - \ketbra{\phi^{+}}{\phi^{+}}_{\I\F}\big)}{d^2 - 1}.
\end{equation}
When the input state is known, without loss of generality \(\ketbra{\psi} = \ketbra{0}\), we may apply Thm.~\ref{thm:performance} to obtain the corresponding performance operator
\begin{align}
    \Omega^{\psi}_{\I\O\F} &= \ketbra{0}{0 }_{\P}* \left( \ketbra{\phi^{+}}{\phi^{+}}_{\P\O} \otimes \ketbra{\phi^{+}}{\phi^{+}}_{\I\F}+ \frac{\big(\mathbb{1}_{\P\O} - \ketbra{\phi^{+}}{\phi^{+}}_{\P\O}\big) \otimes \big(\mathbb{1}_{\I\F} - \ketbra{\phi^{+}}{\phi^{+}}_{\I\F}\big)}{d^2 - 1} \right)\\
    &= \frac{\ketbra{0}{0 }_{\O}\otimes  \ketbra{\phi^{+}}{\phi^{+}}_{\I\F}}{d} + \frac{\left(\mathbb{1}_{\O} - \tfrac{\ketbra{0}{0}_{\O}}{d}\right)\otimes \big(\mathbb{1}_{\I\F} -  \ketbra{\phi^{+}}{\phi^{+}}_{\I\F}\big)}{d^2 - 1}.\label{eq:perform_trans_proof}
\end{align}
Due to the symmetries of the performance operator \(\Omega^{\psi}\), Thm.~\ref{thm:S_det_covariance_f} implies that, without loss in performance, we may restrict our analysis to superchannels whose Choi operators respect the corresponding commutation relations
\begin{equation}\label{symm_S_trans}
    [ S_{\I\O\F}, U_{\I} \otimes \mathbb{1}_{\O} \otimes \overline{U}_{\F}] = 0,
\end{equation}
which ensure that $S$ can be written as
\begin{equation}\label{general_S_trans}
    S = A_{\O}  \otimes \ketbra{\phi^{+}}{\phi^{+}}_{\I\F} + B_{\O} \otimes \big(\mathbb{1}_{\I\F} - \ketbra{\phi^{+}}{\phi^{+}}_{\I\F}\big).
\end{equation}
for some linear operators $A,B\in\mathcal{L}(\mathcal{H}_\O$). The optimization problem then becomes
\begin{align}\label{SDP_trans_k1}
    \textup{given}\quad &\Omega^{\psi}\\
    \textup{max}\quad &\expval{F^{\psi}}=\Tr(S\Omega^{\psi})\label{eq:trans_det_SDP_proof0}\\
    \textup{subject to}\quad &\Tr_{\F}(S) = \frac{\mathbb{1}_{\I} \otimes \mathbb{1}_{\O}}{d},\label{eq:trans_det_SDP_proof1}\\
    & S \geq 0,\label{eq:trans_det_SDP_proof2}\\
    \textup{where}\quad S &= A_{\O}  \otimes \ketbra{\phi^{+}}{\phi^{+}}_{\I\F} + B_{\O} \otimes \big(\mathbb{1}_{\I\F} - \ketbra{\phi^{+}}{\phi^{+}}_{\I\F}\big).\label{eq:trans_det_SDP_proof3}
\end{align}
Now, from Eq.~\eqref{eq:trans_det_SDP_proof0}, Eq.~\eqref{eq:perform_trans_proof} and Eq.~\eqref{eq:trans_det_SDP_proof3} we have
\begin{align}\label{eq:Tr_perform_trans}
\Tr(S\Omega^{\psi})&=\frac{\bra{0}A\ket{0}}{d} + \Tr(B) - \frac{\bra{0}B\ket{0}}{d},
\end{align}
and from Eq.~\eqref{eq:trans_det_SDP_proof1} and Eq.~\eqref{eq:trans_det_SDP_proof3} we obtain
\begin{align}\label{eq:B}
     B=\frac{\mathbb{1} - A}{d^2- 1} \geq 0.
\end{align}
Plugging Eq.~\eqref{eq:B} into Eq.~\eqref{eq:Tr_perform_trans},
\begin{align}\label{eq:tr_proofstep_trans}
\Tr(S\Omega^{\psi})&=\frac{\bra{0}A\ket{0}}{d} + \frac{d - \bra{0}A\ket{0} - \sum_{i\neq 0}\bra{i}A\ket{i}}{d^2- 1} - \frac{1 - \bra{0}A\ket{0}}{d(d^2- 1)}
\end{align}
Additionally, we may now invoke the known input state covariance condition from Thm.~\ref{thm:S_det_covariance_f}: $[S_{\I\O\F}, \overline{V}_{\O}] = 0$ such that $V^{T}\ketbra{0}{0} = \ketbra{0}{0}$ for all 
$V \in \mathrm{SU}(d)$, which is equivalently expressed as 
$[S_{\I\O\F}, V_{\O}] = 0$ with $V\ketbra{0}{0} = \ketbra{0}{0}$ for all 
$V \in \mathrm{SU}(d)$. From this symmetry, we can infer additional structure on the coefficients. In particular, we obtain
\begin{align}\label{eq:symm_coeff}
    A= \alpha \ketbra{0}{0}+ \beta(\mathbb{1} -\ketbra{0}{0}).
\end{align}
Let us notice that~\eqref{eq:B} implies $A \leq \mathbb{1}$. Together with the positivity condition~\eqref{eq:trans_det_SDP_proof2}, we therefore conclude that $\alpha, \beta \in [0,1]$. Now, using~\eqref{eq:symm_coeff} and~\eqref{eq:tr_proofstep_trans},
\begin{align}
    \textup{max}\quad \expval{F^\psi}&=\frac{\alpha}{d} + \frac{d - \alpha}{d^2- 1}-\frac{\beta(d-1)}{d^2- 1} - \frac{1 -\alpha}{d(d^2- 1)} \\
    &=  \frac{1}{d}  + \frac{\alpha}{d + 1} - \beta\frac{d-1}{d^2-1}. 
\end{align}
From this expression, we observe that the dependence on $\beta$ enters with a negative coefficient; hence, we set $\beta = 0$. Conversely, since $\alpha$ appears with a positive coefficient, we maximize it by setting $\alpha = 1$, and obtain
\begin{align}
   \expval{F^\psi}^{\textup{max}}&= \frac{1}{d + 1}+\frac{1}{d}\\
   &=\frac{2d+1}{d(d+ 1)}.
\end{align}
Going back to Eq.~\eqref{eq:symm_coeff} and Eq.~\eqref{eq:symm_coeff}, this means that coefficients of the supermap are,
\begin{align}
    A = \ketbra{0}{0}, \quad B = \frac{\mathbb{\mathbb{1}}- \ketbra{0}{0}}{d^{2} - 1}
\end{align}
Plugging this back into Eq.~\eqref{eq:trans_det_SDP_proof3} we obtain the optimal solution,
\begin{equation}
    S^{\mathrm{(opt)}}_{\mathrm{trans}} = \ketbra{0}{0 }_{\O} \otimes \ketbra{\phi^{+}}{\phi^{+}}_{\I\F} + \frac{\big(\mathbb{1}_{\O} - \ketbra{0}{0 }_{\O}\big)\otimes \big(\mathbb{1}_{\I\F} - \ketbra{\phi^{+}}{\phi^{+}}_{\I\F}\big)}{d^2 - 1}.
\end{equation}

\end{proof}

In Sec.~\ref{sec:transp_det_single}, we analyse the optimal superchannel in terms of its construction with explicit encoders and decoders, and show that the optimal supermap can be implemented via a probabilistic strategy with the  correction step on the failure branch.

\subsection{Conjugation}\label{subsec:proof_conjg}

\ProbConjgSingl*

\begin{proof}
To prove the theorem, we determine the supermap that, given a single use of the unitary operation $U \in \mathrm{SU}(d)$ and a known input state $\ketbra{\psi}{\psi}$, implements the transposition $f(U) = \overline{U}$ with maximal success probability. Thm.~\ref{thm:twirl_prob_constr} implies that, without loss in
performance, we may restrict our analysis to supermaps whose Choi operators respect the
corresponding commutation relations 
\begin{equation}\label{eqv:proof_conjProb1}
    [S_{\I\O\F}, \mathbb{1}_{\I} \otimes U_{\O} \otimes U_{\F}] = 0 
\end{equation}
Any $S$ satisfying~\eqref{symm_S_conjg} can be written in the projector basis $\{P^{S}_{\O\F},  P^{A}_{\O\F} \}$ with operator coefficients, where $P^{S}$ and $P^{A}$ are projectors on symmetric and anti-symmetric space, as defined in Lemma~\ref{lemma:identities}. That is,
\begin{equation}\label{general_S_conjg}
    S_{\I\O\F} = A_{\I}  \otimes P^{S}_{\O\F} + B_{\I} \otimes P^{A}_{\O\F}.
\end{equation}
The optimization problem then becomes
\begin{align}
    &\textup{max} \; \; p\\
    \textup{s.t.} \; \; &(A_{\I}  \otimes P^{S}_{\O\F} + B_{\I} \otimes P^{A}_{\O\F})*\dketbra{\mathbb{1}}{\mathbb{1}}_{\I\O} = p \ketbra{0}{0 }_{\F} \label{eq:poof_c_c1}\\
    &A_{\I}, B_{\I} \geq 0 \\
    &(A_{\I}  \otimes P^{S}_{\O\F} + B_{\I} \otimes P^{A}_{\O\F}) \leq S_{\textup{ch}}\label{eq:poof_c_c2}\\
    &\Tr_{\F}(S_{\textup{ch}}) = \sigma_{\I} \otimes \mathbb{1}_{\O} \; \;\Tr(\sigma_{\I}) = 1.\label{eq:poof_c_c3}
\end{align}
Using the properties of the link product, from Eq.~\eqref{eq:poof_c_c1} we have,
\begin{align}\label{eq:link_conjug_midstep}
    A_{\I} * P^{S}_{\I\F} + B_{\I} * P^{A}_{\I\F} = p \ketbra{0}{0 }_{\F}.
\end{align}
We can now plug the expressions for symmetric and antisymmetric projectors, using 
Lemma~\ref{lemma:identities} and obtain,
\begin{align}\label{eq:conjg_link_expanded}
    \frac{\Tr(A)\,\mathbb{1}_{\F} + A_{\F}^{T}}{2} + \frac{\Tr(B)\,\mathbb{1}_{\F} - B_{\F}^{T}}{2}&= p\ketbra{0}{0}.
\end{align}
Now, on the LHS we have the sum of two terms while on the RHS we have rank one positive semidefinite operator. This impleas that each term in the sum on LHS has to also be rank one positive semidefinite operator. For the first term $\frac{\operatorname{Tr}(A)\,\mathbb{1}_{\F} + A_{\F}^{T}}{2}$, since $A\geq 0$ and $\mathbb{1}$ is full rank, this cannot be rank one: hence $A = 0$. This leaves us with,
\begin{align}
    \frac{\Tr(B)\,\mathbb{1}_{\F} - B_{\F}^{T}}{2}&= p\ketbra{0}{0}.
\end{align}
When $d\geq 2$, $ \mathrm{rank}\left(\frac{\Tr(B)\,\mathbb{1}_{\F} - B_{\F}^{T}}{2} \right )>1$, hence $B=0$ and $p=0$.

When $d=2$, we have,
\begin{align}\label{eq:Bstep}
\frac{(b_{00} + b_{11}) - \bra{0}B_{\F}^{T}\ket{0}}{2}&= p\ketbra{0}{0}
\end{align}
where we used $b_{ij}\coloneqq \bra{i} B\ket{j} $ notation. To maximise the probability then we set $\bra{0}B\ket{0} = 0$ and hence
\begin{equation}
    B = 2p\ketbra{1}{1}.
\end{equation}
Let us recall that, for the case of an unknown input, we have $p = 1$ when $d = 2$. Since the unknown input scenario is less restrictive than the known input state case, any value attainable in the former must also be attainable in the latter. We thus conclude that $p = 1$. The optimal supermap for probabilistic single-call unitary conjugation therefore exists only for $d=2$ and has the form 
\begin{equation}
S^{\textup{max}}_{\I\O\F} = \ketbra{1}{1}_{\I} \otimes P^{A}_{\O\F}.
\end{equation}

\end{proof}

Notice that this is already in the form of the encoder $E_{\I} = \ketbra{1}{1}_{\I}$ and the decoder $D_{\O\F} = P^{A}_{\O\F}$ with no ancillary space, as shown in Fig.~\ref{fig:conjgKNpure}.

\DetConjgSingl*

\begin{proof}
Let us first note that from the previous proof we immediately know that for $d=2$ we have that $\expval{F^{\psi}} = 1$ since $p=1$. Let us then argue $d>2$ case and find the supermap corresponding to the optimal value.

As in the proof for deterministic approximate realisation of transposition ~\ref{proof:trans_sing_det}, we first refer to performance operator for the unknown input state case, which can be also found in Appendix E.3 of Ref.~\cite{Quintino_2022det}
\begin{equation}\label{perform}
    \Omega = \frac{P^{S}_{\P\I} \otimes P^{S}_{\O\F}}{d_{S}} +\frac{P^{A}_{\P\I} \otimes P^{A}_{\O\F}}{d_{A}}.
\end{equation}
where $P^{S}$ and $P^{A}$ are projectors on symmetric and anti-symmetric space, as defined in Lemma~\ref{lemma:identities} and $d_S=\tr(P^S) = \frac{d(d+1)}{2}$, $d_A=\tr(P^A) = \frac{d(d-1)}{2}$ are the dimensions of the symmetric and antisymmetric subspaces. When the input state is known, without loss of generality \(\ketbra{\psi} = \ketbra{0}\), we may apply Thm.~\ref{thm:performance} to obtain the corresponding performance operator 
\begin{align}
\Omega^{\psi}_{\I\O\F} &= \ketbra{0}{0 }_{\P}* \left(\frac{P^{S}_{\P\I} \otimes P^{S}_{\O\F}}{d_{S}} +\frac{P^{A}_{\P\I} \otimes P^{A}_{\O\F}}{d_{A}} \right)\\
    &= \frac{(\mathbb{1}_{\I} + \ketbra{0}{0 }_{\I}) \otimes P^{S}_{\O\F}}{2d_S} + \frac{(\mathbb{1}_{\I} - \ketbra{0}{0 }_{\I}) \otimes P^{A}_{\O\F}}{2d_A}.\label{eq:perform_conjg_proof}  
\end{align}
Due to the symmetries of the performance operator $\Omega^\psi$, it follows from Thm.~\ref{thm:S_det_covariance_f} that, without loss in performance, we may restrict our analysis to superchannels with Choi operator which respect the commutation relations
\begin{equation}\label{symm_S_conjg}
    [ S_{\I\O\F}, \mathbb{1}_{\I} \otimes U_{\O} \otimes U_{\F}] = 0.
\end{equation}
which ensure that $S$ can be written as
\begin{equation}\label{general_S_conj}
    S = A_{\I}  \otimes P^{S}_{\O\F} + B_{\I} \otimes P^{A}_{\O\F}
\end{equation}
for some linear operators $A,B\in\mathcal{L}(\mathcal{H}_\I$). The optimization problem then becomes
\begin{align}\label{SDP_conjg_k1}
    \textup{given}\quad &\Omega^{\psi}\\
    \textup{max}\quad &\expval{F^{\psi}}=\Tr(S\Omega^{\psi})\label{eq:conjg_det_SDP_proof0}\\
    \textup{subject to}\quad &\Tr_{\F}(S) = \sigma_{\I} \otimes \mathbb{1}_{\O}, \label{eq:proof_conjg_1}\\
    &\Tr(\sigma_{\I}) = 1,\quad \Tr(S) = d,\quad S \geq 0,\label{eq:proof_conjg_2}\\
    \textup{where}\quad S &= A_{\I}  \otimes P^{S}_{\O\F} + B_{\I} \otimes P^{A}_{\O\F} \label{eq:conjg_det_SDP_proof3}.
\end{align}
Now, from Eq.~\eqref{eq:conjg_det_SDP_proof0}, Eq.~\eqref{eq:perform_conjg_proof} and Eq.~\eqref{eq:conjg_det_SDP_proof3}
\begin{align}
\Tr(S\Omega^{\psi})&= \Tr(\frac{(A_{\I} + A_{\I}\ketbra{0}{0 }_{\I}) \otimes P^{S}_{\O\F}}{2d_S} + \frac{(B_{\I} - B_{\I}\ketbra{0}{0 }_{\I}) \otimes P^{A}_{\O\F}}{2d_A})\\
&= \frac{(\Tr(A_{\I}) + \bra{0}A_{\I}\ket{0}_{\I})}{2} + \frac{(\Tr(B_{\I}) - \bra{0}B_{\I}\ket{0}_{\I})}{2}\label{eq:Tr_perform_conjg}
\end{align}
From Eq.~\eqref{eq:proof_conjg_1}, ~\eqref{eq:proof_conjg_2} and ~\eqref{eq:conjg_det_SDP_proof3} we have
\begin{align}\label{eq:coef_conjg_det}
      A_{\I} &= \frac{2\sigma_{\I} - (d-1)B_{\I}}{d+1} \geq 0
\end{align}
Going back into Eq.~\eqref{eq:Tr_perform_conjg} we have,
\begin{align}
\textup{max}\quad \Tr(S\Omega^{\psi})&= \frac{\left(\Tr(\frac{2\sigma_{\I} - (d-1)B_{\I}}{d+1}) + \bra{0}\frac{2\sigma_{\I} - (d-1)B_{\I}}{d+1}\ket{0}_{\I}\right)}{2} + \frac{\Tr(B_{\I}) - \bra{0}B_{\I}\ket{0}_{\I}}{2}\\
&= \frac{1}{d+1}+ \frac{\bra{0}\sigma_{\I}\ket{0}}{d+1} - \frac{(d-1)\Tr(B_{\I})}{d+1} - \frac{(d-1)\bra{0}B_{\I}\ket{0}_{\I}}{d+1} + \frac{\Tr(B_{\I}) - \bra{0}B_{\I}\ket{0}_{\I}}{2}\\
&= \frac{1}{d+1}+ \frac{\bra{0}\sigma_{\I}\ket{0}}{d+1} - \Tr(B_{\I})\frac{d-3 }{2(d+1)} - \bra{0}B_{\I}\ket{0}_{\I}\frac{3d-1}{2(d+1)}.
\end{align}
For $d>2$, we see that all the coefficients of the terms having $B$ dependence are negative. Hence, since $B\geq0$, the trace is maximal when $B=0$. To maximise the trace, we set $\sigma_{\I} = \ketbra{0}{0}$ and obtain,
\begin{align}
\expval{F^{\psi}}^{\textup{max}}= \frac{2}{d+1}.
\end{align}
Going back to Eq.~\eqref{eq:coef_conjg_det} and setting $\sigma_{\I} = \ketbra{0}{0}$ we get,
\begin{align}
    A = \frac{2}{d+1}\ketbra{0}{0}.
\end{align}
Hence, the optimal supermap is,
\begin{equation}
    S^{\textup{opt}}_{\I\O\F} = \frac{2}{d+1}\ketbra{0}{0 }_{\I} \otimes P^{S}_{\O\F}
\end{equation}
This means that supermap is realised with encoder $E_{\I} = \ketbra{0}{0}_{\I}$ and decoder $D_{\O\F} =\frac{2}{d+1} P^{S}_{\O\F} $, without using the memory. 
\end{proof}

\subsection{Inversion}\label{subsec:proof_inv}

\InvProbSingl*
\begin{proof}
To prove the theorem, we determine the supermap that, given a single use of the unitary operation $U \in \mathrm{SU}(d)$ and a known input state $\ketbra{\psi}{\psi}$, implements the transposition $f(U) = U^{-1}$ with maximal success probability. 
Thm.~\ref{thm:twirl_prob_constr} implies that, without loss in
performance, we may restrict our analysis to supermaps whose Choi operators respect the
corresponding commutation relations 
\begin{equation}\label{eqv:proof_InvProb1}
    [S_{\I\O\F}, U_{\I} \otimes \mathbb{1}_{\O} \otimes U_{\F}] = 0 
\end{equation}
which ensure that $S$ can be written as
\begin{align}
    S_{\I\O\F} = A_{\O}  \otimes P^{S}_{\I\F} + B_{\O} \otimes P^{A}_{\I\F}.
\end{align}
Plugging this into the constraints, we obtain the following SDP:
\begin{align}
    &\textup{max} \; \; p\\
    \textup{s.t.} \; \; &(A_{\O}  \otimes P^{S}_{\I\F} + B_{\O} \otimes P^{A}_{\I\F})*\dketbra{\mathbb{1}}{\mathbb{1}}_{\I\O} = p \ketbra{0}{0 }_{\F} \label{eq:poof_i_c1}\\
    &A_{\O}, B_{\O} \geq 0 \\
    &(A_{\O}  \otimes P^{S}_{\I\F} + B_{\O} \otimes P^{A}_{\I\F}) \leq S_{\textup{ch}}\label{eq:poof_i_c2}\\
    &\Tr_{\F}(S_{\textup{ch}}) = \sigma_{\I}\otimes \mathbb{1}_{\O}\label{eq:poof_i_c3}
\end{align}
From Eq.~\eqref{eq:poof_i_c1} we have,
\begin{align}
    A_{\I} * P^{S}_{\I\F} + B_{\I} * P^{A}_{\I\F} = p \ketbra{0}{0 }_{\F}.
\end{align}
Let us notice that due to link product, this gives exactly the same equation as Eq.~\eqref{eq:link_conjug_midstep}. Hence, the proof is identical to the proof presented for conjugation and we obtain for $d=2$ 
\begin{equation}
    S^{\textup{max}}_{\I\O\F} = \ketbra{1}{1}_{\O} \otimes P^{A}_{\I\F},
\end{equation}
while for $d>2$ we have $p=0$.

\end{proof}

\InvDetSing*
\begin{proof}
Let us first note that, when $d = 2$, the supermap implementing inversion map has equivalent performance to the supermap implementing transposition map. This equivalence arises because conjugating the transposed operation by the Pauli $Y$ operator, i.e., applying $Y$ before and after the unitary, implements inversion.

As in the proofs before, we first refer to performance operator for the unknown input state case, which can be also found in Appendix E.5 of Ref.~\cite{Quintino_2022det}
\begin{equation}
    \Omega = \frac{P^{S}_{\P\O} \otimes P^{S}_{\I\F}}{d_{S}} +\frac{P^{A}_{\P\O} \otimes P^{A}_{\I\F}}{d_{A}} 
\end{equation}
When the input state is known, without loss of generality \(\ketbra{\psi} = \ketbra{0}\), we may apply Thm.~\ref{thm:performance} to obtain the corresponding performance operator
\begin{align}
\Omega^{\psi}_{\I\O\F} &= \ketbra{0}{0 }_{\P}* \left(\frac{P^{S}_{\P\O} \otimes P^{S}_{\I\F}}{d_{S}} +\frac{P^{A}_{\P\O} \otimes P^{A}_{\I\F}}{d_{A}} \right)\\
    &= \frac{(\mathbb{1}_{\O} + \ketbra{0}{0 }_{\O}) \otimes P^{S}_{\I\F}}{2d_S} + \frac{(\mathbb{1}_{\O} - \ketbra{0}{0 }_{\O}) \otimes P^{A}_{\I\F}}{2d_A}.\label{eq:perform_inv_proof} 
\end{align}
Due to the symmetries of the performance operator $\Omega^\psi$, it follows from Thm.~\ref{thm:S_det_covariance_f} that, without loss in performance, we may restrict our analysis to superchannels with Choi operator which respect the commutation relations
\begin{equation}\label{symm_S_inv}
   [S_{\I\O\F}, U_{\I} \otimes \mathbb{1}_{\O} \otimes U_{\F}] = 0
\end{equation}
which ensure that $S$ can be written as
\begin{equation}\label{general_S_inv}
    S = A_{\O}  \otimes P^{S}_{\I\F} + B_{\O} \otimes P^{A}_{\I\F}
\end{equation}
for some linear operators $A,B\in\mathcal{L}(\mathcal{H}_\O$).
The optimization problem then becomes
\begin{align}
    \textup{given}\quad &\Omega^{\psi}\\
    \textup{max}\quad &\Tr(S\Omega^{\psi})\label{eq:inv_det_SDP_proof0}\\
    \textup{subject to}\quad &\Tr_{\F}(S) = \frac{\mathbb{1}_{\I} \otimes \mathbb{1}_{\O}}{d}{d},\label{eq:inv_det_SDP_proof1}\\
    & S \geq 0,\label{eq:inv_det_SDP_proof2}\\
    \textup{where}\quad S&= A_{\O}  \otimes P^{S}_{\I\F} + B_{\O} \otimes P^{A}_{\I\F}.\label{eq:inv_det_SDP_proof3}
\end{align}
From Eq.~\eqref{eq:inv_det_SDP_proof1} and~\eqref{eq:inv_det_SDP_proof3} we have,
\begin{align}\label{eq:coef_AB}
    B &= \frac{1}{d-1} \left ( \frac{2 \mathbb{1}}{d} - (d+1)A \right) \geq 0
\end{align}
from which follows
\begin{align}\label{eq:A_rel}
    A \leq \frac{2 \mathbb{1}}{d(d+1)}.
\end{align}
Now, from Eq.~\eqref{eq:inv_det_SDP_proof0},~\eqref{eq:inv_det_SDP_proof3} and Eq.~\eqref{eq:perform_inv_proof} we get,
\begin{align}\label{eq:Tr_perform_inv}
\Tr(S\Omega^{\psi})&= \frac{(\Tr(A_{\O}) + \bra{0}A_{\O}\ket{0}_{\O})}{2} + \frac{(\Tr(B_{\O}) - \bra{0}B_{\O}\ket{0}_{\O})}{2}\\
&= \frac{(\Tr(A_{\O}) + \bra{0}A_{\O}\ket{0}_{\O})}{2} +\frac{1}{d} - \frac{d+1}{d-1}\Tr(A_{\O}) +\frac{d+1}{d-1} \frac{\bra{0}A_{\O}\ket{0}_{\O}}{2}
\end{align}
From Thm.~\ref{thm:S_det_covariance_f} we have a following symmetry $[A_{\O}, \overline{V}_{\O}] = 0$ such that $V^{-1}\ketbra{0}{0} = \ketbra{0}{0}$ for all 
$V \in \mathrm{SU}(d)$ (similarly for $B$) which imposes the additional structure on coefficients. In particular, we have
\begin{align}\label{eq:coef_A_proofstep}
    A &= \alpha \ketbra{0}{0} + \beta (\mathbb{1} - \ketbra{0}{0}).
\end{align}
Note that, since $A\geq0$, we must have $\alpha,\beta\geq0$.
Going back to Eq.~\eqref{eq:Tr_perform_inv} we obtain,
\begin{align}
\textup{max}\quad \Tr(S\Omega^{\psi})&= - \beta \frac{ d + 3}{2} +\frac{1}{d} + \alpha
\end{align}
In order to maximise the trace, since $\beta\geq0$, we must set $\beta = 0$. To fix the $\alpha$ term, we use Eq.~\eqref{eq:A_rel},
\begin{align}
    \alpha \ketbra{0}{0}\leq \frac{2 \mathbb{1}_{\I}}{d(d+1)},
\end{align}
from which we conclude that 
\begin{align}\label{eq:alph_proofstep}
    \alpha^{\textup{max}} = \frac{2\ketbra{0}{0}}{d(d+1)}
\end{align}
hence, the maximal value for fidelity is $\expval{F^{\psi}}^{\textup{max}}= \frac{d+3}{d(d+1)}$.

Additionally, in order to conclude the supermap, and hence the circuit structure that attains the optimal performance, we obtain the values of $A$ and $B$. From Eq.~\eqref{eq:alph_proofstep} and~\eqref{eq:coef_A_proofstep} we get $A = \frac{2}{d(d+1)}$. Furthermore, using Eq.~\eqref{eq:coef_AB}, we obtain,
\begin{align}
     B_{\O} &= \frac{1}{d-1} \left ( \frac{2\mathbb{1}_{\I}}{d} - \ketbra{0}{0} )\right) 
\end{align}
The optimal supermap is then given by
\begin{equation}
    S^{\textup{opt}}_{\I\O\F} = \frac{2\ketbra{0}{0 }_{\O}}{d(d+1)}\otimes P^{S}_{\I\F} + \frac{2(\mathbb{1}_{\O} - \ketbra{0}{0 }_{\O})}{d(d-1)}\otimes P^{A}_{\I\F}.
\end{equation}
\end{proof}

\section{Proofs for the probabilistic bipartite task}\label{append:bipart_proofs}

\subsection{Symmetries of the problem}

\SymBipThm*
\begin{proof}

Let us first notice that the first part of the statement, namely, the input state independent covariance follows from the same reasoning as in proof of Thm.~\ref{thm:twirl_prob_constr} and hence, we do not repeat these steps here.

Let us start with transposition function $f(U) = U^{T}$. We are going to construct a covariant operator $S^c$ which respects the commutation relation $[S^c,B_f]=0$ and that attains the same performance of $S$. For that, we define the linear operator 
\begin{equation}
    S^{c}_{\I\O\F} \coloneqq \int_{\mathcal{V}} \mathrm{d}V \,(\mathbb{1}_{\I} \otimes \overline{B}^{\otimes k}_{\O}\otimes \mathbb{1}_{\F}) S_{\I\O\F} (\mathbb{1}_{\I} \otimes B^{T \otimes k}_{\O} \otimes \mathbb{1}_{\F}), \quad \quad \textup{ where } (\mathbb{1}_{\A} \otimes B^T_{\F})\ket{\psi}_{\A\F}=\ket{\psi}_{\A\F}.
\end{equation}
Here $\mathcal{V} \subseteq \mathrm{SU}(d)$ denotes the subset of unitary operators 
$V \in \mathrm{SU}(d)$ satisfying $(\mathbb{1}_{\A} \otimes B^T_{\F})\ket{\psi}_{\A\F}=\ket{\psi}_{\A\F}$ for a given bipartite state $\ket{\psi}_{\A\F} \in \mathbb{C}^{d}$.
Since the elements of this subset themselves 
form a compact group, $\mathcal{V}$ carries its own Haar measure, which we use for integration. Direct calculation shows that $S^c$ respects the commutation relation $[S^c,B_f]=0$.

We now verify that, $\forall U\in\mathrm{SU}(d)$, we have that,
\begin{align}
     S^{c}_{\I\O\F\A} * \ChoiU^{\otimes k} &=\Tr_{\I \O}(S^{c}_{\I\O\F\A} (\dketbra{\overline{U}}{\overline{U}}_{\I\O} \otimes \mathbb{1}_{\F \A}))\\
     &=\Tr_{\I \O}((\mathbb{1}_{\A} \otimes B^{T}_{\O})S_{\I\O\F\A}(\mathbb{1}_{\A} \otimes \overline{B}_{\O}) (\dketbra{\overline{U}}{\overline{U}}_{\I\O} \otimes \mathbb{1}_{\F \A}))\\
     &=\Tr_{\I \O}(S_{\I\O\F\A}(\dketbra{\overline{BU}}{\overline{BU}}_{\I\O} \otimes \mathbb{1}_{\F \A})\\
     &= p ( \mathbb{1}_{\A} \otimes U_{\F}^{T}B^{T}_{\F}) \ketbra{\psi}{\psi}_{\A\F}( \mathbb{1}_{\A} \otimes \overline{BU}_{\F})\\
     &= p (\mathbb{1}_{\A} \otimes U^{T}_{\F})\ketbra{\psi}{\psi}_{\A\F}(\mathbb{1}_{\A} \otimes \overline{U}_{\F}).
\end{align}
This closes the proof for transposition. The same reasoning applies to other functions.
\end{proof}

\BipSymDet*
\begin{proof}
    The proof proceeds analogously to that of Thm.~\ref{thm:S_det_covariance_f}, but now includes auxiliary system, which can be handled straightforwardly owing to the multilinearity of the tensor product, in a manner similar to Thm.~\ref{thm:symbp}.
\end{proof}

\subsection{Proofs for tasks on the known bipartite pure state}
\paragraph{Transposition}

\ThmBipTrans*

\begin{proof}
To prove the theorem, we determine the supermap that, given a single use of the unitary operation $U \in \mathrm{SU}(d)$ and a known bipartite input state $\ketbra{\psi}{\psi}_{\A\F} \in \mathcal{H}_{A} \otimes \mathcal{H}_{\F}$, implements the transposition $(\mathbb{1}_{\A} \otimes U^{T}_{\F})\ketbra{\psi}{\psi}_{\A\F}(\mathbb{1}_{\A} \otimes U^{T}_{\F})^{\dagger}$ with maximal success probability. 
Given a function $f(U) = U^{T}$ acting on single subsystem, Thm.~\ref{thm:symbp} implies that, without loss in
performance, we may restrict our analysis to supermaps whose Choi operators respect the
corresponding commutation relations  \begin{equation}\label{project_basis_trans_bip}
    [S_{A\I\O\F},  \mathbb{1}_{\A} \otimes U_{\I} \otimes \mathbb{1}_{\O} \otimes \overline{U}_{\F}] = 0 
\end{equation}
which ensure that $S$ can be written as
\begin{align}
   S_{\I \O\F\A}=A_{\O \A} \otimes \ketbra{\phi^{+}}{\phi^{+}} _{\I\F} + B_{\O \A}\otimes  \left(\mathbb{1}_{\I\F} -\ketbra{\phi^{+}}{\phi^{+}} _{\I\F}\right)
\end{align}
and leads to the following SDP:
\begin{align}
    \textup{max} \; \; & p\\
    \textup{s.t.} \; \;
    &\left(A_{\O \A} \otimes \ketbra{\phi^{+}}{\phi^{+}} _{\I\F} + B_{\O \A}\otimes  \left(\mathbb{1}_{\I\F} -\ketbra{\phi^{+}}{\phi^{+}} _{\I\F}\right)\right)*\dketbra{\mathbb{1}}{\mathbb{1}}_{\I\O} = p\ketbra{\psi}{\psi}_{\A\F}\label{proof_bip_trans1}\\
    &A_{\O \A} \otimes \ketbra{\phi^{+}}{\phi^{+}} _{\I\F} + B_{\O \A}\otimes  \left(\mathbb{1}_{\I\F} -\ketbra{\phi^{+}}{\phi^{+}} _{\I\F}\right) \leq S_{\textup{ch}} \label{proof_bip_trans2} \\
    &\Tr_{\A\F}(S_{\textup{ch}}) = \frac{\mathbb{1}_{\I} \otimes \mathbb{1}_\O}{d}\\
    &A_{\O \A}, B_{\O \A}\geq 0.
\end{align}
From Eq.~\eqref{proof_bip_trans1} we have,
\begin{align}\label{c1Tb2}
    \frac{A_{\F\A}}{d} + \Tr_{\O}(B_{\O\A})\otimes \mathbb{1}_{\F} - \frac{B_{\F\A}}{d} &= p\ketbra{\psi}{\psi}_{A\F}.
\end{align}
From Eq.~\eqref{proof_bip_trans2} we have,
\begin{align}\label{c2Tb}
   \Tr_{A}(A_{\O \A}) + \Tr_{A}(B_{\O \A}) (d^2 - 1) \leq \mathbb{1}_{\O}.
\end{align}

We now analyse Eq.~\eqref{c1Tb2} on the subspace orthogonal to that spanned by $\ket{\psi}$. To do so, we take the trace of both sides of Eq.~\eqref{c1Tb2} with the operator $\ketbra{\psi_\perp}$, where $\ket{\psi_\perp}$ is an arbitrary vector orthogonal to $\ket{\psi}$, i.e., $\braket{\psi_\perp}{\psi}=0$. This yields
\begin{align}
    \bra{\psi_{\perp}} A_{\F\A}\ket{\psi_{\perp}} + d\bra{\psi_{\perp}} \Tr_{\O}(B_{\O \A})\otimes \mathbb{1}_{\F}\ket{\psi_{\perp}} - \bra{\psi_{\perp}} B_{\F\A}\ket{\psi_{\perp}} &= 0,
\end{align}
from which we obtain the following equation:
\begin{align}
    \bra{\psi_{\perp}} A_{\F\A}\ket{\psi_{\perp}} = \bra{\psi_{\perp}} B_{\F\A}\ket{\psi_{\perp}} - d\bra{\psi_{\perp}} \Tr_{\O}(B_{\O \A})\otimes \mathbb{1}_{\F}\ket{\psi_{\perp}}.
\end{align}
Since $A \geq 0$, we know that the RHS satisfies $\bra{\psi_{\perp}} A_{\F\A}\ket{\psi_{\perp}} \geq 0$. Moreover, using the inequality $\rho_{AB} \leq d, \rho_{A} \otimes \mathbb{1}{B}$ (see Lemma 3 in Ref.~\cite{Bavaresco2019semi} for a proof), we find that the LHS satisfies $\bra{\psi{\perp}} B_{\F\A}\ket{\psi_{\perp}} - d,\bra{\psi_{\perp}} \Tr_{\O}(B_{\O \A})\otimes \mathbb{1}{\F}\ket{\psi{\perp}} \leq 0$. Hence, for any vector $\ket{\psi_\perp}$ orthogonal to $\ket{\psi}$, we have $\bra{\psi_{\perp}} B_{\F\A}\ket{\psi_{\perp}} = 0 = \bra{\psi_{\perp}} A_{\F\A}\ket{\psi_{\perp}}$. This implies that, on the subspace orthogonal to $\ketbra{\psi}$, both operators $A$ and $B$ vanish, and thus we must have $A = \alpha \ketbra{\psi}$ and $B = \beta \ketbra{\psi}$.

With this conclusion, let us go back to Eq.~\eqref{c1Tb2}
\begin{align}
    \alpha \ketbra{\psi}{\psi} + \beta d \Tr_O(\ketbra{\psi})\otimes \mathbb{1} - \beta\ketbra{\psi}{\psi}  = pd\ketbra{\psi}{\psi}.
\end{align}
Note that, since the identity operator is full rank, when $\beta\neq0$, the equality above cannot be satisfied. Hence, we must have that $\beta=0$, and consequently, $B=0$.
This further implies that,
\begin{equation}\label{coeff_A}
     A_{\F\A} = dp\ketbra{\psi}{\psi}_{\F\A}.
\end{equation}
From here and Eq.~\eqref{c2Tb} we have,
\begin{align} 
    dp \Tr_{A}(\ketbra{\psi}{\psi}_{A\O}) &\leq \mathbb{1}_{\O}\\
    p  &\leq \frac{1}{d\norm{\Tr_{A}(\ketbra{\psi}{\psi}_{\O\A})}_{\textup{op}}},
\end{align}
which means that the optimal probability is,
\begin{align}\label{eq:optm_p_biptrans}
 p  &= \frac{1}{d\norm{\Tr_{A}(\ketbra{\psi}{\psi}_{\O\A})}_{\textup{op}}},
\end{align}
Using Eq.~\eqref{coeff_A} and Eq.~\eqref{eq:optm_p_biptrans}, we conclude that the corresponding supermap is
\begin{equation}
     S_{\I\O\F\A} = \frac{\ketbra{\psi}{\psi}_{\O\A} \otimes \ketbra{\phi^{+}}{\phi^{+}} _{\I\F}}{\norm{\Tr_{A}(\ketbra{\psi}{\psi}_{\O\A})}_{\textup{op}}}.
\end{equation}

\end{proof}

Similarly as before, we can find the realisation in terms of the encoder and the decoder:
\begin{align}
    S_{\I\O\F\A} = \ketbra{\phi^{+}}{\phi^{+}}_{\I\M} * \frac{\ketbra{\psi}{\psi}_{\O\A} \otimes \dketbra{\mathbb{1}}{\mathbb{1}}_{\M\F} }{\norm{\Tr_{A}(\ketbra{\psi}{\psi}_{\O\A})}_{\textup{op}}}. 
\end{align}
Using the definition of a filter as $(\mathbb{1}\otimes F)\ket{\phi^+} = \frac{\ket{\psi}}{\sqrt{d\norm{\sigma}_{\textup{op}}}}$, we can rewrite the supermap as
\begin{align}
    S_{\I\O\F\A} = \ketbra{\phi^{+}}{\phi^{+}}_{\I\M} * \left( \dketbra{F}{F}_{\O\A} \otimes \dketbra{\mathbb{1}}{\mathbb{1}}_{\M\F} \right).
\end{align}
This realisation is depicted in Fig.~\ref{fig:bipartTrans}. 

\paragraph{Conjugation}

\ConjgBipThm*

\begin{proof}   

When \(d = 2\), unitary complex conjugation can be implemented deterministically and exactly even for a completely unknown input state~\cite{Miyazaki_2017} (see Fig.~\ref{fig:conjgKNpure}). Consequently, for \(d = 2\), we have \(p^{\psi_{AB}}_{\textup{max,conj}} = 1\) for any quantum state. The nontrivial part of this theorem is therefore to show that for \(d > 2\), \(p^{\psi_{AB}}_{\textup{max,conj}} = 0\).

As shown in Thm.~\ref{thm:symbp}, without any loss in performance, we may assume that the Choi operator of the superinstrument element $S$ is covariant, that is, it satisfies the commutation relations $[S_{\I\O\F\A}, \id_\I \otimes U_\O \otimes U_\F \otimes \id_\A] = 0$ for every $U \in \mathrm{SU}(d)$. By Schur’s lemma, this implies the decomposition $S = A_{\I\A} \otimes P^{S}_{\O\F} + B_{\I\A} \otimes P^{A}_{\O\F}$, where $P^{S}$ and $P^{A}$ denote the projectors onto the symmetric and antisymmetric subspaces, respectively, and $A$ and $B$ are arbitrary operators. Therefore the problem of maximizing the success probability $p^{\psi_{AB}}_{\textup{max,conj}}$ can be expressed as
\begin{align}
    &\textup{max} \; \; p\\
    \textup{s.t.} \; \; &(A_{\I\A}  \otimes P^{S}_{\O\F} + B_{\I\A} \otimes P^{A}_{\O\F})*\dketbra{\mathbb{1}}{\mathbb{1}}_{\I\O} = p \ketbra{\psi}{\psi}_{\A\F} \label{eq:conj_bip_aux}\\
    &A_{\I\A}, B_{\I\A} \geq 0 \\
    &(A_{\I\A}  \otimes P^{S}_{\O\F} + B_{\I\A} \otimes P^{A}_{\O\F}) \leq S_{\textup{ch}}\\
    &\Tr_{\F\A}(S_{\textup{ch}}) = \sigma_{\I} \otimes \mathbb{1}_{\O} \; \;\Tr(\sigma_{\I}) = 1,
\end{align}
where the constraints, $A,B\geq0$ follows from the fact that $S\geq0$ and the projectors $P^{S}$ and $P^{A}$ have orthogonal support. From Eq.~\eqref{eq:conj_bip_aux} we obtain,
\begin{align}\label{eq:link_bip_conjg}
   A_{\O\A}  * P^{S}_{\O\F} + B_{\O\A} * P^{A}_{\O\F} = p\ketbra{\psi}{\psi}_{\A\F}.
\end{align}
Now, we will proceed by showing that when $d > 2$ and both $A \neq 0$ and $B \neq 0$, the linear operator $A_{\O\A} * P^{S}_{\O\F} + B_{\O\A} * P^{A}_{\O\F}$ has rank strictly greater than one. Therefore, the only possible solution to the optimization problem above is to set $A = B = 0$ and $p = 0$.

Let us start by referring to the particular instance where the linear space $\mathcal{H}_\A$ is one dimensional, or, equivalently, that there is no auxiliary space $\mathcal{H}_\A$ and the input state $\ket{\psi}$ is a single party state in $\mathcal{H}_\F$, as presented in Thm.~\ref{single_conjg}. Let us recall that in that case we had
\begin{align}\label{eq:link_conj_bip}
    A_{\O}  * P^{S}_{\O\F} =  \frac{\id_\F \Tr(A) + A^T_\F}{2}, \quad \quad B_{\O}  * P^{A}_{\O\F} =  \frac{\id_\F \Tr(B) - B^T_\F}{2}.
\end{align}
Since the identity operator is full rank, when $A\neq0$, and $d\geq2$, the linear operator  $\frac{\id_\F + A^T_\F}{2}$ has rank strictly greater than one. Analogously, when $B\neq0$, and $d\geq3$, the linear operator  $\frac{\id_\F + A^T_\F}{2}$ has rank strictly greater than one. Hence, when $\mathcal{H}_\A$ is one dimensional and $d>2$, we must have $A=B=0$ and $p=0$.

We now extend this argument to the case where $\mathcal{H}_\A$ has arbitrary dimension $d$. First, note that since $A, B \geq 0$, the operators $A$ and $B$ are proportional to quantum states. Moreover, since $\Tr_\F(P^{S}_{\O\F})$ and $\Tr_\F(P^{A}_{\O\F})$ are proportional to the identity operator, and $P^{S}_{\O\F}, P^{A}_{\O\F} \geq 0$, the projectors $P^{S}_{\O\F}$ and $P^{A}_{\O\F}$ are proportional to the Choi operators of quantum channels. Hence, up to a scalar, the link products $A_{\O\A} * P^{S}_{\O\F}$ and $B_{\O\A} * P^{A}_{\O\F}$ represent a quantum channel acting on part of a bipartite state. To see this clearly, we can expand the link products in Eq.~\eqref{eq:link_conj_bip} using Lemma~\ref{lemma:identities} and conclude
\begin{align}\label{eq:link_cbip_A}
    A_{\O\A} * P^{S}_{\O\F} = \frac{\Tr_{\O}(A_{\O\A}) \otimes \mathbb{1}_{\F} + A^{T_{\F}}_{\A \F}}{2} \sim \textup{Choi}(\map{C}_{\O \rightarrow \F} \otimes \map{\id}_\A(\rho_{\O\A}))
\end{align}
and
\begin{align}\label{eq:link_cbip_B}
    B_{\O\A} * P^{A}_{\O\F} = \frac{\Tr_{\O}(A_{\O\A}) \otimes \mathbb{1}_{\F} - A^{T_{\F}}_{\A \F}}{2} \sim \textup{Choi}(\map{C}_{\O \rightarrow \F} \otimes \map{\id}_\A(\rho_{\O\A})).
\end{align}
Now, by invoking the result from Section~3, ``Tensoring with an ideal channel,'' of Ref.~\cite{Amosov2000additivity}, we can conclude that the rank of Eqs.~\eqref{eq:link_cbip_A} and~\eqref{eq:link_cbip_B} is strictly greater than one for $d > 2$. The cited result establishes that if $\map{C} : \mathcal{L}(\mathcal{H}_1) \to \mathcal{L}(\mathcal{H}_2)$ is a quantum channel such that $\rank(\map{C}(\rho)) > 1$ for any quantum state $\rho \in \mathcal{L}(\mathcal{H}_1)$, then $\rank(\map{C} \otimes \map{\id}_2(\rho_{12})) > 1$ for any $\rho_{12} \in \mathcal{L}(\mathcal{H}_1 \otimes \mathcal{H}_2)$. Since, for $d > 2$, the channels represented in Eqs.~\eqref{eq:link_cbip_A} and~\eqref{eq:link_cbip_B} satisfy $\rank(\map{C}_{2} \otimes \map{\id}_2(\rho_{12})) > 1$, we conclude that the LHS of Eq.~\eqref{eq:link_bip_conjg} has rank greater than one, while the RHS has rank one, hence $A = B = 0$ and $p=0$.

\end{proof}

\paragraph{Inversion}

\InvBipProb*

\begin{proof}
When $d = 2$, the task of unitary inversion is equivalent to unitary transposition, due to the identity $U^{-1} = \sigma_Y U^{T} \sigma_Y$, which holds for all $U \in \mathrm{SU}(2)$.

When $d > 2$, the proof follows the same steps as Thm.~\ref{thm:conjgBip}. Without any loss in performance, we impose that the superinstrument element $S$ is covariant, and can be decomposed as $S = A_{\O\A} \otimes P^{S}{\I\F} + B{\O\A} \otimes P^{A}_{\I\F}$. Using a covariant $S$, the maximization problem reads as
    \begin{align}
    &\textup{max} \; \; p\\
    \textup{s.t.} \; \; &(A_{\O\A}  \otimes P^{S}_{\I\F} + B_{\O\A} \otimes P^{A}_{\I\F})*\dketbra{\mathbb{1}}{\mathbb{1}}_{\I\O} = p \ketbra{\psi}{\psi}_{\A\F} \label{eq:inv_bip_aux}\\
    &A_{\O\A}, B_{\O\A} \geq 0 \\
    &(A_{\O\A}  \otimes P^{S}_{\I\F} + B_{\O\A} \otimes P^{A}_{\I\F}) \leq S_{\textup{ch}}\\
    &\Tr_{\F\A}(S_{\textup{ch}}) = \sigma_1 \otimes \mathbb{1}_{\O} \; \;\Tr(\sigma_1) = 1.
\end{align}
From Eq.~\eqref{eq:conj_bip_aux}, we obtain
\begin{align}
    A_{\O\A} * P^{S}_{\O\F} + B_{\O\A} * P^{A}_{\O\F} = p \ketbra{\psi}{\psi}_{\A\F}.
\end{align}
This is exactly the same equation as Eq.~\eqref{eq:link_bip_conjg} appearing in the previous proof for conjugation (Thm.~\ref{thm:conjgBip}). Hence, by following the same arguments, we conclude that $A = B = 0$ and $p = 0$.

\end{proof}

\section{Proofs for mixed state case}

\LemmaMixVis*

\begin{proof} 
From the form of performance operator $\Omega_{\I\O\F} = \rho_{\P}* \Omega_{\P\I\O\F}$, we have that supermap in the deterministic case admits the analogous form $S_{\I\O\F} = \rho_{\P}* S_{\P\I\O\F}$. From Theorem 2 in \cite{Quintino_2022det}, we have that for a supermap with a universal input implementing a homomorphic or antihomomorphic function, the following holds:
\begin{equation}
    S^{\prime}_{\P\I\O\F} *\ChoiU_{\I \O} ^{\otimes k}=\eta \dketbra{f(U)}{f(U)}_{\P \F}+(1-\eta) \mathbb{1}_\P \otimes \frac{\mathbb{1}_\F}{d},
\end{equation}
where $\expval{F^{\psi}}=\eta+\frac{1-\eta}{d^2}$. We can simply modify now this statement by linking with the known state,
\begin{align}
    \rho_{\P}*S^{\prime}_{\P\I\O\F} *\ChoiU_{\I\O} ^{\otimes k}&=\eta  \rho_{\P}* \dketbra{f(U)}{f(U)}_{P F}+(1-\eta)  \Tr(\rho_{\P})\frac{\mathbb{1}_{\F}}{d},
    \end{align}
from which we obtain
    \begin{align}
    S^{\prime}_{\I\O\F} *\dketbra{U}{U}_{\I\O} ^{\otimes k}&=\eta \rho_{\P} * \dketbra{f(U)}{f(U)}_{P F}+(1-\eta) \frac{\mathbb{1}_{\F}}{d}. 
\end{align}
The fidelity is given in terms of the visibility as
\begin{align}
\expval{F^{\psi}} & =\int_{\text {Haar }} \Tr((S_{\I \O \F} *\ChoiU_{\I \O}^{\otimes k})( \rho_{\P}*\dketbra{f(U)}{f(U)}_{\P\F}))\mathrm{d} U \\
& =\int \Tr\left(\eta \rho_{\P} * \dketbra{f(U)}{f(U)}_{P F}+(1-\eta) \frac{\rho_{\P}*\dketbra{f(U)}{f(U)}_{\P\F}}{d} \right) \\
& =\eta+\frac{1-\eta}{d}.
\end{align}
From here we get
\begin{align}
  \eta &= \frac{d \expval{F^{\psi}}-1}{d-1},
\end{align}
where $\expval{F^{\psi}}$ is optimal average fidelity for the pure states.

\end{proof}

\ThmMixVisib*

\begin{proof}
We start with the equation
\begin{equation}
    \eta = \frac{d \expval{F^{\psi}}-1}{d-1}
\end{equation}
and substitute
\begin{itemize}
    \item $\expval{F^{\psi}}_{\textup{trans}, k=1} = \frac{2d + 1}{d^2 + d} \quad \textup{from which follows} \quad \eta = \frac{d}{d^2 - 1}$
    \item $ \expval{F^{\psi}}_{\textup{conj}, k = 1}  =
\begin{cases}
       1,& \textup{if } d = 2\\
    \frac{2}{d + 1},              &  \textup{if } d >  2 
\end{cases}\quad \textup{from which follows} \quad \eta = \begin{cases}
       1,& \textup{if } d = 2\\
    \frac{1}{d + 1}              &  \textup{if } d >  2 
\end{cases} $
\item $\expval{F^{\psi}}_{\textup{inv,} k = 1}  = \frac{2}{d(d-1)} \; \textup{from which follows} \; \eta = \frac{2}{d^2 - 1}$
\end{itemize}

\end{proof}

\section{Proofs for SAR }\label{sec:proof_SAR}
\ThmSARsingle*
\begin{proof}

We now show that the constructions for the storage-and-retrieval (SAR) task and for transposition are equivalent and yield the same performance. The following argument applies equally to both the probabilistic exact and deterministic approximate realizations.

Let us start by showing that the transposition protocol provides a construction for the SAR that preserves the performance of the task. As stated in Sec.~\ref{symmetriesK}, the transposition protocol is characterized by the symmetry condition $\bigl[S_{\I\O\F},\, U^{\otimes k}_{\I} \otimes \mathbb{1}^{\otimes k}_{\O} \otimes \overline{U}_{\F}\bigr] = 0,
\; \forall\,U\in SU(d)$. Let us recall that the supermap can be expressed as $S_{\I\O\F} = E_{\M\I}*D_{\M\O/\F}$, where $\M$ is a memory subspace representing the wire between the encoder and the decoder. From this, it follows that the encoder satisfies
\begin{align}
\bigl[\Tr_{\M}(E_{\I\M}),\, U^{\otimes k}_{\I}\bigr] = 0,
\qquad \forall\,U\in SU(d).
\end{align}
Any operator $\sigma_{\I} \coloneqq \Tr_{\M}(E_{\I\M})$ that commutes with all $U^{\otimes k}$ must, by Schur-Weyl duality, be block diagonal in the joint representation space
\begin{align}
    (\mathbb{C}^d)^{\otimes k}
\;\cong\;
\bigoplus_{\mu}
\mathcal{H}_{\mu}\otimes \mathbb{C}^{m_\mu},
\end{align}
where $\mathcal{H}_{\mu}$ carries the irreducible representation $U(\mu)$ of $SU(d)$, and $\mathbb{C}^{m_\mu}$ is the corresponding multiplicity space.  
Thus, we have
\begin{align}\label{eq_sigm}
\sigma_{\I}
=
\bigoplus_{\mu}
\frac{\mathbb{1}_{\mu}}{\dim(\mu)}\otimes X_\mu,
\qquad X_\mu\geq0.
\end{align}
To find the encoder $\Tr_{\M}(E_{\I\M}) = \sigma_{\I}$, we need to find the purification of Eq.~\eqref{eq_sigm}. A canonical purification of such a symmetric operator $\sigma_{\I}$ is unique up to an isometry on $\M$ and is precisely the \emph{port-based teleportation (PBT)} resource state
\begin{align}
\ket{\phi_{\mathrm{PBT}}}
\coloneqq
\bigoplus_{\mu}
\sqrt{p_\mu}\,
\ket{\phi^+(\mu)} \otimes \ket{\psi_{m_\mu}},
\end{align}
where
\begin{align}
\ket{\phi^+(\mu)}
\coloneqq
\frac{1}{\sqrt{\dim(\mu)}}
\sum_{i=1}^{\dim(\mu)}
\ket{i_\mu i_\mu}
\;\in\;
\mathcal{H}_{\mu}\otimes\mathcal{H}_{\mu},
\end{align}
$\{p_\mu\}$ is a probability distribution, and $\ket{\psi_{m_\mu}}$ is a normalized state in $\mathcal{H}_{\mu} \otimes \mathcal{H}_{\mu}$, i.e., the maximally entangled state within each irrep block.
Hence, we conclude that the encoder is $E_{\I\M}=\ketbra{\phi_{\mathrm{PBT}}}{\phi_{\mathrm{PBT}}}_{\I\M}$. 

Now, let us recall that for the maximally entangled state we have the following property: $(\mathbb{1} \otimes U)\ket{\phi^{+}} =(U^T \otimes \mathbb{1})\ket{\phi^{+}}$. In the case of a product representation, a similar expression holds for the PBT state. Namely,
\begin{align}\label{PBT_cov}
    (\mathbb{1}\otimes U^{\otimes k})\ket{\phi_{\mathrm{PBT}}}
=
((U^{T})^{\otimes k}\otimes \mathbb{1})\ket{\phi_{\mathrm{PBT}}}.
\end{align}
This identity holds because any tensor product representation
of the form $U^{\otimes k}$ can be decomposed, by Schur-Weyl duality, as
\begin{align}
    U^{\otimes k}
\;\cong\;
\bigoplus_{\mu\in \operatorname{irrep}(U^{\otimes k})}
U(\mu)\otimes \mathbb{1}_{m_\mu},
\end{align}
for some unitaries $U(\mu)$ acting on the irreducible representation spaces
$\mathcal{H}_{\mu}$ and multiplicity spaces of dimension $m_\mu$. Since the PBT resource state is built as a direct sum of maximally entangled states $\ket{\phi^+(\mu)}$ over all irreducible components, it inherits the same
covariance property in each irrep sector. Hence, the identity~\eqref{PBT_cov} follows directly from the property of maximally entangled states and the Schur--Weyl decomposition. Furthermore, from identity~\eqref{PBT_cov} we see immediately that the transposition protocol provides a realization of the SAR protocol with identical performance, for both deterministic approximate and probabilistic exact realizations.

Let us now show that, given an SAR protocol, we can immediately construct a transposition protocol with identical performance. For an SAR protocol, to ensure state-independent performance over all pure input states $\ket{\psi}$, we must impose the covariance condition
\begin{align}
\bigl[\Tr_{\M}(E_{\I\M}),\, U^{\otimes k}_{\I}\bigr] = 0,
\qquad \forall\,U\in SU(d).
\end{align}
This symmetry once again enforces the port-based teleportation form of the encoding state, $E_{\I\M}=\ketbra{\phi_{\mathrm{PBT}}}{\phi_{\mathrm{PBT}}}_{\I\M}$. Using identity~\eqref{PBT_cov}, we again conclude that the SAR construction provides an immediate construction for the transposition protocol with identical performance.
Both tasks therefore admit equivalent implementations and achieve the same performance, regardless of probabilistic or deterministic figures of merit.
\end{proof}

\section{Proofs for unitary transposition, conjugation, and inversion with multiple calls}\label{proofsk}

As mentioned in the main text, in Sec.~\ref{sec:application_k>1}, the proof for transposition is contained in the proof of Thm.~\ref{thm:SAR_transp_Eq}. Let us thus proceed to proving proofs for conjugation and inversion.

\ThmConjgk*

\begin{proof}
Without loss of generality, we can assume that the input state is given by
\begin{align}
    \label{eq:input_state}
\ket{\psi}_{\P} = \ket{+_d} = \frac{1}{\sqrt{d}} \sum_{s=0}^{d-1} \ket{s}
\end{align}
using the computational basis $\{\ket{s}\}_{j=0}^{d-1}$.
We show this no-go theorem holds even if we restrict the input unitary to be the diagonal unitary given by
\begin{align}
    \label{eq:diagonal_unitary}
    U= \sum_{s=0}^{d-1} e^{i\phi_{s}}\ketbra{s}{s}
\end{align}
using an arbitrary phases $\phi_s \in [0, 2\pi)$.
The defining equation of this protocol can be written as
\begin{equation}
    \smap{S}^{\psi}(\map{U}^{\otimes k} ) = p\overline{U}\ketbra{\psi}\overline{U}^\dagger.
\end{equation}
Since any general supermap can be simulated by a parallel supermap with non-zero probability~\cite{Quintino2019prob},
we can assume that the supermap is parallel:
    \begin{equation}
    \smap{S}^{\psi}(\map{U}^{\otimes k} ) = \map{D}^{\psi} \circ [\map{U}^{\otimes k} \otimes \map{\id}] \circ E^{\psi},
\end{equation}
where $E^{\psi}$ is a quantum state corresponding to the encoder and $\map{D}^{\psi}$ is a quantum channel corresponding to the decoder.
Without loss of generality, we can assume that the encoder $\map{E}^{\psi}$ is given by a pure state $\ket{E^\psi} \in \mathcal{H}_{\M} \otimes \mathcal{H}_{\I}$ using an auxiliary space $\M$ and the decoder is an isometry $D^\psi: \mathcal{H}_{\O} \otimes \mathcal{H}_{\M} \to \mathcal{H}_{\M'} \otimes \mathcal{H}_{\F}$ followed by a projective measurement $\{\ketbra{m}\}_m$ on $\mathcal{H}_{\M'}$.
By postselecting the measurement outcome $m$ on $\M'$, we have
\begin{equation}\label{Eq:def_conjg_k}
(\bra{m}_{\M^{\prime}} \otimes \mathbb{1}_{\O})D^{\psi}_{\O \M \rightarrow M^{\prime} \F} \left [ U^{\otimes k}_{\I \rightarrow \O} \otimes \mathbb{1}_{\M} \right ]\ket{E^{\psi}}_{\M \I} = e^{i\theta_{U, \psi, m}}\sqrt{p(U,m)} \overline{U}\ket{\psi},
\end{equation}
where $e^{i\theta_{U, \psi, m}}$ is a phase factor that may depend on $U$, $\psi$ and $m$, and $p(U, m)$ is the success probability that may depend on $U$ and $m$ such that $p = \sum_m p(U, m)$. Let us note that even though the probability of full protocol was assumed to be independent of $U$, a success superinstrument can be decomposed into  can always decompose it into rank-one superinstrument elements $S_1, \ldots, S_n$ such that $S = S_1 + \cdots + S_n$. Hence, the defining equation Eq.~\eqref{Eq:def_conjg_k} on the pure state is a postselection to on one of these rank-one superinstrument elements and probability of that postselection can, in general, depend on $U$ and $m$, hence $p(U,m)$. 
We decompose the isometry $D^\psi$ in the computational basis as
\begin{equation}
    \label{eq:D_psi_m}
    ( \bra{m}_{\M^{\prime}} \otimes \mathbb{1}_{\O})D^{\psi}_{\O \M \rightarrow M^{\prime} \F} = \sum_{\vec{i}jl} D^{m, \psi}_{\vec{i}jl} \ket{l}_{\F}\bra{\smash{\vec{i}}j}_{\O \M},
\end{equation}
where the vector $\vec{i}$ corresponds to a tuple $\vec{i} = [ i_{1}, ..., i_{k}]$ and $D^{m, \psi}_{\vec{i}jl}\in \mathbb{C}$. Similarly, we also decompose the encoder state in the computational basis as
\begin{align}
    \label{eq:E_psi}
        \ket{E^{\psi}}_{\M \I} &= \sum_{m^{\prime}\vec{n}} E^{\psi}_{m^{\prime} \vec{n}}\smash{\ket{m^{\prime} \vec{n}}}_{\M \I} 
\end{align}
using $E^{\psi}_{m^{\prime} \vec{n}}\in \mathbb{C}$.
By substituting Eqs.~\eqref{eq:input_state}, \eqref{eq:diagonal_unitary}, \eqref{eq:D_psi_m}, and \eqref{eq:E_psi} into Eq.~\eqref{Eq:def_conjg_k}, we have
\begin{align}
    \sum_{\vec{s}jl} D^{m, \psi}_{\vec{s}jl} E^\psi_{\vec{s} j} e^{i\sum_{a=1}^{k} \phi_{s_a}} \ket{l} = e^{i\theta_{U, \psi, m}}\sqrt{\frac{p(U, m)}{d}} \sum_{l} e^{-i\phi_l} \ket{l},
\end{align}
By projecting on some $l\in \{1, \ldots, d\}$ we have
\begin{align}
    \label{eq:conjg_final}
    e^{i\theta_{U, \psi, m}}\sqrt{\frac{p(U, m)}{d}} = \sum_{\vec{s}j} D^{m, \psi}_{\vec{s}jl} E^\psi_{\vec{s} j} e^{i\sum_{a=1}^{k} \phi_{s_a} + i\phi_l}.
\end{align}
Defining $\vec{\gamma}_{\vec{s}}$ by the multiplicity of each value in $\vec{s}$, i.e., $\vec{\gamma}_{\vec{s}} = (\gamma_0, \ldots, \gamma_{d-1})$ for $\gamma_i$ being the number of $a\in \{1, \ldots, k\}$ such that $s_a = i$, we can rewrite $\sum_{a=1}^{k} \phi_{s_a}$ as $\vec{\gamma}_{\vec{s}} \cdot \vec{\phi}$, where $\vec{\phi} = (\phi_0, \ldots, \phi_{d-1})$.
We rewrite Eq.~\eqref{eq:conjg_final} as a linear combination of exponential functions $e^{i \vec{\gamma} \cdot \vec{\phi}}$ for $\vec{\gamma}\in \mathbb{Z}^d$:
\begin{align}
    \label{eq:conjg_final_2}
    e^{i\theta_{U, \psi, m}}\sqrt{ \frac{p(U, m)}{d}} = \sum_{\vec{\gamma}} A^{m, \psi}_{\vec{\gamma}, l} e^{i\vec{\gamma}\cdot\vec{\phi}},
\end{align}
where $A^{m, \psi}_{\vec{\gamma}, l}$ is defined by
\begin{align}
    A^{m, \psi}_{\vec{\gamma}, l} \coloneqq \sum_{\vec{s}: \vec{\gamma}_{\vec{s}} = \vec{\gamma}-\vec{e}_l} \sum_j D^{m, \psi}_{\vec{s}jl} E^\psi_{\vec{s} j},
\end{align} 
and $\vec{e}_l$ is a $d$-dimensional vector whose $l$-th entry is $1$ and the others are $0$.
Since $\sum_{i=0}^{d-1} \gamma_{\vec{s}, i} = k$ and $\gamma_{\vec{s},i} \geq 0$ holds for all $i\in \{0, \ldots, d-1\}$, $\vec{\gamma}$ appearing in the summation in Eq.~\eqref{eq:conjg_final_2} satisfies $\sum_{i=0}^{d-1} \gamma_i = k+1 <d$, $\gamma_l\geq 1$, and $\gamma_i \geq 0$ for all $i\in \{0, \ldots, d-1\}$.
Due to the pigeonhole principle, for any $\vec{\gamma}$, there exists $l^*_{\vec{\gamma}}$ such that the $l^*_{\vec{\gamma}}$-th element of $\vec{\gamma}$ is $0$.
Since $l$ in Eq.~\eqref{eq:conjg_final_2} is arbitrary, we can choose $l = l^*_{\vec{\gamma}}$.
Then, we obtain
\begin{align}
    \sum_{\vec{\gamma}} A^{m, \psi}_{\vec{\gamma}, l} e^{i\vec{\gamma}\cdot\vec{\phi}} = \sum_{\vec{\gamma}} A^{m, \psi}_{\vec{\gamma}, l^*_{\vec{\gamma}}} e^{i\vec{\gamma}\cdot\vec{\phi}}.
\end{align}
Let us now notice that $\vec{\gamma}$ appearing in the sum of the right-hand side satisfies $\gamma_{l^*_{\vec{\gamma}}} \geq 1$, due to the linear independence of the exponential functions $e^{i\vec{\gamma}\cdot \vec{\phi}}$ for $\vec{\gamma}\in \mathbb{Z}^d$. Hence, we obtain
\begin{align}
    A^{m, \psi}_{\vec{\gamma}, l} = 0.
\end{align}
Since this holds for all $\vec{\gamma}$, $l$ and $m$, from Eq.~\eqref{eq:conjg_final_2}, we obtain $p(U, m) = 0$ for all diagonal unitary $U$ and $m$. Since $p = \sum_{m}p(U, m)$, and $p(U, m) = 0$ for every $m$, we conclude that $p= 0$.
\end{proof}

\ThmInvgk*

\begin{proof}
   The proof follows directly from the proof for conjugation. In the previous argument for conjugation, we specifically used unitaries that are diagonal in the computational basis, while assuming that the known input state is not an eigenstate of that basis. For diagonal unitaries, however, transposition acts trivially, i.e. $U^T = U.$ Therefore, inversion coincides with conjugation on this subgroup, since $U^{-1} = U^\dagger = (U^T)^\dagger = (U^\dagger)^T. $ Thus, for diagonal unitaries, inversion and conjugation are equivalent transformations, and the same sequence of steps used in the conjugation proof applies.
\end{proof}

\addcontentsline{toc}{section}{References}

\nocite{apsrev42Control} 


%

\end{document}